\documentclass[english, hidelinks, 11pt]{article}
\usepackage{geometry}
\usepackage{setspace}
\usepackage{bm}
\usepackage{graphicx}
\graphicspath{ {.} }
\usepackage{hyperref}
\usepackage{xr}

\usepackage{amssymb, amsmath, amsthm, amsfonts, mathtools, bbm}
\usepackage{float}
\usepackage{color, xcolor}
\usepackage{url}
\usepackage{enumitem}
\usepackage{csquotes}
\usepackage{multirow}
\usepackage{siunitx}
\usepackage{placeins}
\usepackage{authblk}
\usepackage{appendix}
\usepackage{lscape}
\usepackage{booktabs}
\usepackage{longtable}
\usepackage{stackrel}

\allowdisplaybreaks[1]

\usepackage{tikz}
\usetikzlibrary{arrows.meta, decorations.pathreplacing}

\usepackage{algorithm}
\usepackage{algpseudocode}

\usepackage{xr-hyper}
\usepackage{hyperref}
\hypersetup{
    colorlinks=true,
    linkcolor=blue,
    citecolor=blue,
    filecolor=magenta,
    urlcolor=blue
}
\usepackage{cleveref}

\usepackage[round]{natbib}

\newtheorem{theorem}{Theorem}
\newtheorem{lemma}{Lemma}

\newtheorem{remark}{Remark}
\newtheorem{assumption}{Assumption}

\newtheorem{example}{Example}

\numberwithin{equation}{section}

\usepackage{soul}       
\usepackage{ulem}   
\DeclareMathOperator*{\argmax}{arg\,max}

\def\tr{\textnormal{\text{tr}}}
\def\wt#1{\widetilde{#1}}
\def\wh#1{\widehat{#1}}

\def\E{\mathbb{E}} 
\def\L{\mathbb{L}} 
\def\K{\mathbb{K}}

\def\P{\mathbb{P}}
\def\R{\mathbb{R}}

\def\Z{\mathbb{Z}} 

\def\cF{\mathcal{F}}
\def\cG{\mathcal{G}} 
\def\cH{\mathcal{H}}
\def\cI{\mathcal{I}} 
\def\cII{\mathcal{II}} 
\def\cIII{\mathcal{III}} 
\def\cIV{\mathcal{IV}}

\def\cM{\mathcal{M}} 
\def\cN{\mathcal{N}} 
\def\cO{\mathcal{O}}

\def\cS{\mathcal{S}}
\def\cT{\mathcal{T}}

\def\zero{\mathbf{0}} 
 
\def\bI{\mathbf{I}} 
\def\bd{\mathbf{d}} 
\def\be{\mathbf{e}} 
\def\bu{\mathbf{u}} 
\def\bj{\mathbf{j}} 
\def\bx{\mathbf{x}} 
\def\by{\mathbf{y}}
\def\bX{\mathbf{X}} 
\def\bY{\mathbf{Y}} 
\def\bE{\mathbf{E}} 
\def\bU{\mathbf{U}} 
\def\bL{\mathbf{L}} 
\def\bF{\mathbf{F}} 
\def\bA{\mathbf{A}} 
\def\bB{\mathbf{B}} 
\def\bC{\mathbf{C}} 
\def\bD{\mathbf{D}} 
\def\bG{\mathbf{G}} 
\def\bH{\mathbf{H}}
\def\bJ{\mathbf{J}}
\def\bK{\mathbf{K}}
\def\bM{\mathbf{M}} 
 
\def\bR{\mathbf{R}} 
 
\def\bQ{\mathbf{Q}} 
 
\def\bV{\mathbf{V}} 
\def\bW{\mathbf{W}}
\def\bv{\mathbf{v}} 
\def\bg{\mathbf{g}} 
\def\bf{\mathbf{f}}

\def\bmu{\boldsymbol{\mu}} \def\bLambda{\boldsymbol{\Lambda}} \def\bGamma{\boldsymbol{\Gamma}} \def\bgamma{\boldsymbol{\gamma}} \def\bSigma{\boldsymbol{\Sigma}} 
\def\bXi{\boldsymbol{\Xi}}
 
\def\balpha{\boldsymbol{\alpha}}
\def\bbeta{\boldsymbol{\beta}}

\def\1{\mathbf{1}}
\def\0{\mathbf{0}}

\def\trans{\intercal}

\def\var{\mathbb{V}\mathrm{ar}}

\def\vec{\mathrm{vec}} 
\def\Vech{\mathrm{Vech}}

\usepackage[makeroom]{cancel}

\begin{document}
\def\spacingset#1{\renewcommand{\baselinestretch}%
{#1}\small\normalsize} \spacingset{1}

	\begin{titlepage}
		
		\title{Multiple Change Point Detection in the Mean and Covariance of High-Dimensional Matrix Time Series}
                \author{Zetai Cen\thanks{School of Mathematics, University of Bristol.
                Email: \url{zetai.cen@bristol.ac.uk}.
                Supported by Engineering and Physical Sciences Research Council (EP/Z531327/1).}\! } 
                \author{Haeran Cho\thanks{School of Mathematics, University of Bristol.
                Email: \url{haeran.cho@bristol.ac.uk}.
                Supported by Engineering and Physical Sciences Research Council (EP/Z531327/1).}}
		
		\affil{School of Mathematics, University of Bristol}
		
		\date{}
		
		\maketitle

\begin{abstract}
This paper studies simultaneous change point detection in the first- and second-order structures of high-dimensional matrix-valued time series. Specifically, we extend the framework of main effect factor models to accommodate multiple change points so that shifts in the mean and covariance structures are characterized by changes in the trend-stationary main effects main effects and in the factor-driven common components, respectively. Under this framework, the two types of changes are not necessarily aligned and can be arbitrarily large without masking the effects of each other, and each of the detected change points is identified either with the row or column categorizations which improves their interpretability. We propose a unified procedure based on moving-sum statistics to detect and localize both types of changes. Under general regularity conditions that allow for temporal and cross-sectional dependence, we derive detection and localization guarantees for the proposed method. The finite-sample performance of the method is demonstrated through extensive experiments, and its practical usefulness is illustrated by a real data application to New York City Yellow Taxi trip records.
\end{abstract}
		
		\bigskip
		\bigskip

		\noindent
		{\sl Keywords and phrases:}
        Matrix time series;
        Simultaneous change point detection;
        Temporal and cross-sectional dependence.

\noindent

	\end{titlepage}
	
	\setcounter{page}{2}

\maketitle

\tableofcontents

\newpage

\section{Introduction}\label{sec: introduction}

Matrix-valued, high-dimensional time series data are widely observed where it is of interest to learn the row- and column-wise structures as well as interactions therein, with examples ranging from cross-country macroeconomic panels \citep{Wangetal2019}, to international trade flow \citep{chen_chen_2022} and the records of air pollutants across multiple cities \citep{Yuanetal2023}.
One of the key characteristics of modern datasets is the heterogeneity and, when observations are gathered in a temporal order within non-stationary environments, detecting and locating abrupt shifts, a.k.a.\ change points, in stochastic properties enhances flexibility and interpretability of statistical models.
In this paper, we address the problem of detecting (possibly) multiple change points in matrix-valued time series, where the changes may be present in either first- or second-order structures of the data. 

While the literature on change point detection is vast, joint detection of changes in the first and second moments has received relatively less attention, and little work exists in the context of matrix time series modeling.
In univariate settings, it has been studied using both parametric \citep{csorgo1997} and nonparametric \citep{messer2022bivariate} approaches.
There also exist studies on detecting changes in the mean while treating the time-varying second-order structure as a nuisance parameter \citep{pein2017heterogeneous, mcgonigle2021detecting}, and vice versa \citep{dette2019change}.
In high-dimensional panel data, most existing works focus on the detection of changes in the mean \citep{cho2015multiple, cho2016change, pilliat2023optimal} or covariance \citep{avanesov2018change, WaYuRi21}, while assuming that other characteristics of the data are time-invariant. 
One exception is \cite{cui2025simultaneous} where a joint detector statistic for the different types of change is proposed to test for the presence of any change as well as estimating its location.
Beyond the vector time series settings, few papers address the problems of detecting changes in the mean \citep{pevsta2025tensor} or the second moments \citep{ZhangCenCho2026} in matrix or higher-order tensor time series. 

We propose to perform the segmentation (a.k.a.\ multiple change point detection) of matrix-valued time series by adopting a versatile framework in matrix factor models.
Our setup enjoys flexibility, embedding multiple shifts in the first and second moments of the matrix time series in low-dimensional structures, without requiring them to be aligned. 
Further, by associating each detected change with either or both of the row and column categorizations, it lends interpretability to the change points. 

Since the first work by \cite{Wangetal2019}, matrix-valued factor models have drawn increasing attention over the last decade \citep{Yuetal2022, ChenFan2023, Heetal2023, Heetal2024}. The main ingredient is the decomposition of each matrix observation $\bX_t$ as $\bX_t = \bA \bF_t \bB^\trans + \bE_t$, with $\bF_t$ being the low-rank factor matrix, $\bA$ and $\bB$ the row and column loading matrices, and $\bE_t$ the noise.
While vectorizing $\bX_t$ would also admit a form of factor model \citep{BaiNg2002, PanYao2008, Lametal2011}, such an approach may suffer from ignoring the inherent matrix structure 
and lose the interpretation of the row and column categorizations, which has motivated approaches for preserving and benefiting from the matrix structures.
There also exists a line of works that focus on capturing separable row and column effects \citep{GaoTsay2021, Heetal2023, Yuanetal2023}.

More recently, \cite{LamCen2025_MEFM} introduced a main effect matrix factor model (MEFM). 
Under MEFM, each observed matrix $\bX_t \in \R^{p \times q}$ admits the following decomposition
\[
\bX_t = \mu_t\1_p \1_q^\trans + \balpha_t \1_q^\trans + \1_p \bbeta_t^\trans + \bA \bF_t \bB^\trans + \bE_t,
\]
where $\mu_t \in \R$ is the grand mean
which can be either deterministic or stochastic, and
$\balpha_t \in \R^p$ and $\bbeta_t \in \R^q$ are the row and column main effects, respectively, with $\1_m \in \R^m$ denoting the vector of ones. 
Compared to the existing matrix time series factor models, which typically assume that the mean of $\bX_t$ is constant, MEFM is equipped to model (potentially) time-varying mean through $\mu_t$, $\balpha_t$ and $\bbeta_t$, while imposing an interpretable low-rank structure therein. 

We propose to generalize the MEFM framework by allowing for changes in both the mean and the covariance of $\bX_t$, which are characterized as abrupt jumps in otherwise smoothly varying means of $\balpha_t$ and $\bbeta_t$, and changes in the loading matrices of the factor-driven component,respectively. 
Under a condition for the identification of the latent components, our proposed data segmentation method detects and locates the changes in the first- and second-order structures separately. 
In addition to achieving asymptotic consistency in multiple change point detection, it identifies whether a detected change is associated with the row or column main effects, or the factor-driven component, thereby lending more interpretability to the estimators. 

Methodologically, our proposal bears similarities to the procedures proposed by \cite{Chenetal2022} for detecting changes in the mean of vector time series, and by \cite{Zhangetal2025} for segmenting tensor-valued time series under a factor model. 
At the same time, the former requires that the covariance of the vector time series remains constant, while the latter supposes the mean of the tensor time series is time-invariant.
In contrast, our method handles the changes in both the first and second moments and moreover, explicitly attributes the changes to the row or column categorizations. 
The simultaneous detector of \cite{cui2025simultaneous}, while being able to jointly test for changes in the mean and covariance of vector time series, potentially suffers from the complex interaction between different types of changes whereby a large shift in the mean can obscure a change in the covariance.
The MEFM framework, as well as lending interpretability to the detected changes, allows the separation of the estimation and detection of the first- and second-order structures such that our procedure avoids complications arising from such interference between mean and covariance changes. 


The rest of this paper is organized as follows.
Section~\ref{sec: model_est} introduces the MEFM with change points and discusses the identification conditions.
Section~\ref{sec: method} develops the preliminary estimators of the latent structures under MEFM, and the detection procedures for changes in the first- and second-order structures along with guidance on the selection of tuning parameters. 
Section~\ref{sec: theorem} presents the asymptotic consistency of the proposed change point detection methods. 
We report in Sections~\ref{sec: numerical_results} and~\ref{subsec: real_data_taxi} the results of simulation studies and a real data application to New York City Yellow Taxi trip records. 
All proofs and additional numerical results are relegated to the supplementary material.
\bigskip

\textbf{Notations.} 
Unless specified otherwise, we denote scalars, vectors, and matrices by letters, lower-case bold letters, and capital bold letters, i.e.,\ $x$, $\bx$ and $\bX$, respectively. We also use $x_i, X_{ij}, X_{i\cdot}, X_{\cdot i}$ to denote, respectively, the $i$th element of $\bx$, the $(i,j)$th element of $\bX$, the $i$th row vector (as a column vector) and the $i$th column vector of $\bX$. We use $\otimes$ to represent the Kronecker product, and $\circ$ the Hadamard product. By convention, the total Kronecker product for an index set is computed in descending index. We use $a\asymp b$ to denote $a=O(b)$ and $b=O(a)$. Hereafter, given a positive integer $m$, define $[m]=\{1,2,\dots,m\}$ and $\bM_m:= \bI_m - m^{-1} \1_m \1_m^\trans$,
where $\bI_m \in \R^{m\times m}$ denotes the identity matrix.
The $i$th largest eigenvalue of a matrix $\bX$ is denoted by $\lambda_i(\bX)$.
The vectorization of a matrix $\bX$ is denoted as $\vec(\bX)$.
For a symmetric matrix $\bX \in\R^{a\times a}$, the half-vectorization $\Vech(\bX): \R^a \mapsto\R^{a(a+1)/2}$ is defined as the vector by stacking the elements on and below the diagonal of $\bX$.
The first-order derivative of a function $f$ is denoted by $f'$. The Frobenius norm, spectral norm, and (element-wise) max norm of a matrix $\bX$ (or a vector, without loss of generality) is denoted by $\|\bX\|_F$, $\|\bX\|$, and $\|\bX\|_\infty$, respectively.


\section{Model}
\label{sec: model_est}

Consider a series of matrix observations $\bX_t\in \R^{p\times q}$ for $t\in[T]$, where each $\bX_t$ admits a main effect matrix factor model (MEFM) representation as
\begin{align}
    \bX_t &= 
    \mu_t\1_p \1_q^\trans + \balpha_t \1_q^\trans + \1_p \bbeta_t^\trans + \bC_t + \bE_t ,
    \label{eqn: MEFM_structure} 
\end{align}
where $\mu_t$ is the (possibly stochastic) grand mean, $\balpha_t$ and $\bbeta_t$ are vector processes representing the row and column main effects, respectively, $\bC_t$ is the factor-driven common component whose structure is further specified later, and $\bE_t$ is the idiosyncratic noise. 
We allow for both the mean and second moments of $\{\bX_t\}_{t = 1}^T$ to undergo multiple change points that are not necessarily aligned, and the aim of this paper lies in consistent detection of these change points.
We may also permit change points in the stochastic properties of $\mu_t$ such as its mean, for which a plethora of data segmentation methods exist;
see the review in \cite{cho2024data}.
Since under an identifiability condition (Assumption~\ref{ass: identification}), the stochastic properties of $\mu_t$ do not affect the remaining components in~\eqref{eqn: MEFM_structure} and the change points therein, we do not impose any structural constraint on $\mu_t$.

The changes in the mean of $\{\bX_t\}_{t = 1}^T$ enter the model as changes in the row and column effects, which are trend stationary with multiple mean shifts, i.e.,\ 
\begin{equation}
\label{eqn: mean_change}
\begin{split}
    \balpha_t &= \bmu_{\alpha,t} + \sum_{i=1}^{Q_\alpha} \bgamma_{\alpha,i} \cdot \mathbb{I}\{t \geq \theta_{\alpha,i} \} +\be_{\alpha,t} , \\
    \bbeta_t &= \bmu_{\beta,t} + \sum_{j=1}^{Q_\beta} \bgamma_{\beta,j} \cdot \mathbb{I}\{t \geq \theta_{\beta,j} \} +\be_{\beta,t} ,
\end{split}
\end{equation}
where $\bmu_{\alpha,t}:= \bmu_{\alpha}(t/T)$ with $\bmu_{\alpha}(\cdot):[0,1] \mapsto \R^p$ denotes smoothly varying trend function, and $\bgamma_{\alpha,i} \in\R^p$ the jump vector at the $i$th change point with the collection of change points denoted by $\Theta_\alpha = \{\theta_{\alpha, i}, \, i \in [Q_\alpha] : \, 0 < \theta_{\alpha, 1} < \ldots < \theta_{\alpha, Q_\alpha} < T \}$; by convention, we write $\theta_{\alpha, 0} = 0$ and $\theta_{\alpha, Q_\alpha + 1} = T$.
For the column main effect, $\bmu_{\beta,t}$, $\bgamma_{\beta,j}$, and $\Theta_\beta$ are analogously defined. 
Both $\balpha_t$ and $\bbeta_t$ are stochastic with the inclusion of the mean-zero, stationary processes $\{\be_{\alpha,t}\}_{t = 1}^T$ and $\{\be_{\beta,t}\}_{t = 1}^T$.

Changes in the second-order structure of $\{\bX_t\}_{t = 1}^T$ are introduced to the model~\eqref{eqn: MEFM_structure}  as changes in the factor structure of $\{\bC_t\}_{t = 1}^T$, where
\begin{equation}
\label{eqn: cov_change}
    \bC_t = \sum_{w=0}^{Q_f} \bA_{w} \bF_t \bB_{w}^\trans \cdot \mathbb{I}\{\theta_{f,w} <t\le \theta_{f,w + 1}\}.
\end{equation}
Here, $\bF_t \in \R^{k_r \times k_c}$ is the core factor matrix with $k_r$ and $k_c$ as the factor numbers which are fixed for all $p$ and $q$, $\bA_w \in \R^{p\times k_r}$ and $\bB_w\in \R^{q\times k_c}$ are the row and column loading matrices on the $w$th segment, respectively, with the collection of the change points denoted by $\Theta_f=\{\theta_{f,w}, \, w \in [Q_f]: \, \theta_{f,0} \equiv 0 < \theta_{f, 1} < \ldots < \theta_{f, Q_f} < T \equiv \theta_{Q_f+1}\}$. 
We remark that neither $\bA_w$ nor $\bB_w$ are required to be of full column rank, which allows for changes in the row or column factor numbers with $k_r$ (resp.\ $k_c$) imposing an upper bound on the segment-wise row (resp.\ column) factor number.

In fact, one may re-write the model in~\eqref{eqn: cov_change} by identifying global factor loading matrices $\bQ \in\R^{p\times k_r}$ and $\bU \in\R^{q\times k_c}$ with full column ranks, and the transformation matrices $\bR_w \in\R^{k_r\times k_r}$ and $\bV_w \in\R^{k_c\times k_c}$, $w=0,\dots, Q_f$, not necessarily of full rank, which satisfy $\bA_w = \bQ \bR_w$ and $\bB_w = \bU \bV_w$ such that
\begin{align}
    \bC_t 
    &= \bQ \Bigg( \sum_{w=0}^{Q_f} \bR_{w} \bF_t \bV_{w}^\trans \cdot \mathbb{I}\{\theta_{f,w} <t\le \theta_{f,w + 1}\} \Bigg) \bU^\trans
    =: \bQ \bG_t \bU^\trans.
    \label{eqn: Gt_define}
\end{align}
Referred to as the pseudo-factor, the matrix time series $\{\bG_t\}_{t = 1}^T$ with $\bG_t \in \R^{k_r \times k_c}$, absorbs the changes in the loading matrices as the changes in its second moments; see \cite{ZhangCenCho2026} for further discussions on the equivalence between~\eqref{eqn: cov_change} and~\eqref{eqn: Gt_define} in a tensor setting.
Then, denoting the segment-wise covariance of the vectorized pseudo-factor by
\[
\bGamma_{G,w} := \E\Big[ \vec(\bG_t) \vec(\bG_t)^\trans \cdot\mathbb{I}{\{\theta_{f,w} < t \le\theta_{f,w+1}\}} \Big]
\in \R^{k_rk_c \times k_rk_c}
\]
for each $w=0,\dots, Q_f$, we measure the size of the change at $\theta_{f,w}$ with
\begin{equation}
\label{eqn: def_jump_factor}
\omega_{f,w} := \| \Delta_{G,w} \| \text{ \ where \ }
\Delta_{G,w} := \bGamma_{G,w} - \bGamma_{G,w-1}.
\end{equation}

The MEFM model with the changes allowed as in~\eqref{eqn: mean_change} and~\eqref{eqn: cov_change} enjoys considerable flexibility.
In the absence of any change point (i.e.,\ $Q_\alpha = Q_\beta = Q_f = 0$), it embeds the classical matrix factor model while allowing for time-varying mean, see Remark~1 of \cite{LamCen2025_MEFM}. 
The model serves as a tool for separately handling the mean and second moment change point problems.
This enables us to avoid the interference between the changes in different structures which is particularly beneficial in high dimensions, whereas existing methods for joint detection of mean and covariance change tend to suffer, e.g., from large changes in the mean potentially masking any change in the covariance, see Theorem~3 of \cite{cui2025simultaneous}.
Also, any change detected from the main effects is identifiable with either the row or column categorization which, complementing the mode-identification procedure of \cite{ZhangCenCho2026} for the changes detected from the factor-driven component, enhances the interpretability of the changes. 
We provide additional discussions on the model in Section~\ref{app: model}.

We formally present the condition for the identifiability of the latent components under the model~\eqref{eqn: MEFM_structure}.
\begin{assumption}[Identification]
\label{ass: identification}
For any $t\in[T]$, $i\in[Q_\alpha]$, $j\in[Q_\beta]$, $w=0,\dots, Q_f$, we have
\[
\1_p^\trans \bmu_{\alpha,t} =0, \quad
\1_q^\trans \bmu_{\beta,t} =0, \quad
\1_p^\trans \bgamma_{\alpha,i} =0, \quad
\1_q^\trans \bgamma_{\beta,j} =0, \quad
\1_p^\trans \bA_w =\0, \quad
\1_q^\trans \bB_w =\0 .
\]
\end{assumption}
The conditions on the loading matrices $\bA_w$ and $\bB_w$ are comparable to Assumption~(IC1) in \cite{LamCen2025_MEFM}, while those on the trend functions and jump vectors can be seen as requiring $\1_p^\trans \E(\balpha_t)=0$ and $\1_q^\trans \E(\bbeta_t)=0$.
By Theorem~1 of \cite{LamCen2025_MEFM}, the (time-varying) means of $\mu_t$, $\balpha_t$ and $\bbeta_t$ are identifiable, and so are $\mu_t$, $\balpha_t$ and $\bbeta_t$ from the common component, under Assumption~\ref{ass: identification}.
Separation between the trend function and jump vectors relies on the smoothness of the trend and the jump size, and we defer their characterization to Section~\ref{sec: theorem}.

\section{Methodology}\label{sec: method}

\subsection{Preliminary estimation}
\label{subsec: estimation}

We present the estimators of the latent components in~\eqref{eqn: MEFM_structure} which serve as the main building blocks for detecting change points under~\eqref{eqn: mean_change} and~\eqref{eqn: cov_change}.

Notice that for each $t\in[T]$, left-multiplying $\1_p^\trans$ and right-multiply $\1_q$ to $\bX_t$ gives
\begin{align*}
    \1_p^\trans \bX_t \1_q &= pq\mu_t + q\1_p^\trans \balpha_t + p\1_q^\trans \bbeta_t + \1_p^\trans \bE_t \1_q 
    = pq\mu_t + q\1_p^\trans \be_{\alpha,t} + p\1_q^\trans \be_{\beta,t} + \1_p^\trans \bE_t \1_q.
\end{align*}
by \Cref{ass: identification}.
This suggests an estimator for the grand mean as
\begin{equation}
\label{eqn: mu_hat}
\wh\mu_t := \frac{1}{pq} \1_p^\trans \bX_t \1_q .
\end{equation}
Similarly, for each $t\in[T]$, by right-multiplying $\1_q$ (resp.\ $\1_p$) on $\bX_t$ (resp.\ $\bX_t^\trans$), we arrive at $\bX_t \1_q = q \mu_t \1_p + q \balpha_t + \bE_t \1_q$ (resp.\ $\bX_t^\trans \1_p = p \mu_t \1_q + p \bbeta_t + \bE_t^\trans \1_p$), which motivates us to approximate the main effect processes by
\begin{align}
    \wh\balpha_t &:= \frac{1}{q} \bX_t \1_q - \wh\mu_t \1_p \text{ \ and \ }
    \wh\bbeta_t := \frac{1}{p} \bX_t^\trans \1_p - \wh\mu_t \1_q.
    \label{eqn: mu_alpha_hat} 
\end{align}
These estimators form the basis for the method for detecting change points in the mean of $\bX_t$ under~\eqref{eqn: mean_change}, see \Cref{subsec: detection_ME}.
With the above estimators, we define
\begin{equation}
\label{eqn: Lt_hat}
\wh\bL_t := \bX_t - \wh\mu_t\1_p \1_q^\trans - \wh\balpha_t \1_q^\trans - \1_p \wh\bbeta_t^\trans .
\end{equation}

Next, recall that under the re-formulation in~\eqref{eqn: Gt_define}, any change in the factor loading matrices are detectable as the changes in the second moments of the pseudo factor time series $\{\bG_t\}_{t = 1}^n$ of the reduced dimensions. 
Hence, we propose to detect the change points under~\eqref{eqn: cov_change} from the estimated $\{\bG_t\}_{t = 1}^n$, for which we adapt a commonly adopted two-step procedure in the matrix factor modeling literature \citep{Yuetal2022, Heetal2024}, with a key difference that we apply the procedure to $\{\wh\bL_t\}_{t = 1}^T$ rather than the observed time series. 
For this, we first obtain the initial estimators of the global loading matrix $\bQ$, say $\wt\bQ$, as $\sqrt{p}$ times the leading $k_r$ eigenvectors of $(Tpq)^{-1} \sum_{t=1}^T \wh\bL_t \wh\bL_t^\trans$; we similarly estimate $\bU$ by $\wt\bU$ containing the leading $k_c$  eigenvectors of $(Tpq)^{-1} \sum_{t=1}^T \wh\bL_t^\trans \wh\bL_t$ multiplied by $\sqrt q$. 
Then, the final estimator $\wh\bQ$ is given by $\sqrt{p}$ times the leading $k_r$ eigenvectors of $(Tpq^2)^{-1} \sum_{t=1}^T \wh\bL_t \wt\bU \wt\bU^\trans \wh\bL_t^\trans$, where the projection onto the column space of $\wt\bU$ enhances the signal; we analogously obtain $\wh\bU$. 
Lastly, the pseudo-factor time series is estimated by
\begin{equation}
\label{eqn: Gt_hat}
\wh\bG_t := \frac{1}{pq} \wh\bQ^\trans \wh\bL_t \wh\bU, \quad t \in [T].
\end{equation}

\subsection{Generic algorithm}
\label{sec:alg}

In Sections~\ref{subsec: detection_ME} and~\ref{subsec: detection_fac}, we present detector statistics for detecting multiple change points in the main effects and the factor structure, respectively. 
They are constructed on the principle of locally contrasting the stochastic property of interest over the adjacent intervals $[\tau - \gamma, \tau]$ and $(\tau, \tau + \gamma]$ for some bandwidth $\gamma$ and varying $\tau \in \{\gamma + 1, \ldots, T - \gamma\}$, whereby large values of the detector statistics indicate the presence of significant changes in the neighborhood.

\begin{algorithm}[h!t!]
\caption{Generic MOSUM algorithm for change point detection under~\eqref{eqn: MEFM_structure}.}\label{alg: detect}
\begin{algorithmic}[1]
\setstretch{1.25}
\State \textbf{Input:} Scan statistic $\cT_\tau^*$ with bandwidth $\gamma$, threshold $\pi^*$
\State \textbf{Initialize:} Set $\wh\Theta \gets \emptyset$ and $t \gets \gamma + 1$
\While{$t \leq T- \gamma$}
\If{$t =\argmax_{\tau \in[t- \gamma, t+ \gamma]} \cT^*_{\tau}$} 
\If{$\cT^*_{t} > \pi^*$}
\State $\wh\Theta \gets \wh\Theta \cup \{t\}$
\EndIf
\State $t \gets t + \gamma +1$
\Else
\State $t \gets t +1$
\EndIf
\EndWhile
\State \textbf{Output:} Set of change point estimators $\wh\Theta$
\end{algorithmic}
\end{algorithm}

We present a generic algorithm which performs multiple change point detection in combination in combination with any such detector statistic, which is denoted by $\cT^*_\tau$, see Algorithm~\ref{alg: detect}.
It iteratively identifies significant local maximizers of $\cT^*_\tau$, in the sense that (i)~any $t$ added to the set of change point estimators $\wh\Theta$ maximizes $\cT^*_\tau$ over the interval $[t - \gamma, t + \gamma]$, and (ii)~$\cT^*_t$ exceeds some threshold $\pi^* > 0$. 
If such $t$ is identified, the updating rule in line~8 ensures that the same change point is not detected again and, provided that the next change point is sufficiently distanced away from the one previously detected by~$t$, the algorithm continues to search for the next change point until it reaches the end of the data sequence.
It is possible for Algorithm~\ref{alg: detect} to return $\wh\Theta=\emptyset$ which is equivalent to $\max_{\gamma < \tau \leq T-\gamma} \cT^*_{\tau} \leq \pi^*$, suggesting that there is little evidence for large shifts in the stochastic property of interest.
We discuss the choice of problem-specific tuning parameters in Section~\ref{sec: theorem}. 

\subsection{Change point detection in the first-order structure}
\label{subsec: detection_ME}

We describe how to detect the multiple change points in the main effects, which drive the shifts in the mean of $\bX_t$ under~\eqref{eqn: mean_change}. It suffices to focus on the row main effects, since the method for the column main effects is constructed analogously.

Following \cite{Chenetal2022}, who address the problem of detecting abrupt changes in (otherwise) smoothly varying mean of high-dimensional time series, we characterize the changes in $\{\balpha_t\}_{t = 1}^T$ by the sequence of vectors  
$\bj_{\alpha}(\tau) := \balpha_\tau^{+} - \balpha_\tau^{-} := \lim_{t \,\downarrow\, \tau} \balpha_t - \lim_{t \,\uparrow\, \tau} \balpha_t$,
$\tau\in[T]$.
Under certain smoothness conditions, we have $\bj_{\alpha}(\theta_{\alpha, i}) = \bgamma_{\alpha,i}$ for all $i \in [Q_\alpha]$, while $\bj_{\alpha}(\tau) = \mathbf{0}$ for $\tau \in [T] \setminus \Theta_\alpha$. 
Based on these observations, we construct the detector statistic by estimating $\bj_{\alpha}(\tau)$ via local linear estimators of $\balpha_\tau^{+}$ and $\balpha_\tau^{-}$, constructed with $\{\wh{\balpha}_t\}_{t = 1}^T$ in~\eqref{eqn: mu_alpha_hat}.
Let $b\in (0,1)$ be some bandwidth parameter
and suppose that $\gamma := bT$ is an integer for simplicity.
Then we obtain
\begin{equation}
\label{eqn: def_hat_alpha_pm}
\wh\balpha_{\tau}^{+} := \sum_{t= \tau+1}^{\tau+\gamma} w_{t-\tau} \wh\balpha_t ,\quad
\wh\balpha_{\tau}^{-} := \sum_{t= \tau-\gamma}^{\tau-1} w_{\tau-t} \wh\balpha_t, \text{ \ and \ }
\wh\bj_{\alpha}(\tau) := \wh\balpha_{\tau}^{+} - \wh\balpha_{\tau}^{-}
\end{equation}
where, for some kernel function $K(\cdot)$, the weights $w_t$ are defined as
\begin{align}
\label{eq:kernel}
    w_{t} &:= \frac{K(t/\gamma) (S_2 +S_1 t/T)}{S_2 S_0 - S_1^2} \text{ \ with \ } 
    S_\ell := \sum_{i=1}^T \Big(-\frac{i}{T} \Big)^{\ell} K\Big(\frac{i}{\gamma}\Big) ,\; \ell=0,1,2.
\end{align}
Note that $\wh\bj_{\alpha}(\tau)$ in \eqref{eqn: def_hat_alpha_pm} involves symmetric construction of $\wh\balpha_{\tau}^{+}$ and $\wh\balpha_{\tau}^{-}$, to facilitate technical arguments.
We expect $\wh\bj_{\alpha}(\tau)$ to approximate 
$\sum_{t= \tau +1}^{\tau+\gamma} w_{t-\tau} \bmu_{\alpha}(t/T) - \sum_{t= \tau-\gamma}^{\tau-1} w_{\tau-t} \bmu_{\alpha}(t/T)$ which, for $\tau$ satisfying $\Theta_\alpha \cap [\tau-\gamma, \tau+\gamma] =\emptyset$, is made sufficiently small in any vector norm, for sufficiently smooth $\bm\mu_\alpha(\cdot)$ and the kernel; 
see Assumptions~\ref{ass: trend_fun} and \ref{ass: kernel} below.
On the other hand, if 
$\min\{ \theta_{\alpha, i} - \theta_{\alpha, i - 1}, \theta_{\alpha, i + 1} - \theta_{\alpha, i} \} \ge \gamma$ for some $i \in [Q_\alpha]$, 
we expect any norm of $\wh\bj_{\alpha}(\tau)$ to take its local maximum near $\theta_{\alpha,i}$ with $\wh\bj_{\alpha}(\theta_{\alpha, i}) \approx \bm\gamma_{\alpha, i}$.
Combining these observations with that $\1_p^\trans \bm\gamma_{\alpha, i} = 0$ for all $i$ under Assumption~\ref{ass: identification}, we propose to utilize $\bM_p \wh\bj_{\alpha}(\tau)$ in constructing the detector statistic with $\bM_p = \bI_p - p^{-1} \1_p\1_p^\trans$, which further helps to mitigate the noise. 

To this end, we propose the following detector statistics to scan for the mean shifts in the row and column main effects as
\begin{equation}
\label{eqn: mosum_main_effect}
\cT^{(\alpha)}_{b, \tau} := \left\| 
\bW_{\alpha, \tau}^{-1/2} 
\bM_p\, \wh\bj_{\alpha}(\tau) \right\|_\infty 
\text{ \ and \ }
\cT^{(\beta)}_{b, \tau} := \left\| 
\bW_{\beta, \tau}^{-1/2} 
\bM_q\, \wh\bj_{\beta}(\tau) \right\|_\infty ,
\end{equation}
respectively,
where $\wh\bj_{\beta}(\tau)$ is defined similarly as \eqref{eqn: def_hat_alpha_pm},
and $\bW_{\alpha, \tau}$ and $\bW_{\beta, \tau}$ are (possibly) location-dependent
diagonal matrices with bounded entries, accounting for the heterogeneity across the row and column cross-sections. 
For detecting changes in the row (resp.\ column) main effect, Algorithm~\ref{alg: detect} takes $\cT^{(\alpha)}_{b, \tau}$ (resp.\ $\cT^{(\beta)}_{b, \tau}$) as its detector statistic with the bandwidth $\gamma$.

\begin{remark}
\label{rem:sparse}
As seen later in Theorem~\ref{thm: consistency_change_ME}, the choice of the element-wise max norm for measuring the size of $\wh\bj_{\alpha}(\tau)$ after projection and standardization, implies that the detector statistic $\cT^{(\alpha)}_{b, \tau}$ has better power against the situations where the change vectors $\bgamma_{\alpha, i}$ are sparse, i.e.,\ they contain prominent jumps at a few coordinates. 
There exists change point detection procedures which are shown to be adaptively optimal with respect to the unknown sparsity of the change in high dimensions \citep{pilliat2023optimal}.
However, such studies focus on the cross-sectionally independent situations and it is largely unknown how the interplay between the change and the cross-sectional dependence influences the phase transition with respect to the sparsity with few exceptions \citep{cho2025covariance}. 
We adopt the element-wise max norm which facilitates uniform control of the partial sums of the errors arising from estimating the latent $\balpha_t$ by $\wh{\balpha}_t$, while fully permitting temporal and cross-sectional dependence. 
In Section~\ref{subsec: real_data_taxi}, we observe that this element-wise max norm-based approach is successful in detecting large shifts in the taxi usage patterns associated with the beginning and end of daylight saving time, which are noticeable at some hours but not so prominent at others (Figure~\ref{fig: NYC_hour_plot_sel2}).
\end{remark}


\subsection{Change point detection in the second-order structure}
\label{subsec: detection_fac}

We consider the detection of multiple change points in the factor structure, which drive the changes in the second moments of $\{\bX_t\}_{t = 1}^T$ under~\eqref{eqn: cov_change}. 
To motivate the proposed detector statistic, recall the definition of the jump magnitude at $\theta_{f, w}$ in~\eqref{eqn: def_jump_factor}.
Then, with $\wh\bG_t$ in \eqref{eqn: Gt_hat} estimating the pseudo factor $\bG_t$, and some bandwidth $\gamma\geq 1$, 
we scan the following series of moving sum (MOSUM) statistics
\begin{align}
\label{eqn: test_stat_aggreg}
    \cT^{(f)}_{\gamma, \tau} &:= \frac{1}{\sqrt{2\gamma}} \left\Vert
    \bW_{G, \tau}^{-1/2} \Vech\big( \wh{\bGamma}_{G ,\gamma ,\tau} \big) \right\Vert  ,      \quad \gamma \leq \tau \leq T-\gamma,
\text{ \ where} 
\\
    \wh{\bGamma}_{G ,\gamma ,\tau} &:= \sum_{t=\tau+1}^{\tau+\gamma} \vec(\wh\bG_t) \vec(\wh\bG_t)^\trans - \sum_{t=\tau-\gamma+1}^{\tau} \vec(\wh\bG_t) \vec(\wh\bG_t)^\trans,
    \notag
\end{align}
and $\bW_{G, \tau}$ is a location-dependent, positive definite matrix, which is introduced to account for scale-heterogeneity in the elements of $\wh{\bGamma}_{G ,\gamma ,\tau}$.
Unlike in~\eqref{eqn: def_hat_alpha_pm}, we adopt the above MOSUM statistic which amounts to setting $w_t$ to be piecewise constant. 
This is justifiable as under~\eqref{eqn: cov_change}, the second-order structure of $\{\bC_t\}_{t = 1}^T$ and thus $\{\bG_t\}_{t = 1}^T$ is exactly piecewise constant.
Then, considering the population version of $\wh{\bGamma}_{G ,\gamma ,\tau}$ denoted by ${\bGamma}_{G ,\gamma ,\tau}$, it behaves as
\begin{align*}
\left\Vert
\Vech\big( {\bGamma}_{G ,\gamma ,\tau} \big) \right\Vert
&=
\Bigg\Vert
\Vech\Bigg\{ \sum_{t=\tau+1}^{\tau+\gamma} \E[\vec(\bG_t) \vec(\bG_t)^\trans] - \sum_{t=\tau-\gamma+1}^{\tau} \E[\vec(\bG_t) \vec(\bG_t)^\trans] \Bigg\}  \Bigg\Vert \\
&=
\left\{
    \begin{array}{ll}
	(\gamma - \vert \tau - \theta_{f,w} \vert ) \, \omega_{f,w}  & \hbox{if $\vert \tau - \theta_{f, w} \vert \le \gamma$ for some $w \in [Q_f]$,} \\
    0 & \text{otherwise.}
    \end{array}
\right. 
\end{align*}
This suggests that $\cT^{(f)}_{\gamma, \tau}$ is close to zero if $\Theta_f \cap (\tau-\gamma, \tau+\gamma] =\emptyset$,
while it takes its local maximum near $\tau = \theta_{f,w}$ if the interval $(\tau-\gamma, \tau + \gamma]$ contains exactly one change point from $\Theta_f$.
We combine the detector statistic $\cT^{(f)}_{\gamma, \tau}$ with Algorithm~\ref{alg: detect}, for detecting the change points in $\Theta_f$ under~\eqref{eqn: cov_change}.

\begin{remark}
\label{rem:zcc}
While the detector statistic $\cT^{(f)}_{\gamma, \tau}$ is effectively formed by contrasting local behavior of vector processes of dimension $k_rk_c(k_rk_c + 1)/2$, \cite{Zhangetal2025} propose a detector statistic which contrasts the core factor covariance along each mode separately, and thus involves vector processes of dimension $k_r(k_r + 1)/2 + k_c(k_c + 1)/2$ in the current context. 
The latter approach was shown to enjoy numerical benefits from the reduced dimension for higher-order tensor time series, but it also requires a stronger identification condition on the global factor loading matrices, namely that $p^{-1} \bQ^{\trans} \bQ = \bI_p$ and $q^{-1} \bU^{\trans} \bU = \bI_q$, for handling the bilinear forms of the errors from scanning the data for changes. 
In contrast, thanks to the vectorization of the pseudo-factor in constructing $\cT^{(f)}_{\gamma, \tau}$, we make a considerably weaker condition which is comparable with those found in the classic factor model literature \citep{BaiNg2002}, see Assumption~\ref{ass: loadings}~(ii) for further detail.
\end{remark}



\section{Theoretical properties}\label{sec: theorem}

\subsection{Consistency of the preliminary loading estimators}\label{subsec: thm_prelim}

We first define a class of linear processes which plays a fundamental role in our theoretical investigation.
Given a matrix time series $\{\bY_t\}_{t \in \Z}$ with $\bY_t \in \R^{m \times n}$,
we denote by $\{\bY_t\}_{t \in \Z} \in \cG(\nu, \eta)$ for some $\nu \geq 4$ and $\eta\geq 0$, if and only if it holds that $\bY_t = \sum_{h\geq 0} a_{y,h} \bX_{y,t-h}$ where 
(i)~the innovation process $\{\bX_{y, t}\}_{t \in \Z}$ with $\bX_{y,t} \in \R^{m \times n}$, has independent elements 
both temporally and cross-sectionally, 
with zero mean, unit variance, and uniformly bounded $\nu$th order moments;
(ii)~the scalar coefficients $a_{y,h}$ satisfy $\sum_{h\geq 0} a_{y,h}^2=1$ and $\sum_{h\geq 0} |a_{y,h}| \leq c$ for some constant $c > 0$; 
and (iii)~there exists some constant $c_\eta > 0$ depending on $\eta$, such that $\sup_{m \ge 0} (m + 1)^\eta \sum_{s = m}^\infty \vert a_{y, s} \vert \le c_{\eta}$.

To establish the consistency of the preliminary estimators of the main effects and pseudo-factor introduced in Section~\ref{subsec: estimation}, we introduce the following assumptions on model~\eqref{eqn: MEFM_structure}. 

\begin{assumption}[Global factor loadings]\label{ass: loadings}
There exists some constant $c$ such that: 
\begin{itemize}
    \item [(i)] Both $\bQ$ and $\bU$ are deterministic with $\|\bQ\|_{\max}, \|\bU\|_{\max} \leq c <\infty$;
    \item [(ii)] 
    As $p,q \to\infty$, we have $p^{-1} \bQ^\trans \bQ \to \bI_{k_r}$ and $q^{-1} \bU^\trans \bU \to \bI_{k_c}$.
\end{itemize}
\end{assumption}

\begin{assumption}[Core factors]\label{ass: core_factor}
There exist some constants $\nu_f \ge 4$ and $\eta_f \ge 0$ such that the factor time series satisfies $\{\bF_t\}_{t=1}^T \in \cG(\nu_f, \eta_f)$.
\end{assumption}


\begin{assumption}[Idiosyncratic noise]\label{ass: noise}
For any $t\in[T]$, it holds that $\bE_t = \bA_{e,r} \bF_{e,t} \bA_{e,c}^\trans + \bSigma_{e} \circ \bE_{e,t}$,
where there exist some constants $\nu_e \ge 4$ and $\eta_e \ge 0$ such that $\{\bF_{e,t}\}_{t=1}^T, \, \{\bE_{e,t}\}_{t=1}^T \in \cG(\nu_e, \eta_e)$.
Moreover, $\{\bF_{e,t}\}_{t = 1}^T$ and $\{\bE_{e,t}\}_{t = 1}^T$ are independent of $\{\bF_t\}_{t = 1}^T$, as well as being independent of each other. 
The deterministic matrices $\bA_{e,r}$ and $\bA_{e,c}$ consist of uniformly bounded elements and satisfy $\|\bA_{e,r}\|_1, \|\bA_{e,c}\|_1=\cO(1)$,
and $\bSigma_{e}$ contains 
the standard deviations of the corresponding elements in $\bE_{e,t}$, with elements uniformly bounded.
\end{assumption}

\begin{assumption}[Local factor loadings]
\label{ass: change_loc}
We assume:
\begin{itemize}
    \item [(i)] The change points in $\Theta_f$ are linearly spaced, i,e.,\ there exists $\rho_{f, w}$, $w\in[Q_f]$, such that $0< \rho_{f, 1} < \dots < \rho_{f, Q_f} < 1$ and $\theta_{f,w}=\lfloor \rho_{f, w} T \rfloor$.
    \item [(ii)] 
    Both $[\bR_{0}, \bR_{1}, \dots, \bR_{Q_f}]$ and $[\bV_{0}, \bV_{1}, \dots, \bV_{Q_f}]$ are of full row rank.
    \item [(iii)] For each $w\in\{0,\dots,Q_f\}$, we have $\|\bR_{w}\|_F$ and $\|\bV_{w}\|_F$ bounded
    away from zero and infinity.
    \item [(iv)] As $\min\{T,p,q\} \to\infty$, the following limiting matrices exist:
    \begin{align*}
        \bSigma_{R} &:= \lim_{T\to\infty} \frac{1}{T} \sum_{w=1}^{Q_f +1} \big\| \bV_{w-1} \big\|_F^2 \cdot (\theta_{f,w} - \theta_{f,w-1}) \cdot \bR_{w-1} \bR_{w-1}^\trans ,\\
        \bSigma_{V} &:= \lim_{T\to\infty} \frac{1}{T} \sum_{w=1}^{Q_f +1} \big\| \bR_{w-1} \big\|_F^2 \cdot (\theta_{f,w} - \theta_{f,w-1}) \cdot \bV_{w-1} \bV_{w-1}^\trans .
    \end{align*}
\end{itemize}
\end{assumption}

Assumption~\ref{ass: loadings} is commonly found in the literature on factor modeling.
Markedly, part~(ii) only imposes asymptotic orthogonality, which is weaker than many conditions seen in the similar context.
For instance, Assumption~2 of \cite{barigozzi2024moving} restricts how fast the scaled global loading matrix converges to the identity matrices, whereas \cite{ZhangCenCho2026} requires exact orthogonality.
Our relaxed condition is thanks to the use of vectorized $\wh\bG_t$ in the construction of $\cT^{(f)}_{\gamma, \tau}$ in~\eqref{eqn: test_stat_aggreg}, 
see also Remark~\ref{rem:zcc}.

Assumption~\ref{ass: core_factor} concerns the time series in the core factors, 
where the bounded fourth moment is similarly seen in 
Assumption~2.1 of \cite{ChenFan2023} and
Assumption~(F1) of \cite{LamCen2025_MEFM}, among others.
Assumption~\ref{ass: noise} allows both serial and cross-sectional dependence in $\bE_t$.
Assumption~\ref{ass: change_loc}~(i) is commonly found in the literature on change point analysis under time series factor models. 
As remarked earlier, the pseudo-factor contains factors attributed to changes in the loading space as well as (dis)appearance of factors and, jointly with Assumption~\ref{ass: loadings}, such a condition ensures that $\{\bG_t\}_{t = 1}^T$ is recoverable. 
Assumption~\ref{ass: change_loc}~(ii)--(iv) regulate the transformation matrices of changes in the factor structure.
It is noteworthy that we do not need to impose any restriction on the dependence between the main effects and $\{\bF_t\}_{t = 1}^T$ or $\{\bE_t\}_{t = 1}^T$, thanks to the model identification condition in Assumption~\ref{ass: identification}.

By construction of the scan statistics in Sections~\ref{subsec: detection_ME} and \ref{subsec: detection_fac}, it suffices to consider loading estimation consistency to underpin the investigation into the subsequent change point detection methods.
Hence in the following,
we present the consistency of the two-step estimators $\wh\bQ$ and $\wh\bU$, which
can also be of independent interest.

\begin{theorem}\label{thm: consistency_proj}
Let Assumptions~\ref{ass: identification}, \ref{ass: loadings}, \ref{ass: core_factor}, \ref{ass: noise} and \ref{ass: change_loc} hold, with $\nu = \min\{\nu_f, \nu_e\} \ge 4$ and $\eta = \min\{\eta_f, \eta_e\} \ge 0$.
Then the estimators $\wh\bQ$ and $\wh\bU$ satisfy that
\begin{align*}
    & \frac{1}{p} \left\| \wh\bQ - \bQ \wh\bH_r \right\|_F^2 = \cO_P\left( \frac{1}{Tq} + \frac{1}{T^2p} + \frac{1}{Tp^2} + \frac{1}{p^2q^2} \right), \\
    & \frac{1}{q} \left\| \wh\bU - \bU \wh\bH_c \right\|_F^2 = \cO_P\left( \frac{1}{Tp} + \frac{1}{T^2q} + \frac{1}{Tq^2} + \frac{1}{p^2q^2} \right),
\end{align*}
where $\wh\bH_r$ and $\wh\bH_c$ are asymptotically invertible such that as $\min\{T,p,q\} \to \infty$, we have $\wh\bH_r \wh\bH_r^\trans = \bI_{k_r} + o_P(1)$ and $\wh\bH_c \wh\bH_c^\trans = \bI_{k_c} + o_P(1)$.
\end{theorem}


Theorem~\ref{thm: consistency_proj} parallels Proposition~1 of \cite{Zhangetal2025} derived for tensor factor models with multiple change points.
It showcases that compared to the first-step estimator $\wt{\bQ}$, which attains the rate $\cO_P(1/(Tq) + 1/p^2)$ (see Theorem~2 of \citeauthor{LamCen2025_MEFM}, \citeyear{LamCen2025_MEFM}, where it is supposed that $Q_f = Q_\alpha = Q_\beta = 0$ and $\balpha_t$ and $\bbeta_t$ are deterministic), the second-step projection improves the rate of estimation as long as $(p+q)/T^2 =o(1)$,
which is non-trivial extension in the context of MEFM due to the interplay among different components in \eqref{eqn: MEFM_structure}.

\subsection{Consistency in first-order change point detection}

We establish the consistency of the change point detection method described in Section~\ref{subsec: detection_ME} which, through detecting the changes in the main effects under~\eqref{eqn: mean_change}, identify the shifts in the mean of $\{\bX_t\}_{t = 1}^T$.
For this, we make the following assumptions.

\begin{assumption}[Noise in the main effects]\label{ass: ME_noise}
For any $t\in[T]$, we have $\be_{\alpha,t} = \bA_{\alpha} \bf_{\alpha,t} + \bSigma_{\alpha} \circ \be_{\alpha,t}$,
with $\{\bf_{\alpha,t}\}_{t=1}^T, \, \{\be_{\alpha,t}\}_{t=1}^T \in \cG(\nu_\alpha, \eta_\alpha)$.
The deterministic matrices $\bA_{\alpha}$ and $\bSigma_{\alpha}$ contain uniformly bounded elements, and $\|\bA_{\alpha}\|_1 = \cO(1)$.
Analogous conditions are placed on $\{\be_{\beta,t} \}_{t=1}^T$.
\end{assumption}


\begin{assumption}[Trend functions]\label{ass: trend_fun}
Denote the $i$th component of the function $\bmu_{\alpha}$ by $(\bmu_{\alpha})_i$ and the $j$th component of the function $\bmu_{\beta}$ by $(\bmu_{\beta})_j$. For any $i\in[p]$ and $j\in[q]$, we assume that $(\bmu_{\alpha})_i, (\bmu_{\beta})_j \in \mathcal{C}^2[0,1]$, and that there is some constant $c>0$ satisfying
\[
\max\left( \max_{i\in[p]} \left\{ |(\bmu_{\alpha})_i^{'}|_\infty, |(\bmu_{\alpha})_i^{''}|_\infty\right\}, \, 
\max_{j\in[q]} \left\{ |(\bmu_{\beta})_j^{'}|_\infty, |(\bmu_{\beta})_j^{''}|_\infty\right\} \right) \leq c .
\]
\end{assumption}

\begin{assumption}[Kernel function and bandwidth]\label{ass: kernel}
The kernel function $K(\cdot)$ in~\eqref{eq:kernel} is symmetric and non-negative with support $[-1,1]$ and first-order derivative. Further, we assume that $\vert K\vert_{\infty}, \vert K'\vert_{\infty} <\infty$, $\int_{-1}^1 K(x)dx =1$, and $\kappa_1^2 \neq \kappa_0 \kappa_2$ where $\kappa_i:= \int_0^1 x^i K(x) dx$.
Also, the bandwidth satisfies $b\to 0$ and $bT \to\infty$.
\end{assumption}

\begin{assumption}[Changes in the first-order structure]
\label{assum:me:change:spacing}
We assume that the change points in $\Theta_\alpha$ are linearly spaced, i,e.,\ there exists $\rho_{\alpha, i}, \, i \in[Q_\alpha]$, such that $0< \rho_{\alpha, 1} < \dots < \rho_{\alpha, Q_\alpha} < 1$ and $\theta_{\alpha, i}=\lfloor \rho_{\alpha, i} T \rfloor$. Analogously, the change points in $\Theta_\beta$ are linearly spaced.
\end{assumption}

Similarly to Assumption~\ref{ass: noise}, Assumption~\ref{ass: ME_noise} allows for temporal and cross-sectional dependence in the noise time series $\{\be_{\alpha, t}\}_{t = 1}^T$ and $\{\be_{\beta, t}\}_{t = 1}^T$. 
Assumption~\ref{ass: trend_fun} ensures the smoothness of the trend functions and thus identifiability of the jump vectors in~\eqref{eqn: mean_change}
Assumption~\ref{ass: kernel} characterizes the kernel used in the construction of the detector statistics in~\eqref{eqn: mosum_main_effect}, and the requirement on the bandwidth $b$; we suggest a default choice for the kernel in Section~\ref{subsec: tuning}.
Assumption~\ref{assum:me:change:spacing} may be relaxed to permit sub-linear spacing in the change points with their total number increasing with $T$, at the price of more stringent condition on $\min_{i \in [Q_\alpha]} \Vert \bm\gamma_{\alpha, i} \Vert_\infty$ for detection consistency (see~\eqref{cond:thm: consistency_change_ME} below) and a worse rate of estimation;
we further discuss this point below the following theorem.

\begin{theorem}[Asymptotic consistency in main effects change point detection]
\label{thm: consistency_change_ME}
Suppose that Assumptions~\ref{ass: identification}, \ref{ass: noise}, \ref{ass: ME_noise}, \ref{ass: trend_fun}, \ref{ass: kernel} and~\ref{assum:me:change:spacing} hold, with $\nu = \min\{\nu_e, \nu_\alpha, \nu_\beta\} \ge 4$ and $\eta = \min\{\eta_e, \eta_\alpha, \eta_\beta\} > 1/2 - 1/\nu$. 
Suppose that $\cT^{(\alpha)}_{b, \tau}$ in \eqref{eqn: mosum_main_effect} is constructed with 
$\bW_{\alpha, \tau} = \bW_\alpha$ for all $\tau$,
where $\bW_\alpha$ is diagonal with all of its elements bounded away from zero and infinity. 
Let $\wh\Theta_\alpha = \big\{ \wh\theta_{\alpha, i}, \, i \in [\wh Q_\alpha]: \, \wh\theta_{\alpha, i} < \ldots < \wh\theta_{\alpha, \wh Q_\alpha} \big\}$ denote the output of Algorithm~\ref{alg: detect} with the scan statistic $\cT^{(\alpha)}_{b, \tau}$, bandwidth $bT$, and threshold $\pi^*= c_{T,p,q}^{(\alpha)} \big[ b^2 + (pT)^{1/\nu} / (Tb)^{1 - 1/\nu} + \sqrt{\log(pT)/(Tb)} \big]$ with $c_{T,p,q}^{(\alpha)} \to \infty$ arbitrarily slowly.
We further assume that 
\begin{align}
\label{cond:thm: consistency_change_ME}
\frac{c_{T,p,q}^{(\alpha)} \big[ b^2 + (pT)^{1/\nu} / (Tb)^{1 - 1/\nu} + \sqrt{\log(pT)/(Tb)} \big] }{\min_{i\in[Q_\alpha]} \big\| \bgamma_{\alpha,i} \big\|_\infty } \leq \frac{1}{2}
\text{ \ and \ }
b = o\left( \min_{i\in[Q_\alpha]} \big\| \bgamma_{\alpha,i} \big\|_\infty \right).
\end{align}
Then, as $\min\{T, p\} \to\infty$,
for any sequence $\kappa_{T,p}^{(\alpha)}$ diverging arbitrarily slowly,
we have
\[
\P\left( \wh{Q}_{\alpha} = Q_{\alpha}; \max_{i\in[\wh{Q}_{\alpha}]}
\big\| \bgamma_{\alpha,i} \big\|_\infty^2 |\wh\theta_{\alpha,i} - \theta_{\alpha,i}| 
\leq 
\Big\{ (pTb)^{2/\nu}  + \log(pTb) \Big\}
\, \kappa_{T,p}^{(\alpha)}
\right) \to 1 .
\]
Analogous results hold for $\wh\Theta_\beta$ output from Algorithm~\ref{alg: detect} with $\cT^{(\beta)}_{b, \tau}$ as its scan statistic. 
\end{theorem}

Theorem~\ref{thm: consistency_change_ME} establishes that the proposed method for detecting change points in the row main effect, consistently estimates both the total number and the locations of the change points, with the rate of estimation explicitly characterized by $\big\| \bgamma_{\alpha,i} \big\|_\infty$ and the tail behavior of the data through~$\nu$. 
We suppose that $\bW_{\alpha, \tau} = \bW_\alpha$ for all $\tau$, which facilitates the derivation of the rate of estimation; a similar approach has also been taken by \cite{EichingerKirch2018} in the context of univariate time series segmentation, and by \cite{barigozzi2024moving} for time series segmentation under a factor model.

The condition~\eqref{cond:thm: consistency_change_ME} stems from controlling the bias and random fluctuation in estimating the jump vectors $\bj_\alpha(\tau)$ uniformly over $\tau$ in the max norm. 
This and the rate of estimation simultaneously place restrictions on the growth of $T$ and $p$ (and $q$ for detecting changes in the column main effect) as well as the rate at which $b \to 0$.
Provided that $\Vert \bgamma_{\alpha,i} \Vert_\infty = \cO(1)$, it translates to
\begin{align}
\label{eq:bound:p}
\max\{p, \, q \} \ll \min\left\{ (Tb)^{\nu - 1}/T, \, T^{\nu/2 - 1}/b  \right\},
\end{align}
which shows that larger dimensions are permitted as the number of moments $\nu$ increases.
We remark that the term $b^2$ in both $\pi^*$ and the first condition in~\eqref{cond:thm: consistency_change_ME}, is due to the bias in estimating the smooth varying trend of $\{\bmu_{\alpha,t}\}_{t = 1}^T$, and it can be removed if the trend is exactly constant. 
Further, if such knowledge is available, we can simply scan the difference in local averages (cf.\ \citeauthor{EichingerKirch2018}, \citeyear{EichingerKirch2018}, for the univariate time series segmentation), lifting the second condition in~\eqref{cond:thm: consistency_change_ME}.

\subsection{Consistency in second-order change point detection}

We derive the consistency of the method introduced in Section~\ref{subsec: detection_fac} for detecting the change points in the second-order structure of $\{\bX_t\}_{t = 1}^T$ through detecting multiple shifts in the factor-driven component under~\eqref{eqn: cov_change}.

\begin{theorem}[Asymptotic consistency in factor structure change point detection]\label{thm: consistency_change_fac}
Suppose that all assumptions in Theorem~\ref{thm: consistency_proj} hold with $\nu = \min\{\nu_f, \nu_e\} \ge 8$ and $\eta = \min\{\eta_f, \eta_e\} > 1$. 
Suppose that $\cT^{(f)}_{\gamma, \tau}$ in~\eqref{eqn: test_stat_aggreg} is constructed with $\bW_{G, \tau} = \bW_W$ for all $\tau$, where $\bW_W$ is positive definite with all of its eigenvalues bounded away from zero and infinity.
Let $\wh\Theta_f = \{ \wh\theta_{f,w}, \, w\in [\wh{Q}_f] : \, \wh\theta_{f,1} < \ldots < \wh\theta_{f,\wh{Q}_f} \}$ denote the output of Algorithm~\ref{alg: detect} with the scan statistic $\cT^{(f)}_{\gamma, \tau}$, bandwidth $\gamma$,
and threshold $\pi^*= c_{T,p,q}^{(f)} \sqrt{\log(T)}$ with $c_{T,p,q}^{(f)} \to \infty$ arbitrarily slowly. 
We further assume that
\begin{align}
2\gamma & \leq \min_{w\in[Q_f +1]}\{\theta_{f,w} - \theta_{f,w-1}\}, 
\label{cond:thm: consistency_change_fac:one}
\\
\frac{T}{\log(T)^{\nu/4} \gamma^{\nu/4 - 1}} & = \cO(1) 
\text{ \ and \ }
\frac{c_{T,p,q}^{(f)} \sqrt{\log(T)}}{\min_{w\in[Q_f]} \omega_{f,w} \sqrt{\gamma}} \leq \frac{1}{4}, 
\label{cond:thm: consistency_change_fac:two}
\end{align}
where $\omega_{f,w}$ is the jump magnitude defined in \eqref{eqn: def_jump_factor}.
Then, as $\min\{T, p, q\} \to\infty$, for any sequence $\kappa_{T,p,q}^{(f)}$ diverging arbitrarily slowly, we have
\[
\P\left( \wh{Q}_f = Q_f; \max_{w\in[\wh{Q}_f]} \omega_{f,w}^2 |\wh\theta_{f,w} - \theta_{f,w}| \leq \kappa_{T,p,q}^{(f)} \right) \to 1 .
\]
\end{theorem}

The requirement that $\nu \geq 8$ is necessary for the bandwidth $\gamma$ to simultaneously satisfy \eqref{cond:thm: consistency_change_fac:one}--\eqref{cond:thm: consistency_change_fac:two}. 
This tail condition is comparable to those found in the literature on change point detection under factor models for vector \citep{Duanetal2023}, matrix \citep{peng2025detection} or higher-order tensor \citep{ZhangCenCho2026} time series. 
Compared to \cite{barigozzi2024moving}, also adopting a MOSUM procedure in this context, we make a somewhat stronger requirement on the rate at which $\gamma$ diverges with $T$.
This may be attributed to that their theoretical analysis is based on deriving the asymptotic null distribution of the maximum scan statistic over $\tau$, whereas we characterize the deviation of the detector statistic from its population counterpart by means of Hanson--Wright-type inequality of \cite{ZhangWu2021}.

\section{Numerical experiments}\label{sec: numerical_results}

\subsection{Selection of tuning parameters}
\label{subsec: tuning}

We provide a detailed guidance on the selection of tuning parameters; the default recommendations are adopted in all our numerical experiments reported later in this section.

\paragraph{Factor numbers.}
Inspired by \cite{LamCen2025_MEFM} and \cite{barigozzi2026statistical},
we select $k_r$ and $k_c$ by utilizing the ratios of the eigenvalues of $\sum_{t=1}^T \wh\bL_t \wh\bL_t^\trans$ and $\sum_{t=1}^T \wh\bL_t^\trans \wh\bL_t$, respectively.
Specifically, we set $\wh{k}_r := \argmax_{1 \le \ell \le \overline{k}_r} \lambda_\ell\big( \sum_{t=1}^T \wh\bL_t \wh\bL_t^\trans \big)/ \lambda_{\ell+1}\big( \sum_{t=1}^T \wh\bL_t \wh\bL_t^\trans \big)$, where $\overline{k}_r = \lceil p/3 \rceil$ is an upper bound on the row factor number;
$\wh{k}_c$ is computed similarly.
These estimators then are input to estimate the pseudo-factor time series as in~\eqref{eqn: Gt_hat}.

\paragraph{Kernel functions.}
As in \cite{Chenetal2022}, we use the Epanechnikov kernel defined as $K(x)=3 (1-x^2)/4$ with support $[-1,1]$ in constructing the detector statistics for the row and column main effects, which meets Assumption~\ref{ass: kernel}. Alternative choices include the kernels e.g.,\ from the symmetric Beta family.

\paragraph{Bandwidth for change point detection.}
We set $\gamma= \lfloor C_\gamma T/\log(T) \rfloor$ in Algorithm~\ref{alg: detect}, where $C_\gamma >0$ is some constant to be specified for flexibility,
and $C_\gamma = 1.5$ is used in our experiments.
This choice allows us to detect all changes that are linearly spaced as assumed in Assumptions~\ref{ass: change_loc}~(i) and~\ref{assum:me:change:spacing}.
For practitioners, we recommend a similar bandwidth when no prior information is available.

\paragraph{Threshold for change point detection.}
While Theorems~\ref{thm: consistency_change_ME} and~\ref{thm: consistency_change_fac} specify the theoretical requirement on the threshold $\pi^*$ for detection consistency, these specifications involve unknown quantities and cannot directly be utilized. 
Following a popularly adopted approach in the change point literature, we propose to train a rule for selecting the threshold which depends only on $T$, $p$ and $q$. 
For this, we simulate the data under MEFMs without any change point,
record the maximal detector statistics computed on the simulated data, take the empirical $0.9$-quantiles of these maxima over $100$ such realizations for different combinations of dimensions,
and regress them on known functions of $T$, $p$, and~$q$.
The fitted models are then used to form the thresholds.
We provide further detail in Section~\ref{subsec: thres_pi_numerical}.

\paragraph{Weight matrices.}
For ease of deriving the rate of estimation,
both Theorems~\ref{thm: consistency_change_ME} and \ref{thm: consistency_change_fac} suppose that the detector statistics in~\eqref{eqn: mosum_main_effect} and~\eqref{eqn: test_stat_aggreg} are constructed with \textit{global} weight matrices (i.e.,\ time-invariant) with bounded entries and eigenvalues. 
In practice, it is recommended that \textit{local} weight matrices (i.e.,\ time-dependent) are adopted to account for potential heteroskedasticity in the data.
We adopt local weight matrices in the numerical experiments reported in the main text, and report additional results with the global weights in the supplementary material.

For change point detection in the row main effect,
we set as $\bW_{\alpha, \tau}$ the diagonal matrix consisting of the local sample variances $\wh{w}_{\alpha,i,\tau}, \, i \in [p]$. 
where 
\begin{align*}
    \wh{w}_{\alpha,i,\tau} :=
    \frac{1}{2\gamma -2}
    \left\{
    \sum_{t=\tau+1}^{\tau+ \gamma} \left( \wh\alpha_{t,i} - \frac{\sum_{\ell=\tau+1}^{\tau+ \gamma} \wh\alpha_{\ell,i}}{\gamma } \right)^2
    + \sum_{t=\tau- \gamma}^{\tau-1} \left( \wh\alpha_{t,i} - \frac{\sum_{\ell=\tau- \gamma}^{\tau-1} \wh\alpha_{\ell,i}}{\gamma} \right)^2
    \right\} ,
\end{align*}
whose construction accounts for the possible presence of a change in the mean.
The weights for detecting changes in the column main effects can be specified analogously.
For change point detection in the second-order structure,
we use the diagonal version of the HAC-type estimator for the long-run covariance matrix of $\{\vec(\wh\bG_t)\}_{t = 1}^T$. 
Specifically,
given any interval $[s,e]$ and bandwidth $w_0$,
the HAC-type estimator $\bW_{s,e}^{\mathsf{HAC}}$ is computed as
\begin{align*}
\bW_{s,e}^{\mathsf{HAC}} &:= \wh\bGamma_{s,e}(0) + \sum_{\ell = 1}^{w_0} \left(1-\frac{\ell}{w_0+1} \right) \left\{\wh\bGamma_{s,e}(\ell)+ \wh\bGamma_{s,e}(\ell)^\trans \right\} , \text{\ where \ } \\
\wh\bGamma_{s,e}(\ell) &:= \frac{1}{e-s+1} \sum_{t = s + \ell}^{e}
\wh\bg_{s,e}(t) \, \wh\bg_{s,e}(t-\ell)^\trans \text{ \ with} \\
\wh\bg_{s,e}(t) &:=
\Vech\left\{ \vec(\wh\bG_t) \vec(\wh\bG_t)^\trans - \frac{1}{e-s+1}\sum_{h=s}^{e} \vec(\wh\bG_h) \vec(\wh\bG_h)^\trans \right\} .
\end{align*}
In our numerical experiments, we set $\bW_{G, \tau} = \bW_{\tau- \gamma +1, \tau+ \gamma}^{\mathsf{HAC}}$ with $w_0=\lfloor (e-s+1)^{1/8}\rfloor$.

\subsection{Simulation studies}\label{subsec: simulation}

In this subsection, we demonstrate the numerical performance of our methods by Monte Carlo experiments.

\subsubsection{Settings}
\label{sec:setting}

We first describe the data generation processes under the model~\eqref{eqn: MEFM_structure}, apart from the specifications of the change points which are given in the following subsections.

\paragraph{Data generating mechanism.}
The global row loading matrix is generated with $\bQ=\bM_p \bQ^*$, where each entry of $\bQ^* \in\R^{p\times k_r}$ is i.i.d.\ $\cN(0,1)$, and pre-multiplication by $\bM_p$ ensures the identification condition in Assumption~\ref{ass: identification} is fulfilled.
The global column loading matrix $\bU$ is formed in an analogous manner.
We fix $k_r=k_c=2$ throughout this section,
and set any transformation matrix for the first segment as an identity matrix, i.e.,\ $\bR_0 = \bI_{k_r}$ and $\bV_0 = \bI_{k_c}$.
Entries of $\bA_{e,r}$ and $\bA_{e,c}$ in Assumption~\ref{ass: noise} are randomly set to zero with probability $0.95$, and the non-zero ones are i.i.d.\ $\cN(0,1)$.
We fix $k_{e,r} = k_{e,c} = 2$.
The elements of $\mu_t$, $\be_{\alpha,t}$, $\be_{\beta,t}$ and $\bF_t$ are independently generated as AR(1) processes with unit variance with AR coefficient~0.5. 
We generate the components of $\bE_t$ analogously, where the elements in $\bF_{e,t}$ and $\bE_{e,t}$ are also independent, AR(1) processes with unit variance, except that their AR coefficients are $-0.5$ and $0.4$, respectively.
Each entry in $\bSigma_{e}$ is generated as i.i.d.\ from the distribution $|\cN(0,1)|$. 
To illustrate the sensitivity of our method to heavy-tailed data, 
the innovation processes in all AR processes above are sampled from i.i.d.\ $t_6/\sqrt{5}$ distribution.
For the trend function in the row main effects, we adapt the generating process in \cite{Chenetal2022} and compute $\bmu_{\alpha}(x) = \bM_p \bmu_{\alpha}^*(x)$, where $\bmu_{\alpha}^*$ is sinusoidal such that $(\bmu_{\alpha}^*(x))_i = \sin(2\pi x + i/p)$. 
The trend for the column main effects is similarly generated, with $p$ replaced by $q$. We specify the changes in the factor structure (characterized by the transformation matrices) and in the main effects later. Every experiment in this section is repeated $n = 100$ times unless specified otherwise.

\paragraph{Competing methods.}
We refer to each component of our method detecting changes in the row and column main effects, and in the factor structure as: Our-row, Our-col, and Our-fac, respectively.
To the best of our knowledge, there is no method designed to simultaneously detect (possibly) multiple change points in the mean and covariance of matrix time series.
Therefore, we compare with the following change point detection methods:
MMEAN \citep{pevsta2025tensor} is designed to detect mean shifts in matrix time series,
and TFMSEG \citep{ZhangCenCho2026} detects changes in the covariance characterized by changes in a tensor factor model.
We also consider methods developed for handling vector time series, with the vectorization of the matrix time series as an input:
SIMUL \citep{cui2025simultaneous} jointly tests for a mean and/or covariance shift in vector time series,
VMEAN \citep{Chenetal2022} detects multiple mean shifts in the mean of a vector time series,
FTS \citep{barigozzi2024moving} performs time series segmentation under a vector factor model,
and
COVRK \citep{RyanKillick2023} detects covariance shifts in a vector time series.
Any tuning parameters for the above methods are either determined by the default values recommended in the corresponding papers, or we set them to be the same as ours if such information is not given.


\subsubsection{Single change point scenarios} 
\label{subsec:single}

\paragraph{When the changes in the mean and covariance are aligned.}

First, we consider a scenario with a single change point at the same location appearing in both mean and covariance. 

\begin{itemize}
    \item [(S1)] Both row main effects and factor structure have one change at $\theta_{\alpha,1} = \theta_{f,1} = \lfloor 0.4T \rfloor$, and there is no change in the column main effects.
    For the change in row main effects, entries in $\bgamma_{\alpha,1}$ are such that
    $(\bgamma_{\alpha,1})_j = 8/\log(pT)$ for all $j\in[p/2]$ and $-8/\log(pT)$ otherwise.
    The change in the factor structure occurs in column factor loading matrix, characterized by the transformation matrix $\bV_1=[1,0;0,0]$, i.e.,\ one factor disappears after the change point. We consider $T \in \{600, 900, 1200\}$ and $p, q \in \{10, 15, 20\}$ with $p = q$.
\end{itemize}
The dense shift in the mean is designed to be well-detected by SIMUL or MMEAN where $\ell_2$-based aggregation is adopted. 
The covariance change is characterized by the disappearance of a factor, which is motivated by the simulation design used in \cite{Duanetal2023} and \cite{barigozzi2024moving}.
Under this scenario, we compare the proposed method with those dedicated to detecting changes in the mean (MMEAN, VMEAN), in the covariance (TFMSEG, FTS, COVRK), or both with the change points perfectly aligned (SIMUL),
and we apply the competing methods with the true number of change points in the structure of interest for each, to investigate their ``oracle'' performance in change point estimation.
For this, we report the scaled RMSE of the estimated change (as a portion of $T$), which is computed as $\text{scaled RMSE} := \sqrt{ n^{-1} \sum_{i=1}^n \{ (\wh\theta_i - \theta)/T\}^2 } ,$
for any true change location $\theta$ and the estimated changes $\{\wh\theta_i\}$ over $n$ realizations.
On the other hand, our method is applied without such knowledge, and we also report its detection consistency as the percentage of realizations where exactly one change point is detected for Our-row and Our-fac, or no change is detected for Our-col; see Table~\ref{tab: sim_S1}. 
Since there is no change in the column main effects, we do not report the scaled RMSE for Our-col.

\begin{table}[!t]
\centering
\resizebox{\ifdim\width>\linewidth\linewidth\else\width\fi}{!}{
\begin{tabular}[t]{cc cccccccccccccccccccccccc}
\toprule
$T$ & $(p,q)$  & Our-row & Our-col & Our-fac & TFMSEG & SIMUL & MMEAN & VMEAN & FTS & COVRK \\
\cmidrule(lr){1-2} \cmidrule(lr){3-5} \cmidrule(lr){6-11}
600 & (10,10) & 2.32 ${(83)}$ & $\texttt{NA}$${(69)}$ & 7.53 ${(72)}$ & 9.94 & 5.35 & 6.75 & 15.79 & 4.94 & 48.96 \\
900 & (10,10) & 1.59 ${(81)}$ & $\texttt{NA}$${(63)}$ & 11.20 ${(65)}$ & 10.15 & 3.04 & 3.04 & 19.23 & 10.53 & 49.62 \\
1200 & (10,10)  & 1.71 ${(87)}$ & $\texttt{NA}$${(49)}$ & 7.44 ${(69)}$ & 10.90 & 4.06 & 4.17 & 16.86 & 12.84 & 0.62 \\
\addlinespace
600 & (15,15)  & 2.05 ${(94)}$ & $\texttt{NA}$${(85)}$  & 7.45 ${(80)}$ & 2.06 & 4.71 & 6.23 & 18.71 & 7.72 & 41.16 \\
900 & (15,15)  & 1.34 ${(89)}$ & $\texttt{NA}$${(76)}$  & 5.64 ${(68)}$ & 7.34 & 2.86 & 3.51 & 17.25 & 12.54 & 45.92 \\
1200 & (15,15)  & 1.70 ${(91)}$ & $\texttt{NA}$${(74)}$  & 6.22 ${(64)}$ & 5.69 & 2.92 & 2.93 & 16.80 & 16.56 & 48.35 \\
\addlinespace
600 & (20,20)  & 1.93 ${(95)}$ & $\texttt{NA}$${(91)}$ & 7.61 ${(75)}$ & 2.06 & 5.49 & 6.18 & 17.42 & 6.90 & 44.49 \\
900 & (20,20)  & 1.60 ${(99)}$ & $\texttt{NA}$${(85)}$ & 6.29 ${(73)}$ & 0.47 & 3.49 & 3.49 & 18.69 & 14.35 & 39.48 \\
1200 & (20,20)  & 2.35 ${(96)}$ & $\texttt{NA}$${(66)}$ & 4.42 ${(78)}$ & 0.21 & 3.22 & 3.22 & 18.47 & 21.30 & 43.93 \\
\bottomrule
\end{tabular}}
\caption{Scaled RMSE of the estimated change point for all methods under Setting~(S1), with all values multiplied by $10^2$.
Each competing method is applied with the knowledge that the true number of the change point in the structure it is designed to detect, is one;
for our methods, the proportion of exact detection over all realizations are also reported in brackets in percentage. 
There is no change on the column main effects and hence no scaled RMSE is reported.
}
\label{tab: sim_S1}
\end{table}

Our-row exhibits competitive and stable performance across all combinations of dimensions, which is comparable to that of SIMUL and MMEAN, 
and markedly superior to that of VMEAN.
This good performance is attributed to the explicit estimation step extracting the row-specific effects, which contributes to enhancing the signal-to-noise ratio.
From the reported proportion of exact detection, Our-col correctly identifies the absence of a change in the majority of cases, and the detection performance improves with increasing dimensions. 
Our-fac also shows competitive performance against those methods dedicated to detecting changes in the covariance, such as COVRK and FTS.
TFMSEG consistently performs well in all settings, and tends to benefit from the presence of the coincident mean change by detecting its presence as a factor (see our discussion of Setting~(S2) below).
COVRK tends to perform poorly for most settings except for when $(T,p,q)=(1200,10,10)$, indicating that its good performance relies on large sample size relative to the dimension.

\paragraph{When the changes in the mean and covariance are not aligned.}
We also consider the situation when the changes in the first- and second-order structure are not aligned. 
\begin{itemize}
    \item [(S2)] Same as (S1), except that the change in the factor structure is at $\theta_{f,1} = \lfloor 0.6T \rfloor$;
    we focus on dimension $(T,p,q)=(1200,20,20)$ and report the results with varying $T$ and $(p, q)$ in Section~\ref{appendix: numerical}.
\end{itemize}

\begin{figure}[!h!t!p!]
\centering
\includegraphics[width=0.8\linewidth]{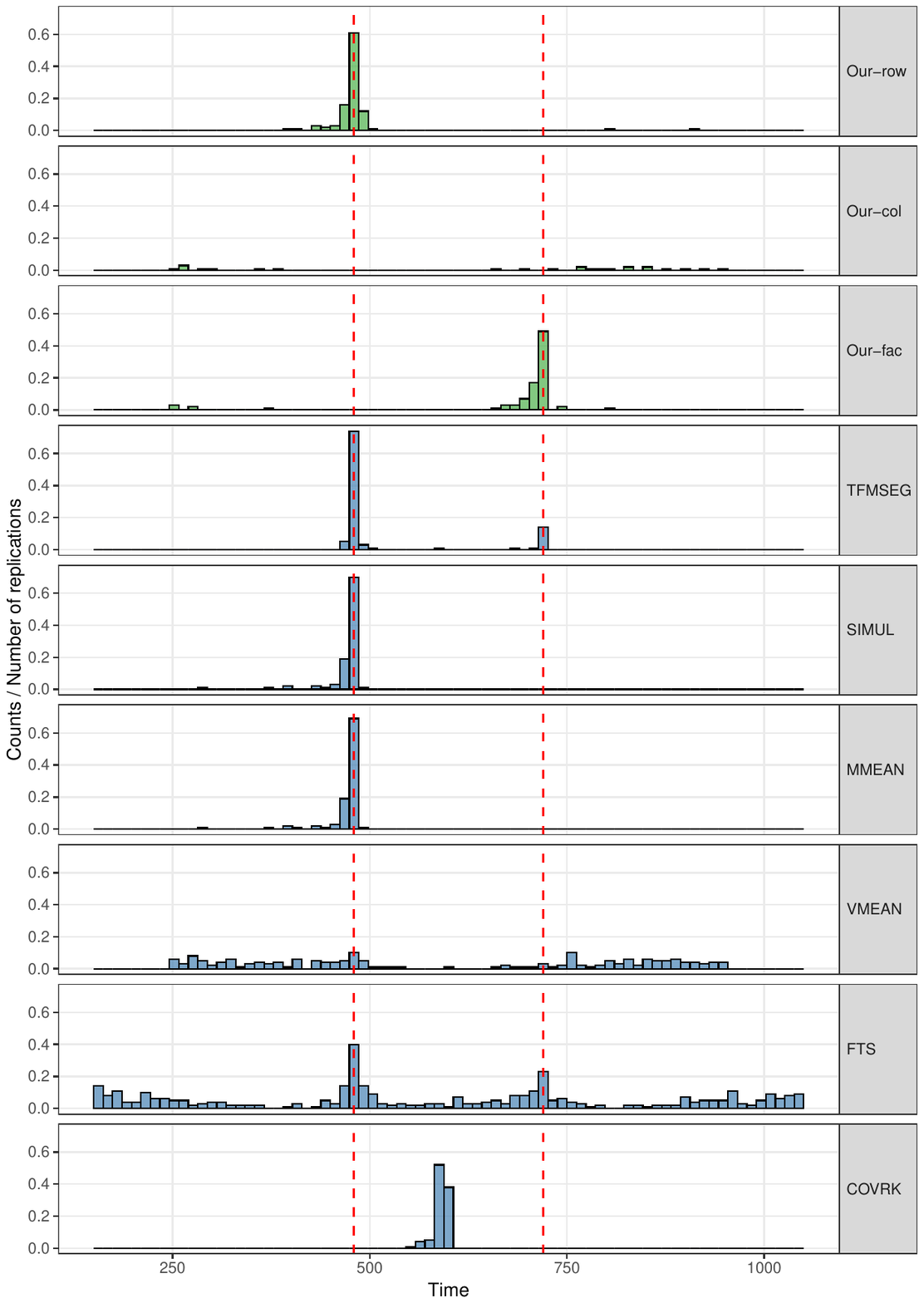}
\caption{
Histograms of detected change points for each method under Setting~(S2), where the top three panels contain the performance of our proposed method. The locations of the true change points in the mean ($\theta_{\alpha, 1} = \lfloor 0.4T \rfloor$) and covariance ($\theta_{f, 1} = \lfloor 0.6T \rfloor$) are marked by the vertical dashed lines.
}
\label{fig: sim_S2}
\end{figure}

For this setting, we do not supply the true number of the change points to any methods.
The results are summarized in Figure~\ref{fig: sim_S2}, which shows that our proposed method works satisfactorily: Our-row detects the single change point at $\theta_{\alpha, 1} = \lfloor 0.4T \rfloor$, performing as well as or better than those methods designed for mean change detection (MMEAN, VMEAN); Our-col performs consistently by detecting no change point; Our-fac also shows desirable performance by detecting the single change point at $\theta_{f, 1} = \lfloor 0.6T \rfloor$, unlike other methods that are specialized in detecting changes in the covariance (TFMSEG, FTS, COVRK).
SIMUL, tailored to detect at most one change in either mean or covariance, invariably prefers to detect the change in the mean. 
It is interesting to note that TFMSEG and FTS, while proposed to detect a change in the covariance structure under factor models, also detects the first change point in the mean, possibly as such a change affecting all row cross-sections of the data manifests itself as a ``factor''; 
this highlights the benefit of the proposed method which is able to attribute each detected change to either a shift in the mean in the row or column categorizations, or that in the second-order structure.



\subsubsection{Multiple change point scenarios}
\label{subsec:multi}

To assess the performance of our method in detecting multiple change points, we consider the following four types of scenarios:
(i) changes in the row main effects only, indexed by ``R'';
(ii) changes in the column main effects only, indexed by ``C'';
(iii) changes in the factor structure only, indexed by ``F'';
and (iv) mixtures of the three types, indexed by ``M''.
Types~(i)--(iii) are experimented over dimensions $T \in \{600, 800, 1000\}$ and $(p,q) \in \{(20,20), (30,30), (40,40)\}$, whereas we fix $(T,p,q) = (800, 30, 30)$ for Type~(iv). We repeat $n = 100$ realizations per setting, which is detailed in sequel below.
\begin{itemize}
    \item [(R1)] Two changes at $(\theta_{\alpha,1}, \theta_{\alpha,2})=(\lfloor0.3T\rfloor, \lfloor0.6T\rfloor)$.
    The jump vectors $\bgamma_{\alpha,1}$ and $\bgamma_{\alpha,2}$ have all entries at zero except for $(\bgamma_{\alpha,1})_{5}=(\bgamma_{\alpha,2})_{5}=5$ and $(\bgamma_{\alpha,1})_{10}=(\bgamma_{\alpha,2})_{10}=-5$.
    \item [(R2)] Three changes at $(\theta_{\alpha,1}, \theta_{\alpha,2}, \theta_{\alpha,3})=(\lfloor0.25T\rfloor, \lfloor0.5T\rfloor, \lfloor0.75T\rfloor)$, with jump vectors constructed as in (R1).
    \item [(C1)] Same as (R1) but for the column main effects, with $\bgamma_{\beta,i}$ playing the role of $\bgamma_{\alpha,i}$.
    \item [(C2)] Same as (R2) but for the column main effects, with $\bgamma_{\beta,i}$ playing the role of $\bgamma_{\alpha,i}$.
    \item [(F1)] One change in the column factor loading matrix at $\theta_{f,1}=\lfloor0.4T\rfloor$, characterized by $\bV_1=[1,0;0,0]$, i.e.,\ one factor disappears after the change point.
    \item [(F2)] Two changes in the factor structure: a change in the column loading at $\theta_{f,1}= \lfloor0.3T\rfloor$ with $\bV_1=[1,0; 0,0]$, followed by a change in the row loading at $\theta_{f,2}= \lfloor0.7T\rfloor$ with $\bR_2=[0.5,0; 0.5,0.5]$.
    \item [(M1)] One row change at $\theta_{\alpha,1}=\lfloor0.4T\rfloor$, one column change at $\theta_{\beta,1}=\lfloor0.6T\rfloor$, and one factor change in the column factor loading matrix at $\theta_{f,1}=\lfloor0.5T\rfloor$,
    with the changes introduced as in (R1), (C1), and (F1), respectively.
    \item [(M2)] Two row changes at $(\theta_{\alpha,1}, \theta_{\alpha,2}) = (\lfloor0.3T\rfloor, \lfloor0.6T\rfloor)$ as in (R1); and two factor changes at $(\theta_{f,1}, \theta_{f,2}) = (\lfloor0.4T\rfloor, \lfloor0.7T\rfloor)$ as in (F2). 
    \item [(M3)] Two row changes at $(\theta_{\alpha,1}, \theta_{\alpha,2}) = (\lfloor0.3T\rfloor, \lfloor0.6T\rfloor)$, two column changes at $(\theta_{\beta,1}, \theta_{\beta,2}) =(\lfloor0.4T\rfloor, \lfloor0.7T\rfloor)$, and one factor change in the column factor loading matrix at $\theta_{f,1}=\lfloor0.5T\rfloor$, with the changes introduced as in (R1), (C1), and (F1), respectively.
\end{itemize}

As our settings involve several true change points, we evaluate each component of our procedure separately.
Let $\Theta=\{\theta_1 < \dots < \theta_Q\}$ denote the true change points of a given component, and $\wh\Theta= \{\wh\theta_1 <\dots <\wh\theta_{\wh{Q}}\}$ the detected changes
in a single realization.
The detected change points are matched to the true ones by minimizing the total absolute distance such that, when $\wh{Q}\geq Q$, the $Q$ true change points are paired with the $Q$ closest detected change points, and vice versa. 
This yields a set of $\min(Q,\wh{Q})$ matched pairs per realization, denoted as $\wh\Theta_{\mathsf{paired}} = \{ (\theta(\iota), \wh\theta(\iota)) \}_{\iota=1}^{\min(Q,\wh{Q})}$, where $(\theta(\iota), \wh\theta(\iota)) \in \Theta\times \wh\Theta$ denotes the $\iota$th matched pair.
We report three quantities:
(a) the scaled RMSE (reducing to the definition used in Table~\ref{tab: sim_S1} when $Q=\wh{Q}=1$), defined as
\[
\text{scaled RMSE} := \sqrt{ \sum_{i=1}^{n}
\sum_{\iota=1}^{\min(Q,\wh{Q}_i)} \left( \frac{\wh\theta(\iota)_i - \theta(\iota)_i}{T}\right) ^2  
\Bigg/
\left( \sum_{i=1}^{n} \min(Q,\wh{Q}_i) \right)
};
\]
(b) the exact detection rate, i.e.,\ the percentage of realizations in which $\wh{Q} = Q$;
and (c) the Rand Index \citep{Rand1971} of the estimated time segmentation against the truth. Apart from (M1)--(M3), scaled RMSE is reported only for the component undergoing any change.
See Table~\ref{tab:sim_multi_all} for the complete results.

For the detection of row and column main effects, across all combinations of $T$ and $(p,q)$, the scaled RMSE is at most $0.13\%$ of the sample size, the proportion of exact detection is $100\%$ in every setting, and the Rand Index exceeds $99.7\%$.
This 
confirms that the presence of multiple change points does not degrade the accuracy of the procedure.
The corresponding false positive rates for the components that undergo no change, e.g., the column and factor components under Settings~(R1)--(R2), are low, ranging between $5\%$ and $15\%$ in the smallest dimension and shrinking as $T$ and $(p,q)$ increase.
In comparison, the detection of the factor structure is markedly harder than the main effects, with scaled RMSE ranging from $4\%$ to $8\%$, exact detection rate between $52\%$ and $96\%$, and Rand Index between $80\%$ and $92\%$.
Such gap is expected, as these changes in the second-order structure bring changes to the variability of the detector statistic $\cT^{(f)}_{\gamma, \tau}$ in~\eqref{eqn: test_stat_aggreg}, which makes the detection of the change points more challenging.
This is more pronounced at smaller $T$ and $(p,q)$, where the pseudo-factor itself is less accurately recovered (see Theorem~\ref{thm: consistency_proj}).
The mixed settings demonstrate that our procedure correctly attributes changes to the corresponding components. 
Under (M1)--(M3), even when different types of change coexist, the performance of Our-row and Our-col remains outstanding, and Our-fac is hardly affected by the presence of the mean changes.
We also note that adding more change points incurs only a marginal loss in accuracy, as observed e.g.,\ under (R2) relative to (R1).

\begin{table}[!h!t!p]
\centering
\caption{\label{tab:sim_multi_all}
Multi-change performance over different settings. For each component, RMSE, Exact, and RI denote the scaled RMSE, proportion of exact detection, and Rand Index, respectively.}
\centering
\resizebox{\ifdim\width>\linewidth\linewidth\else\width\fi}{!}{
\fontsize{7}{9}\selectfont
\begin{tabular}[t]{lrllllllllll}
\toprule
Setting & $T$ & $(p,q)$ & RMSE$^{\text{row}}$ & RMSE$^{\text{col}}$ & RMSE$^{\text{fac}}$ & Exact$^{\text{row}}$ & Exact$^{\text{col}}$ & Exact$^{\text{fac}}$ & RI$^{\text{row}}$ & RI$^{\text{col}}$ & RI$^{\text{fac}}$\\
\cmidrule(lr){1-3} \cmidrule(lr){4-6} \cmidrule(lr){7-9} \cmidrule(lr){10-12}
R1 & 600 & (20,20) & 0.13 & --- & --- & 100.00 & 91.00 & 91.00 & 99.80 & 96.05 & 95.75\\
R1 & 600 & (30,30) & 0.12 & --- & --- & 100.00 & 96.00 & 91.00 & 99.79 & 98.39 & 95.89\\
R1 & 600 & (40,40) & 0.12 & --- & --- & 100.00 & 97.00 & 88.00 & 99.81 & 98.67 & 94.22\\
R1 & 800 & (20,20) & 0.09 & --- & --- & 100.00 & 85.00 & 80.00 & 99.83 & 93.20 & 90.84\\
R1 & 800 & (30,30) & 0.09 & --- & --- & 100.00 & 95.00 & 85.00 & 99.84 & 97.99 & 93.34\\
R1 & 800 & (40,40) & 0.09 & --- & --- & 100.00 & 98.00 & 85.00 & 99.83 & 99.13 & 93.01\\
R1 & 1000 & (20,20) & 0.08 & --- & --- & 100.00 & 89.00 & 85.00 & 99.85 & 94.83 & 93.02\\
R1 & 1000 & (30,30) & 0.08 & --- & --- & 100.00 & 93.00 & 84.00 & 99.87 & 97.12 & 92.34\\
R1 & 1000 & (40,40) & 0.08 & --- & --- & 100.00 & 95.00 & 72.00 & 99.88 & 97.98 & 87.07\\
R2 & 600 & (20,20) & 0.12 & --- & --- & 100.00 & 92.00 & 79.00 & 99.75 & 96.56 & 89.82\\
R2 & 600 & (30,30) & 0.12 & --- & --- & 100.00 & 95.00 & 84.00 & 99.75 & 97.77 & 92.70\\
R2 & 600 & (40,40) & 0.12 & --- & --- & 100.00 & 100.00 & 95.00 & 99.73 & 100.00 & 97.72\\
R2 & 800 & (20,20) & 0.09 & --- & --- & 100.00 & 90.00 & 88.00 & 99.80 & 95.28 & 94.90\\
R2 & 800 & (30,30) & 0.09 & --- & --- & 100.00 & 98.00 & 83.00 & 99.81 & 99.06 & 92.26\\
R2 & 800 & (40,40) & 0.09 & --- & --- & 100.00 & 98.00 & 84.00 & 99.81 & 99.20 & 92.45\\
R2 & 1000 & (20,20) & 0.07 & --- & --- & 100.00 & 87.00 & 75.00 & 99.85 & 94.21 & 89.16\\
R2 & 1000 & (30,30) & 0.07 & --- & --- & 100.00 & 95.00 & 80.00 & 99.84 & 97.85 & 90.89\\
R2 & 1000 & (40,40) & 0.08 & --- & --- & 100.00 & 97.00 & 76.00 & 99.84 & 98.64 & 88.98\\
C1 & 600 & (20,20) & --- & 0.12 & --- & 95.00 & 100.00 & 92.00 & 97.75 & 99.76 & 96.47\\
C1 & 600 & (30,30) & --- & 0.12 & --- & 100.00 & 100.00 & 91.00 & 100.00 & 99.79 & 96.02\\
C1 & 600 & (40,40) & --- & 0.12 & --- & 100.00 & 100.00 & 87.00 & 100.00 & 99.78 & 94.32\\
C1 & 800 & (20,20) & --- & 0.09 & --- & 93.00 & 100.00 & 86.00 & 96.92 & 99.83 & 93.57\\
C1 & 800 & (30,30) & --- & 0.10 & --- & 99.00 & 100.00 & 86.00 & 99.54 & 99.81 & 93.51\\
C1 & 800 & (40,40) & --- & 0.09 & --- & 100.00 & 100.00 & 84.00 & 100.00 & 99.84 & 92.82\\
C1 & 1000 & (20,20) & --- & 0.08 & --- & 90.00 & 100.00 & 80.00 & 95.59 & 99.86 & 91.21\\
C1 & 1000 & (30,30) & --- & 0.08 & --- & 98.00 & 100.00 & 82.00 & 99.08 & 99.86 & 91.67\\
C1 & 1000 & (40,40) & --- & 0.08 & --- & 99.00 & 100.00 & 70.00 & 99.53 & 99.85 & 85.49\\
C2 & 600 & (20,20) & --- & 0.13 & --- & 97.00 & 100.00 & 79.00 & 98.54 & 99.73 & 90.57\\
C2 & 600 & (30,30) & --- & 0.13 & --- & 99.00 & 100.00 & 89.00 & 99.54 & 99.74 & 94.70\\
C2 & 600 & (40,40) & --- & 0.13 & --- & 100.00 & 100.00 & 86.00 & 100.00 & 99.74 & 93.49\\
C2 & 800 & (20,20) & --- & 0.10 & --- & 94.00 & 100.00 & 85.00 & 97.55 & 99.79 & 93.10\\
C2 & 800 & (30,30) & --- & 0.10 & --- & 99.00 & 100.00 & 85.00 & 99.53 & 99.81 & 92.92\\
C2 & 800 & (40,40) & --- & 0.09 & --- & 100.00 & 100.00 & 79.00 & 100.00 & 99.80 & 90.20\\
C2 & 1000 & (20,20) & --- & 0.08 & --- & 94.00 & 100.00 & 81.00 & 97.49 & 99.82 & 91.14\\
C2 & 1000 & (30,30) & --- & 0.07 & --- & 99.00 & 100.00 & 77.00 & 99.59 & 99.85 & 89.18\\
C2 & 1000 & (40,40) & --- & 0.07 & --- & 99.00 & 100.00 & 78.00 & 99.53 & 99.86 & 90.01\\
F1 & 600 & (20,20) & --- & --- & 8.13 & 100.00 & 97.00 & 69.00 & 100.00 & 98.75 & 79.68\\
F1 & 600 & (30,30) & --- & --- & 8.26 & 99.00 & 97.00 & 73.00 & 99.51 & 98.65 & 80.25\\
F1 & 600 & (40,40) & --- & --- & 6.51 & 100.00 & 98.00 & 72.00 & 100.00 & 99.19 & 80.81\\
F1 & 800 & (20,20) & --- & --- & 5.66 & 96.00 & 91.00 & 88.00 & 98.24 & 96.21 & 90.68\\
F1 & 800 & (30,30) & --- & --- & 5.76 & 96.00 & 95.00 & 90.00 & 98.14 & 97.68 & 89.92\\
F1 & 800 & (40,40) & --- & --- & 6.52 & 99.00 & 99.00 & 93.00 & 99.54 & 99.54 & 89.10\\
F1 & 1000 & (20,20) & --- & --- & 6.14 & 97.00 & 90.00 & 94.00 & 98.64 & 95.27 & 91.25\\
F1 & 1000 & (30,30) & --- & --- & 5.27 & 100.00 & 96.00 & 96.00 & 100.00 & 98.25 & 92.38\\
F1 & 1000 & (40,40) & --- & --- & 5.85 & 98.00 & 97.00 & 94.00 & 99.15 & 98.70 & 91.85\\
F2 & 600 & (20,20) & --- & --- & 4.80 & 94.00 & 91.00 & 64.00 & 97.45 & 95.99 & 84.24\\
F2 & 600 & (30,30) & --- & --- & 4.77 & 98.00 & 97.00 & 52.00 & 99.17 & 98.77 & 81.94\\
F2 & 600 & (40,40) & --- & --- & 5.03 & 99.00 & 96.00 & 59.00 & 99.55 & 98.24 & 82.56\\
F2 & 800 & (20,20) & --- & --- & 4.77 & 95.00 & 90.00 & 85.00 & 97.81 & 95.57 & 88.01\\
F2 & 800 & (30,30) & --- & --- & 4.22 & 100.00 & 95.00 & 86.00 & 100.00 & 97.96 & 88.95\\
F2 & 800 & (40,40) & --- & --- & 4.91 & 100.00 & 99.00 & 85.00 & 100.00 & 99.64 & 88.41\\
F2 & 1000 & (20,20) & --- & --- & 4.92 & 89.00 & 80.00 & 92.00 & 95.26 & 91.15 & 89.77\\
F2 & 1000 & (30,30) & --- & --- & 4.33 & 97.00 & 93.00 & 92.00 & 98.77 & 96.92 & 90.32\\
F2 & 1000 & (40,40) & --- & --- & 4.50 & 99.00 & 98.00 & 84.00 & 99.53 & 99.06 & 88.84\\
M1 & 800 & (30,30) & 0.10 & 0.09 & 5.10 & 100.00 & 99.00 & 92.00 & 99.85 & 99.70 & 90.79\\
M2 & 800 & (30,30) & 0.10 & --- & 4.28 & 100.00 & 92.00 & 71.00 & 99.83 & 96.55 & 86.83\\
M3 & 800 & (30,30) & 0.10 & 0.10 & 4.71 & 100.00 & 100.00 & 96.00 & 99.82 & 99.82 & 92.24\\
\bottomrule
\end{tabular}}
\end{table}

\subsubsection{Sensitivity to model misspecification}

We examine the sensitivity of our detection procedure when the data generating process does not conform to the MEFM model in~\eqref{eqn: MEFM_structure}, by considering the following additional scenarios:
\begin{itemize}
    \item[(i)] Matrix factor models with change points. This scenarios parallels the tensor factor model with change points considered by \cite{ZhangCenCho2026}. Each observed matrix has the form $\bX_t = \sum_{w = 0}^{Q_f} \bA_w \bF_t \bB_w^\trans \cdot \mathbb{I}\{\theta_{f, w} < t \le \theta_{f, w + 1}\} +\bE_t$ with $Q_f = 1$.
    All the relevant model parameters (apart from those specifying the change) are generated as described in Section~\ref{sec:setting} except that we simply set $\bQ = \bQ^*$ and $\bU = \bU^*$, such that the loading matrices are not restricted by the identifications condition in Assumption~\ref{ass: identification}. 
    We introduce one change to the column loading matrix at $\theta_{f, 1} = \lfloor 0.4T \rfloor$ similarly as in (S1) in Section~\ref{subsec:single}. This scenario is referred to as ``MFM''.
    
    \item[(ii)] Changes in the mean or/and covariance in vector time series, inspired by the experiments in \cite{cui2025simultaneous}.
    Each $\bX_t$ is obtained by folding a vector-valued observation $\bx_t \in \R^{pq}$ into a matrix, where $\bx_t = \bmu_t + \bSigma_t^{1/2} \be_t$, with $\be_t \stackrel{\text{i.i.d.}}{\sim} \cN(\zero, \bI_{pq})$, and the mean and/or the covariance undergo a single change at $\lfloor 0.4T \rfloor$ such that:
    \begin{itemize}
        \item[(a)] $\bmu_t = \bmu_0 \cdot \mathbb{I}\{t \geq \lfloor 0.4T \rfloor\}$ where $\bmu_0 \in \R^{pq}$ is a vector of zeros except for the randomly chosen 10\% of its elements which are set to be one.
        \item[(b)] For each $i,j\in[pq]$, $(\bSigma_t)_{ij} = (0.2^{|i-j|}) \cdot \mathbb{I}\{t <\lfloor 0.4T \rfloor\} + (0.8^{|i-j|}) \cdot \mathbb{I}\{t \geq \lfloor 0.4T \rfloor\}$.
    \end{itemize}
    We consider the scenarios where a single change point is present in ``mean only'' according to~(a), or ``covariance only'' according to~(b), as well as when both are present.
\end{itemize}
In summary, we consider four scenarios deviating from the MEFM framework, fixing $(T,p,q)=(200,10,10)$.
Figure~\ref{fig: sim_misspec} summarizes the results where we also include the competitors considered in Section~\ref{subsec:single}.

\begin{figure}[h!t!p!]
\centering
\includegraphics[width=1\linewidth]{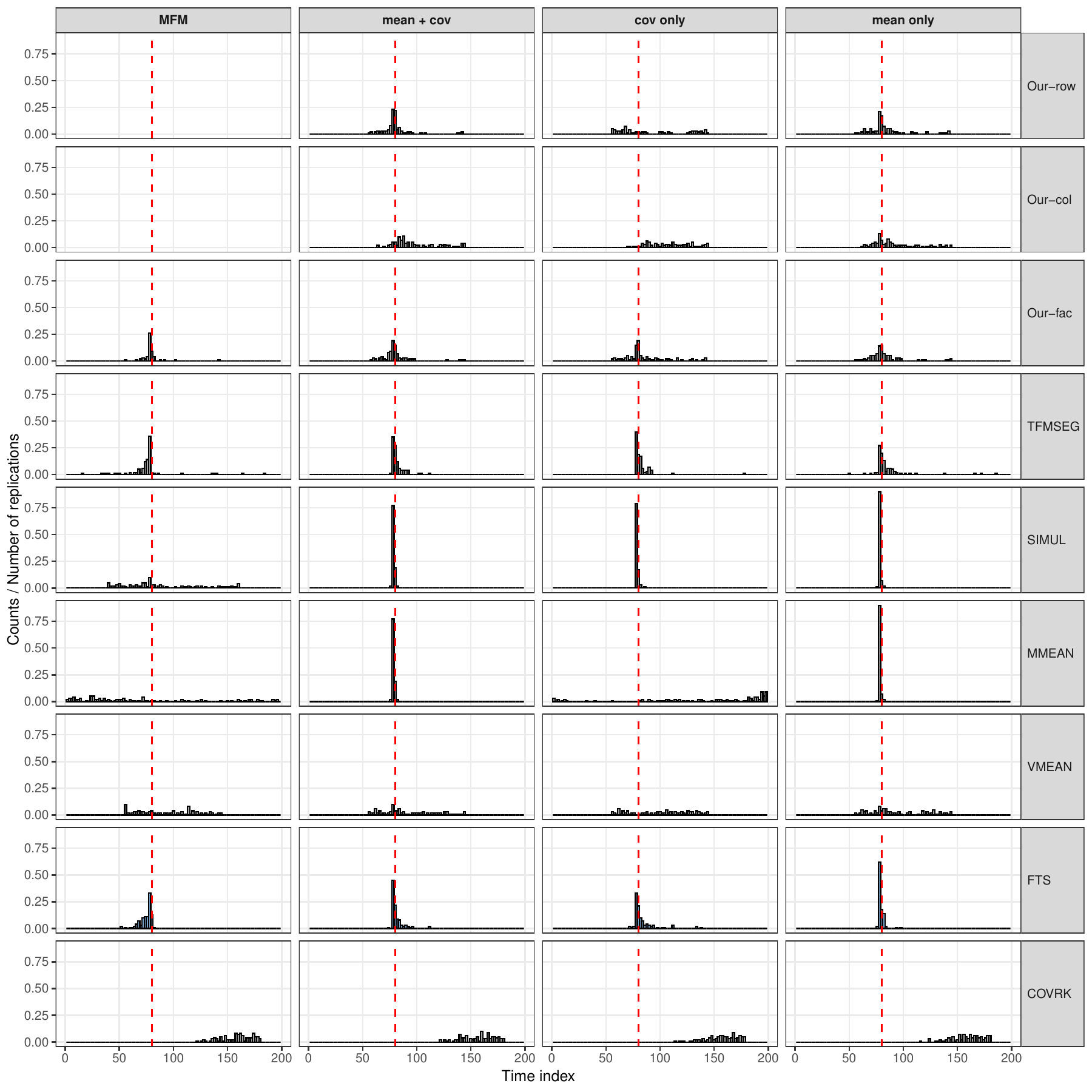}
\caption{
Histograms of detected change points for each method under four scenarios described in (i)--(ii): MFM, mean and covariance, covariance-only, and mean-only.
Top three panels contain the performance of our proposed method. The locations of the true change point is marked by the vertical dashed line. 
}
\label{fig: sim_misspec}
\end{figure}

Overall, our method simultaneously handles changes in the mean and covariance and correctly attributes the detected changes to the mean or covariance, even when the data do not conform to the MEFM framework, unlike the majority of the competitors.
Our-fac performs comparably to TFMSEG and FTS in the MFM scenario in~(i), noting that they are specifically designed to handle this type of change. 
Interestingly, as observed under (S2) in Section~\ref{subsec:single}, these methods tend to detect the changes in the mean even when the covariance is time-invariant. 
Comparing to the performance of MMEAN or SIMUL, Our-row and Our-col appear to suffer from the model mis-specification, they still perform more reliably than VMEAN dedicated to detecting mean shifts in vector time series. 
SIMUL performs well for the three scenarios in~(ii) but it cannot tell whether the change is due to a shift in the mean or covariance, or both, and does not appear to have power against the covariance change of type~(i).


\section{Application to New York City Yellow Taxi trip records}
\label{subsec: real_data_taxi}

We illustrate the proposed method using the New York City (NYC) Yellow Taxi Trip records.\footnote{Retrieved from \url{https://www.nyc.gov/site/tlc/about/tlc-trip-record-data.page}.}
Each raw trip record includes pick-up/drop-off timestamps and location codes. 
The Taxi and Limousine Commission's map partitions NYC into 265 zones, and we focus on the 69 zones on Manhattan Island which accounts for the vast majority of trip activity.
Let $\bX_t\in \R^{69\times 24}$ denote the matrix for calendar day $t$. Its $(i,j)$-th entry counts all trips with drop-off zone~$i$ and pick-up time in the $j$th hourly slot. We analyze the period from 1 January 2018 to 31 December 2022 and focus on the business-day data, which contains $T=1260$ days in total.

We apply the proposed method to detect any change in the first- and second-order structures of $\{ \bX_t \}_{t = 1}^T$, with the tuning parameters chosen according to the guide given in Section~\ref{subsec: tuning}, apart from the bandwidth. 
Many studies pointing out the association between road accidents and changes in sunrise and sunset times brought by daylight saving time (DST, \citealp{bunnings2021spring}).
This motivates us to set the bandwidth $\gamma = 80 \le  52 \times 5 \times 4 / 12 \approx 86.67$, such that changes in the taxi usage patterns related to the beginning and the end of DST, see also~\eqref{cond:thm: consistency_change_fac:two}.
On $\{\wh{\bL}_t\}_{t=1}^T$ computed as in~\eqref{eqn: Lt_hat}, the factor number estimator described in Section~\ref{subsec: tuning} gives $(\wh{k}_r, \wh{k}_c) = (2,2)$, which is the same under an alternative factor number estimator from \cite{ChenLam2024}.
The summary of the change point estimators detected by our method as well as the competing methods considered in the simulation studies (see Section~\ref{sec:setting}) is provided in Table~\ref{tab: NYC_competitors} and Figure~\ref{fig: NYC_log_vol2}. 

\begin{figure}[!h!t!p]
\centering
\includegraphics[width=1\linewidth]{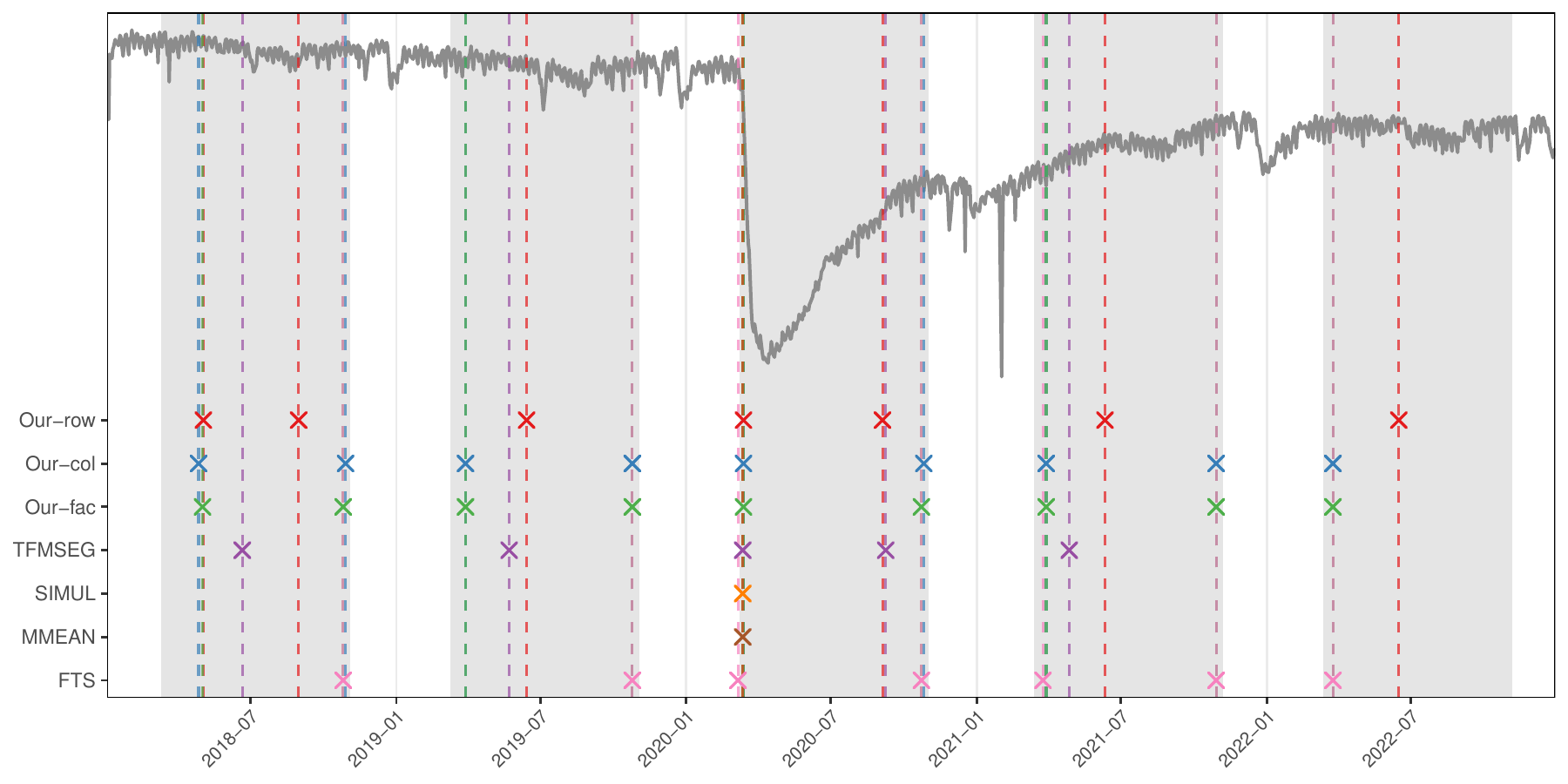}
\caption{
Log of the daily total number of trips in the NYC taxi data, i.e.,\ $\log(\sum_{i,j} (\bX_t)_{ij})$, $t=1,\dots,T$.
Vertical dashed lines with the crosses mark the change points detected by various methods. The shaded bands highlight when the daylight saving time is active in the US.
}
\label{fig: NYC_log_vol2}
\end{figure}


\begin{table}[!h!t]
\caption{Change points detected by various methods from the New York City Yellow Taxi data.}
\label{tab: NYC_competitors}
\centering
\begin{tabular}{lc}
\toprule
Method & Detected change points \\
\cmidrule(lr){1-1} \cmidrule(lr){2-2}
Our-row (zone) & 2018-05-03, 2018-08-31, 2019-06-14, 2020-03-13, 2020-09-04, \\
& 2021-06-11, 2022-06-16 \\
Our-col (hour) & 2018-04-27, 2018-10-29, 2019-03-29, 2019-10-25, 2020-03-13, \\
&  2020-10-26, 2021-03-29, 2021-10-29, 2022-03-25 \\
Our-fac & 2018-05-02, 2018-10-26, 2019-03-29, 2019-10-25, 2020-03-13, \\
& 2020-10-23, 2021-03-29, 2021-10-29, 2022-03-25 \\
\cmidrule(lr){1-1} \cmidrule(lr){2-2}
Our-fac (de-mean) & 2019-06-13, 2020-03-11, 2020-07-23, 2022-06-16 \\
\cmidrule(lr){1-1} \cmidrule(lr){2-2}
TFMSEG & 2018-06-21, 2019-05-23, 2020-03-12, 2020-09-08, 2021-04-27 \\
SIMUL  & 2020-03-12 \\
MMEAN  & 2020-03-12 \\
VMEAN  & --- \\
FTS    & 2018-10-26, 2019-10-25, 2020-03-06, 2020-10-23, 2021-03-25, \\
& 2021-10-29, 2022-03-25 \\
COVRK  & --- \\
\bottomrule
\end{tabular}
\end{table}

Some changes, such as those detected in March 2020 by all Our-row, Our-col and Our-fac as well as TFMSEG, SIMUL, MMEAN and FTS are closely related to the Covid-19 lockdown, which brought changes both to the volume of the trips as well as variability. Another commonly detected change around October 2020, March/April and later in 2021 may be related to the phased re-opening of activities as the restrictions eased off, and vaccine rollout. 
We notice that SIMUL is designed to test for a change and as such, it only detects a single, most prominent change.
We notice an interesting alignment between the change points detected by Our-col (from the hour effects) and DST, which is also striking from Figure~\ref{fig: NYC_hour_plot_sel2} visualizing the estimated column main effects at selected hours. 
In the analysis of a similar dataset, \cite{ZhangCenCho2026} pre-process the data by de-meaning the raw trip records via rolling window. The outputs returned by Our-row and Our-col reveal that such an approach may omit to mark large shifts amid gradual changes in the mean, which offer an interesting interpretation of the relationship between the operation of DST and taxi usage patterns.

At the same time, we note the potential limitation of the MEFM framework, as the strong agreement between the changes detected by Our-fac and Our-col may be due to the model mis-specification by which some mean shifts may remain even after the step in~\eqref{eqn: Lt_hat}. 
We thus analyze the residuals from de-meaning the individual series $(\bX_t)_{ij}$ over each segment determined by the successive elements of $\wh\Theta_\alpha \cup \wh\Theta_\beta$, for the change points in the factor structure via Our-fac; see Table~\ref{tab: NYC_competitors} for the results. 
The reduced number of change points detected from the residual time series suggests that an additional de-meaning step can act as a safeguard against the mode mis-specification, which we leave for future research.

\begin{figure}[!h!t!p]
\centering
\includegraphics[width=0.9\linewidth]{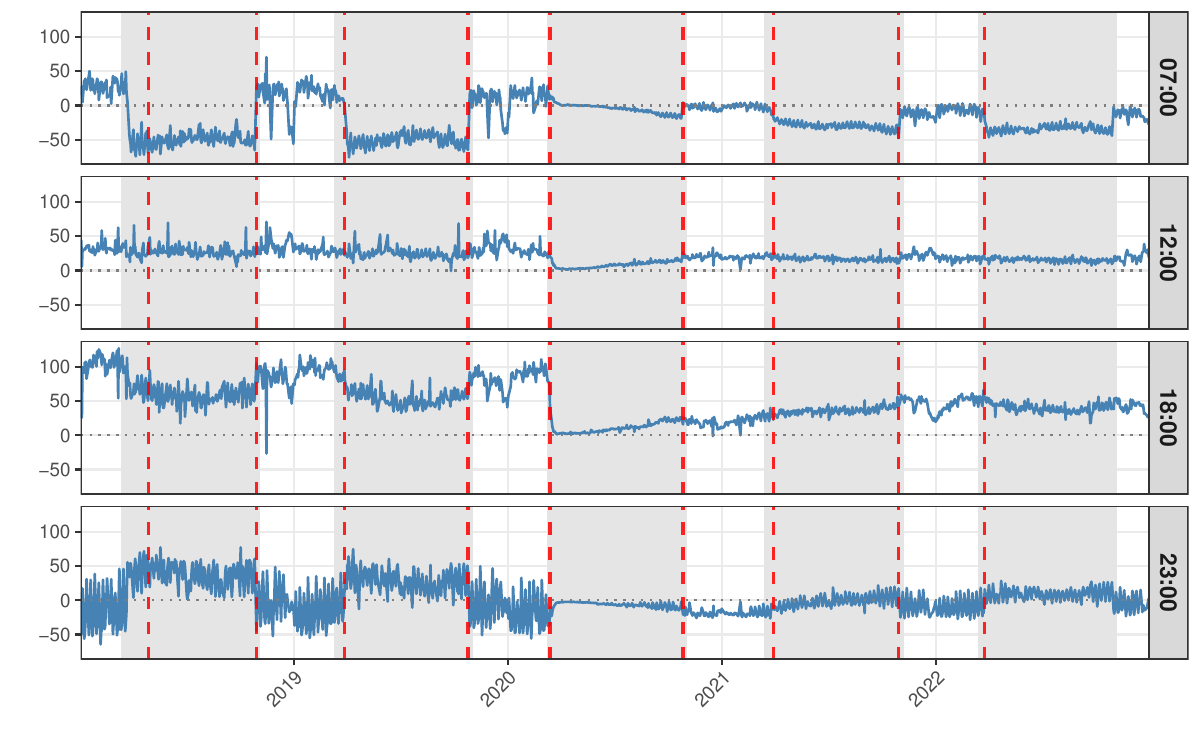}
\caption{
Selected plots of estimated column (hour) main effects $\wh\bbeta_t, \, 1 \le t \le T$, for the NYC taxi data, with vertical dashed red lines marking the change points returned by Our-col. Shaded bands highlight when the daylight saving time is active in the US.
Each hour represents morning rush hour (07:00), business hour (12:00), evening rush hour (18:00), and night life (23:00), respectively, according to \cite{chen2022factor}.
}
\label{fig: NYC_hour_plot_sel2}
\end{figure}

\bibliographystyle{apalike}
\bibliography{reference.bib}	


\clearpage
\appendix
\renewcommand{\thesection}{\Alph{section}}
\renewcommand{\theHsection}{\Alph{section}} 
\renewcommand{\thesubsection}{\Alph{section}.\arabic{subsection}}
\renewcommand{\theHsubsection}{\Alph{section}.\arabic{subsection}} 

\setcounter{section}{0}
\setcounter{equation}{0}
\def\theequation{\Alph{section}.\arabic{equation}}
\def\thetheorem{\Alph{section}.\arabic{theorem}}
\def\thelemma{\Alph{section}.\arabic{lemma}}
\def\thefigure{\Alph{section}.\arabic{figure}}
\def\thedefinition{\Alph{section}.\arabic{definition}}
\def\theexample{\Alph{section}.\arabic{example}}
\def\theremark{\Alph{section}.\arabic{remark}}
\def\thetable{\Alph{section}.\arabic{table}}
\def\thecorollary{\Alph{section}.\arabic{corollary}}

\section{Proof of theorems and auxiliary results}

For a symmetric matrix $\bX \in\R^{a\times a}$, given the half-vectorization $\Vech(\bX): \R^a \mapsto\R^{a(a+1)/2}$,
the matrix $\L_a \in \R^{a(a+1)/2 \times a^2}$ is such that $\Vech(\bX) = \L_a \vec(\bX)$, and the matrix $\K_a \in \R^{a^2\times a(a+1)/2}$ is such that $\vec(\bX) = \K_a \Vech(\bX)$.

\subsection{Auxiliary results for the proof of Theorem~\ref{thm: consistency_proj}}


We will use Lemma~B.3 and Lemma~B.4 in \cite{ZhangCenCho2026}.
For completeness, we directly quote (the adapted versions of) them respectively as Lemmas~\ref{lemma: FMF_expectation} and \ref{lemma: inequality_sandwich}. Additionally, we construct some new results in Lemma~\ref{lemma: inequality_sandwich}~(i) as an extension, with the proof attached right after the lemma.

\begin{lemma}\label{lemma: FMF_expectation}
Let Assumption~\ref{ass: core_factor} hold. Then for any $t\in[T]$ and constant matrices $\bW_1 \in\R^{k_c\times k_c}$, $\bW_2\in\R^{k_r\times k_r}$, we have $\E\left( \bF_t \bW_1 \bF_t^\trans \right) = \tr(\bW_1) \, \bI_{k_r}$ and $\E\left( \bF_t^\trans \bW_2 \bF_t \right) = \tr(\bW_2) \, \bI_{k_c}$.
For any interval of timestamps $\cS$, if it further holds that $\|\bM_1\|_F^2, \|\bM_2\|_F^2= o(|\cS|)$, then
\[
\frac{1}{|\cS|} \sum_{t\in\cS} \bF_t \bM_1 \bF_t^\trans \xrightarrow{P} \tr(\bM_1) \, \bI_{k_r}, \quad
\frac{1}{|\cS|} \sum_{t\in\cS} \bF_t^\trans \bM_2 \bF_t \xrightarrow{P} \tr(\bM_2) \, \bI_{k_c} .
\]
\end{lemma}

\begin{lemma}\label{lemma: inequality_sandwich}
\begin{enumerate}[itemsep=0pt, label = (\roman*), left = 0pt]
    \item \footnote{Note that all inequalities in part~(i) are tight, which can be seen by taking $\bM$ as the identity matrix.} For any square matrix $\bM$ and sequence of matrices $\bA_t, \, t \in [T]$, of compatible dimensions, we have
    \[
    \left\| \sum_{t=1}^T \bA_t \bM \bA_t^\trans \right\|_F \leq \|\bM\| \cdot \left\| \sum_{t=1}^T \bA_t \bA_t^\trans \right\|_F .
    \]
    In addition, we also have
    \[
    \left\| \sum_{t=1}^T \bA_t \bM \bA_t^\trans \right\| \leq \|\bM\| \cdot \left\| \sum_{t=1}^T \bA_t \bA_t^\trans \right\| 
    \text{\ and \ }
    \left\| \sum_{t=1}^T \E(\bA_t \bM \bA_t^\trans) \right\| \leq \E\|\bM\| \cdot \left\| \sum_{t=1}^T \E(\bA_t \bA_t^\trans) \right\| ,
    \]
    where $\bM$ in the second inequality is independent of $\{\bA_t\}$.
    \item (Lemma~B.16 in \cite{Barigozzietal2025_robust}) For any sequence of matrices $\bA_t=[A_{t,ij}]$, $\bB_t$ for $t\in[T]$, and any matrix $\bM$ of compatible dimensions, we have
    \[
    \left\| \sum_{t=1}^T \bA_t \bM \bB_t^\trans \right\|_F^2 \leq \|\bM\|_F^2 \cdot \sum_{i,j} \left\| \sum_{t=1}^T A_{t,ij} \bB_t \right\|_F^2 .
    \]
\end{enumerate}
\end{lemma}

\begin{proof}[Proof of \Cref{lemma: inequality_sandwich}]
As mentioned previously, all the results are from \cite{ZhangCenCho2026} except for the second display in part~(i). 
For clarity, all the defined notations in the proof of this lemma is used only within this proof.
To show the additional results, let $\bB := \sum_{t=1}^T \bA_t \bA_t^\trans$.  Then for any unit vectors $\bu$ and $\bv$ (with compatible dimensions), by the Cauchy--Schwarz inequality we have
\begin{align*}
    &\;\quad
    \left| \bu^\top \left( \sum_{t=1}^T \bA_t \bM \bA_t^\trans \right) \bv \right| = \left| \sum_{t=1}^T (\bA_t^\trans \bu)^\trans \bM (\bA_t^\trans \bv) \right|
    =: \left| \sum_{t=1}^T \bx_t^\trans \bM \by_t \right| \\
    &\leq
    \sum_{t=1}^T \left| \bx_t^\trans \bM \by_t \right|
    \leq
    \|\bM\| \sum_{t=1}^T \|\bx_t\| \|\by_t\| \\ 
    &\leq
    \|\bM\| \sqrt{\left( \sum_{t=1}^T \|\bx_t\|^2 \right) \left( \sum_{t=1}^T \|\by_t\|^2 \right)}
    \leq \|\bM\| \sqrt{\|\bB\| \|\bB\|}
    = \|\bM\| \|\bB\|,
\end{align*}
which gives the first additional result by taking supremum over all unit vectors. For the second result, we use a similar argument. i,e.,\ for any unit vectors $\bu$ and $\bv$ (say the same as above, and hence the same definitions of $\bx_t$ and $\by_t$), by the Jensen's inequality and the Cauchy--Schwarz inequality, it holds that
\begin{align*}
    &\;\quad
    \left| \bu^\top \E\left( \sum_{t=1}^T \bA_t \bM \bA_t^\trans \right) \bv \right| \leq \E\left( \left| \sum_{t=1}^T (\bA_t^\trans \bu)^\trans \bM (\bA_t^\trans \bv) \right| \right) 
    = \E\left( \left| \sum_{t=1}^T \bx_t^\trans \bM \by_t \right| \right) \\
    &\leq
    \sum_{t=1}^T \E\left( \left| \bx_t^\trans \bM \by_t \right| \right)
    \leq
    \E( \|\bM\| ) \cdot \E\left( \sum_{t=1}^T  \|\bx_t\| \|\by_t\| \right) \\ 
    &\leq
    \E( \|\bM\| ) \cdot \E\left( \sqrt{\sum_{t=1}^T \|\bx_t\|^2} \sqrt{\sum_{t=1}^T \|\by_t\|^2} \right)
    \leq
    \E( \|\bM\| ) \cdot \sqrt{ \E\left( \sum_{t=1}^T \|\bx_t\|^2 \right) \cdot \E\left( \sum_{t=1}^T \|\by_t\|^2 \right) } \\
    &\leq 
    \E( \|\bM\| ) \cdot \sqrt{ \|\E(\bB)\| \cdot \|\E(\bB)\|}
    = \E( \|\bM\| ) \cdot \|\E(\bB)\|,
\end{align*}
where the last line used the fact that
\[
\E\left( \sum_{t=1}^T \|\bx_t\|^2 \right) = \bu^\trans \sum_{t=1}^T \E(\bA_t \bA_t^\trans) \bu \leq \|\E(\bB)\| ,
\]
and similarly for $\E\left( \sum_{t=1}^T \|\by_t\|^2 \right)$. With supremum taken over all unit vectors $\bu$ and $\bv$, we conclude the second additional result and therefore the proof of this lemma.
\end{proof}

We also state Lemma~\ref{lemma: rate_corollary_lemma2_LamCen} as follows, which is direct from the proof of Lemma~2 in \cite{LamCen2025_MEFM}, with all the required conditions included in our Assumption~\ref{ass: noise}.

\begin{lemma}\label{lemma: rate_corollary_lemma2_LamCen}
Under Assumption~\ref{ass: noise}, it holds that
\begin{align*}
    & \left\| \sum_{t=1}^T \bE_t\bE_t^\trans \right\|_F^2 = \cO_P(T p^2 q + T^2 p q^2) ,\\
    & \left\| \sum_{t=1}^T \{ \bE_t\bE_t^\trans - \E(\bE_t\bE_t^\trans) \} \right\|_F^2 = \cO_P(T p^2 q), \\
    & \left\| \sum_{t=1}^T \{ \1_q^\trans \bE_t^\trans \1_p \bE_t \1_q \1_p^\trans - \E(\1_q^\trans \bE_t^\trans \1_p \bE_t \1_q \1_p^\trans) \} \right\|_F^2 =
    \cO_P(T p^3 q^2), \\
    & \left\| \sum_{t=1}^T \{ \bE_t \bE_t^\trans \1_p \1_p^\trans - \E(\bE_t \bE_t^\trans \1_p \1_p^\trans) \} \right\|_F^2 = \cO_P(Tp^4q), \\
    & \left\| \sum_{t=1}^T \{ \bE_t\1_q\1_q^\trans \bE_t^\trans - \E(\bE_t\1_q\1_q^\trans \bE_t^\trans) \} \right\|_F^2 = \cO_P(T p^2 q^2), \\
    & \left\| \sum_{t=1}^T \{(\1_q^\trans \bE_t^\trans \1_p)^2 - \E((\1_q^\trans \bE_t^\trans \1_p)^2)\} \cdot \1_p \1_p^\trans \right\|_F^2 = \cO_P( T^2 p^4 q^2 ), \\
    & \left\| \sum_{t=1}^T \{ \1_p^\trans \bE_t \bE_t^\trans \1_p - \E(\1_p^\trans \bE_t \bE_t^\trans \1_p) \} \cdot \1_p \1_p^\trans \right\|_F^2 = \cO_P(T p^6 q) .
\end{align*}
\end{lemma}

\begin{lemma}\label{lemma: eigenvalue_PCA_consistency}
\begin{enumerate}[itemsep=0pt, label = (\roman*), left = 0pt]
    \item Under Assumption~\ref{ass: change_loc}, the matrices $\bSigma_{R}$ and $\bSigma_{V}$ are positive definite.
    \item Let Assumptions~\ref{ass: loadings}, \ref{ass: core_factor}, \ref{ass: noise}, and \ref{ass: change_loc} hold.
    Denote by $\wt{\bD}_r := \wt\bQ^\trans \left\{ (Tpq)^{-1} \sum_{t=1}^T \wh\bL_t \wh\bL_t^\trans \right\} \wt\bQ /p$ the diagonal matrix consisting of the leading $k_r$ eigenvalues of $(Tpq)^{-1} \sum_{t=1}^T \wh\bL_t \wh\bL_t^\trans$,
    and $\wt{\bD}_c := \wt\bU^\trans \left\{(Tpq)^{-1} \sum_{t=1}^T \wh\bL_t^\trans \wh\bL_t \right\} \wt\bU /q$ the diagonal matrix consisting of the leading $k_c$ eigenvalues of $(Tpq)^{-1} \sum_{t=1}^T \wh\bL_t^\trans \wh\bL_t$.
    Then as $\min\{T,p,q\} \to \infty$, the eigenvalues of $\wt{\bD}_r$ satisfy $\lambda_i(\wt{\bD}_r) = \lambda_i(\bSigma_{R}) + o_P(1)$ for any $i\leq k_r$,
    and similarly, $\lambda_j(\wt{\bD}_c) = \lambda_j(\bSigma_{V}) + o_P(1)$ for any $j\leq k_c$.
\end{enumerate}
\end{lemma}

\begin{proof}[Proof of \Cref{lemma: eigenvalue_PCA_consistency}]
Part~(i) is proved by following exactly the same arguments as those given in the proof of Lemma~B.6~(i) in \cite{ZhangCenCho2026}.

For part~(ii) of the lemma, it suffices to show the result for $\wt\bD_r$ as that for $\wt\bD_c$ is similar. We first find the eigenvalues of $(Tpq)^{-1} \sum_{t=1}^T \bC_t\bC_t^\trans$. By \eqref{eqn: Gt_define},
\begin{align}
    &\;\quad
    \frac{1}{Tpq} \sum_{t=1}^T \bC_t\bC_t^\trans
    =
    \frac{1}{Tpq} \sum_{t=1}^T \sum_{w=1}^{Q_f +1} \bQ \bR_{w-1} \bF_t \bV_{w-1}^\trans \bU^\trans \bU \bV_{w-1} \bF_t^\trans \bR_{w-1}^\trans \bQ^\trans \cdot \mathbb{I}\{\theta_{f,w-1} <t\le \theta_{f,w}\} \notag \\
    &=
    \frac{1}{Tp} \sum_{w=1}^{Q_f +1} \sum_{t=\theta_{f,w-1} +1}^{\theta_{f,w}} \bQ \bR_{w-1} \bF_t \bV_{w-1}^\trans (\bU^\trans \bU/q) \bV_{w-1} \bF_t^\trans \bR_{w-1}^\trans \bQ^\trans \notag \\
    &\xrightarrow{P}
    \frac{1}{p} \bQ \left\{ \frac{1}{T} \sum_{w=1}^{Q_f +1} \|\bV_{w-1} \|_F^2 \cdot (\theta_{f,w} - \theta_{f,w-1}) \cdot \bR_{w-1} \bR_{w-1}^\trans \right\} \bQ^\trans
    \to
    \frac{1}{p} \bQ \bSigma_{R} \bQ^\trans ,
    \label{eqn: CC_limit}
\end{align}
where the last line used Assumptions~\ref{ass: loadings}, \ref{ass: core_factor}, \ref{ass: change_loc}~(iii), and Lemma~\ref{lemma: FMF_expectation}. Then with Assumption~\ref{ass: loadings}, the leading $k_r$ eigenvalues of $(Tpq)^{-1} \sum_{t=1}^T \bC_t\bC_t^\trans$ are asymptotically equal to those of $\bSigma_{R}$, and hence are all bounded away from zero and infinity by part~(i) of this lemma.

Next, recall the definition of $\wh\bL_t$ from \eqref{eqn: Lt_hat} and the definition of matrix $\bM_m$ in Section~\ref{sec: introduction}, we can write
\begin{align*}
    \wh\bL_t &= \bX_t - \wh\mu_t\1_p \1_q^\trans - \wh\balpha_t \1_q^\trans - \1_p \wh\bbeta_t^\trans \\
    &=
    \bM_p \bX_t \bM_q + \frac{1}{p} \1_p \1_p^\trans \bX_t + \frac{1}{pq} \1_p \1_p^\trans \bX_t \1_q \1_q^\trans
    - \wh\mu_t\1_p \1_q^\trans - \wh\balpha_t \1_q^\trans - \1_p \wh\bbeta_t^\trans \\
    &=
    \bM_p \bX_t \bM_q + \frac{1}{p} \1_p \1_p^\trans \bX_t + \frac{1}{pq} \1_p \1_p^\trans \bX_t \1_q \1_q^\trans \\
    &\;\quad
    - \wh\mu_t\1_p \1_q^\trans - \left(\frac{1}{q} \bX_t \1_q - \wh\mu_t \1_p\right) \1_q^\trans - \1_p \left(\frac{1}{p} \bX_t^\trans \1_p - \wh\mu_t \1_q\right)^\trans
    = \bM_p \bX_t \bM_q .
\end{align*}
Together with the fact that $\bM_m \1_m =\0$, by the model~\eqref{eqn: MEFM_structure} we have
\begin{align}
    \wh\bL_t \wh\bL_t^\trans &= \bM_p \bX_t \bM_q \bX_t^\trans \bM_p \notag 
    = \bM_p (\bC_t + \bE_t) \bM_q (\bC_t^\trans + \bE_t^\trans) \bM_p \\
    &=
    \bC_t\bC_t^\trans + \bC_t\bE_t^\trans + \bE_t\bC_t^\trans + \bE_t\bE_t^\trans + (pq)^{-1}\bE_t\1_q\1_q^\trans \bE_t^\trans\1_p\1_p^\trans + (pq)^{-1}\1_p\1_p^\trans\bE_t\1_q\1_q^\trans \bE_t^\trans \notag \\ 
    &\;\quad
    - p^{-1}\bC_t\bE_t^\trans \1_p\1_p^\trans - p^{-1}\1_p \1_p^\trans \bE_t\bC_t^\trans - p^{-1}\bE_t \bE_t^\trans \1_p\1_p^\trans - p^{-1}\1_p \1_p^\trans \bE_t\bE_t^\trans - q^{-1}\bE_t \1_q\1_q^\trans \bE_t^\trans \notag \\
    &\;\quad
    + p^{-2}\1_p\1_p^\trans \bE_t\bE_t^\trans \1_p\1_p^\trans - (pq)^{-1}p^{-1} \1_p\1_p^\trans\bE_t\1_q\1_q^\trans\bE_t^\trans\1_p\1_p^\trans
    =: 
    \bC_t\bC_t^\trans + \mathcal{R}_{r,t} .
    \label{eqn: sample_covariance_row}
\end{align}

Next, we consider $\sum_{t=1}^T \mathcal{R}_{r,t}$.
Lemma~\ref{lemma: rate_corollary_lemma2_LamCen} handles all the terms only involving $\bE_t$.
All the remaining terms, i.e.,\ $\sum_{t=1}^T \bC_t\bE_t^\trans$, $\sum_{t=1}^T \bE_t\bC_t^\trans$, $\sum_{t=1}^T p^{-1}\bC_t\bE_t^\trans \1_p\1_p^\trans$, and $\sum_{t=1}^T p^{-1}\1_p \1_p^\trans \bE_t\bC_t^\trans$  have their Frobenius/spectral norm be upper bounded by that of $\sum_{t=1}^T \bC_t\bE_t^\trans$ by sub-multiplicativity of those norms.

For $\sum_{t=1}^T \bC_t \bE_t^\trans$, note that by \eqref{eqn: cov_change} we have
\[
\bC_t \bE_t^\trans = \sum_{w=1}^{Q_f +1} \left( \bA_{w-1}\bF_t\bB_{w-1}^\trans \bE_t^\trans \right) \cdot \mathbb{I}\{\theta_{f,w-1} <t\le \theta_{f,w}\} ,
\]
thus it is possible to consider $\bA_{w-1}\bF_t\bB_{w-1}^\trans \bE_t^\trans$ for each segment $(\theta_{f,w-1}, \theta_{f,w}]$, $w\in[Q_f +1]$. Then by Assumption~\ref{ass: core_factor} and similarly applying Lemma~1 in \cite{LamCen2025_MEFM} as in their proof of Lemma~2, we conclude that for any $w\in[Q_f +1]$,
\begin{align}
    &\;\quad
    \left\| \sum_{t=\theta_{f,w-1}+1}^{\theta_{f,w}} \bA_{w-1} \bF_t \bB_{w-1}^\trans \bE_t^\trans \right\|_F^2 
    \lesssim_P (\theta_{f,w} - \theta_{f,w-1})pq \cdot \|\bA_{w-1}\|_F^2 \notag \\
    &\leq
    (\theta_{f,w} - \theta_{f,w-1})pq \cdot \|\bQ\|_F^2 \cdot \|\bR_{w-1}\|_F^2 = \cO_P\left\{ (\theta_{f,w} - \theta_{f,w-1}) p^2q \right\},
    \label{eqn: AFBE_segment}
\end{align}
where the last equality used Assumption~\ref{ass: loadings} and Assumption~\ref{ass: change_loc}~(iii). Thus,
\begin{equation}
\label{eqn: CE_bound}
\left\| \sum_{t=1}^T \bC_t \bE_t^\trans \right\|_F^2 \leq \sum_{w=1}^{Q_f +1} \left\| \sum_{t=\theta_{f,w-1}+1}^{\theta_{f,w}} \bA_{w-1} \bF_t \bB_{w-1}^\trans \bE_t^\trans \right\|_F^2 = \cO_P( T p^2 q ).
\end{equation}

Combining all the previous arguments, we arrive at
\begin{equation}
\label{eqn: rate_cR_t}
\left\| \sum_{t=1}^T \mathcal{R}_{r,t} \right\|_F^2 = \cO_P\left( Tp^2q + T^2pq^2 \right).
\end{equation}
Then with \eqref{eqn: sample_covariance_row}, it holds that
\begin{align}
    \frac{1}{Tpq} \sum_{t=1}^T \wh\bL_t \wh\bL_t^\trans = \frac{1}{Tpq} \sum_{t=1}^T \bC_t\bC_t^\trans + \frac{1}{Tpq} \sum_{t=1}^T \mathcal{R}_{r,t} ,
    \label{eqn: sample_covariance_row_sum}
\end{align}
which, together with \eqref{eqn: CC_limit} and the argument below it and applying Weyl's inequality iteratively, implies that the first $k_r$ eigenvalues of $(Tpq)^{-1} \sum_{t=1}^T \wh\bL_t \wh\bL_t^\trans$ are led by those of $(Tpq)^{-1} \sum_{t=1}^T \bC_t\bC_t^\trans$. This completes the proof of the lemma.
\end{proof}

\begin{lemma}\label{lemma: consistency_PCA}
Let all the assumptions in Lemma~\ref{lemma: eigenvalue_PCA_consistency}~(ii) hold. 
With $\wt\bD_r$ and $\wt\bD_c$ from \Cref{lemma: eigenvalue_PCA_consistency}~(ii), define
\[
\wt\bH_r =\frac{1}{Tpq} \left(\sum_{t=1}^T \bG_t \bU^\trans \bU \bG_t^\trans \right) \bQ^\trans \wt\bQ \wt\bD_r^{-1} 
\text{ \ and \ }
\wt\bH_c =\frac{1}{Tpq} \left(\sum_{t=1}^T \bG_t^\trans \bQ^\trans \bQ \bG_t \right) \bU^\trans \wt\bU \wt\bD_c^{-1} .
\]
Then $\wt\bH_r$ and $\wt\bH_c$ are asymptotically invertible as $\min\{T,p,q\} \to \infty$ and, moreover, respectively satisfies $\wt\bH_r \wt\bH_r^\trans = \bI_{k_r} + o_P(1)$ and $\wt\bH_c \wt\bH_c^\trans = \bI_{k_c} + o_P(1)$. The initial estimators $\wt\bQ$ and $\wt\bU$ satisfy
\[
\frac{1}{p}\left\| \wt\bQ - \bQ \wt\bH_r \right\|_F^2 = \cO_P\Big( \frac{1}{Tq} + \frac{1}{p^2} \Big)  
\text{ \ and \ }
\frac{1}{q}\left\| \wt\bU - \bU \wt\bH_c \right\|_F^2 = \cO_P\Big( \frac{1}{Tp} + \frac{1}{q^2} \Big) .
\]
\end{lemma}

\begin{proof}[Proof of \Cref{lemma: consistency_PCA}]
We will show the results on $\wt\bH_r$ and $\wt\bQ$, and those of $\wt\bH_c$ and $\wt\bU$ follow the similar arguments. As $\wt\bD_r$ is asymptotically invertible from Lemma~\ref{lemma: eigenvalue_PCA_consistency}~(ii), by definitions we can write $\wt\bQ =(Tpq)^{-1} \sum_{t=1}^T \wh\bL_t \wh\bL_t^\trans \wt\bQ \wt\bD_r^{-1}$. Then using the decomposition in \eqref{eqn: sample_covariance_row_sum},
\[
\wt\bQ - \bQ \wt\bH_r = \frac{1}{Tpq} \sum_{t=1}^T \mathcal{R}_{r,t} \wt\bQ \wt\bD_r^{-1} .
\]
Given the result in \eqref{eqn: rate_cR_t} and the fact that Assumption~\ref{ass: noise} being the same as Assumption~3 in \cite{ZhangCenCho2026}, all the remaining arguments follow the same (in the case of order-2 tensor time series, i.e.,\ matrix time series) as those given in the proof of Lemma~B.7 of \cite{ZhangCenCho2026}. This thus concludes the proof of this lemma.
\end{proof}

\begin{lemma}\label{lemma: population_cov_PCA_consistency}
Let all the assumptions in Lemma~\ref{lemma: eigenvalue_PCA_consistency}~(ii) hold.
Define
\begin{equation*}
\begin{split}
    \Ddot{\bGamma}_r &:= \frac{1}{Tpq} \sum_{t=1}^T \bC_t\bC_t^\trans + \frac{1}{Tpq} \sum_{t=1}^T \bM_p \E\left( \bE_t \bM_q \bE_t^\trans \right) \bM_p
    =: \frac{1}{Tpq} \sum_{t=1}^T \bC_t\bC_t^\trans + {\bGamma}_{E,r} , \\
    \Ddot{\bGamma}_c &:= \frac{1}{Tpq} \sum_{t=1}^T \bC_t^\trans \bC_t + \frac{1}{Tpq} \sum_{t=1}^T \bM_q \E\left( \bE_t^\trans \bM_p \bE_t \right) \bM_q
    =: \frac{1}{Tpq} \sum_{t=1}^T \bC_t^\trans \bC_t + {\bGamma}_{E,c} .
\end{split}
\end{equation*}
Let $\Ddot\bQ$ be $\sqrt{p}$ times the $k_r$ leading eigenvectors of $\Ddot{\bGamma}_r$, and $\Ddot\bU$ be $\sqrt{q}$ times the $k_c$ leading eigenvectors of $\Ddot{\bGamma}_c$. As $\min\{T,p,q\} \to \infty$, there exists orthogonal matrices $\bJ_r\in \R^{k_r\times k_r}$ $\bJ_c\in \R^{k_c\times k_c}$ such that
\[
\frac{1}{p}\left\| \wt\bQ - \Ddot\bQ \bJ_r \right\|_F^2 = \cO_P\Big( \frac{1}{Tq} + \frac{1}{p^2q^2} \Big)  
\text{ \ and \ }
\frac{1}{q}\left\| \wt\bU - \Ddot\bU \bJ_c \right\|_F^2 = \cO_P\Big( \frac{1}{Tp} + \frac{1}{p^2q^2} \Big) .
\]
\end{lemma}

\begin{proof}[Proof of Lemma~\ref{lemma: population_cov_PCA_consistency}]
We show the results for $\wt\bQ$ and those for $\wt\bU$ can be proved by analogous steps. In particular, we adapt Lemma~B.4~(i) of \cite{Barigozzietal2025_robust}, by setting our $(Tq)^{-1} \sum_{t=1}^T \wh\bL_t \wh\bL_t^\trans$, $p \Ddot{\bGamma}_r$, and $(Tq)^{-1} \sum_{t=1}^T \bC_t\bC_t^\trans$ respectively as their $\wh{\mathbf{S}}$, $\wt{\mathbf{S}}$, and $\mathbf{S}$. Note that the eigenvalues required in condition~(C1), Lemma~B.4~(i) of \cite{Barigozzietal2025_robust}, need not be necessarily distinct, in which case the diagonal matrix $\wt{\mathbf{J}}$ therein is replaced with some orthogonal matrix, say matrix $\bJ_r$. Hence it suffices to show the first $k_r$ eigenvalues of $(Tq)^{-1} \sum_{t=1}^T \bC_t\bC_t^\trans$ multiplied by $p^{-1}$ are bounded away from zero and infinity, which is true from the argument below \eqref{eqn: CC_limit}.

To verify condition (C2) in Lemma~B.4 of \cite{Barigozzietal2025_robust}, by the decomposition \eqref{eqn: sample_covariance_row} and triangle inequality we have
\begin{align*}
    &\;\quad
    \frac{1}{p^2} \left\| \frac{1}{Tq} \sum_{t=1}^T \wh\bL_t \wh\bL_t^\trans - p \Ddot{\bGamma}_r \right\|_F^2 \\
    &\leq
    \left\| \frac{1}{Tpq} \sum_{t=1}^T \bC_t\bE_t^\trans \right\|_F^2 
    + \left\| \frac{1}{Tpq} \sum_{t=1}^T \bE_t\bC_t^\trans \right\|_F^2 
    + \left\| \frac{1}{Tp^2q} \sum_{t=1}^T \bC_t\bE_t^\trans \1_p\1_p^\trans \right\|_F^2
    + \left\| \frac{1}{Tp^2q} \sum_{t=1}^T \1_p \1_p^\trans \bE_t\bC_t^\trans \right\|_F^2 \\
    &\;\quad
    + \frac{1}{T^2 p^2 q^2} \left\| \sum_{t=1}^T \{ \bE_t\bE_t^\trans - \E(\bE_t\bE_t^\trans) \} \right\|_F^2 
    + \frac{1}{T^2 p^4 q^4} \left\| \sum_{t=1}^T \{ \1_q^\trans \bE_t^\trans \1_p \bE_t \1_q \1_p^\trans - \E(\1_q^\trans \bE_t^\trans \1_p \bE_t \1_q \1_p^\trans) \} \right\|_F^2 \\
    &\;\quad
    + \frac{1}{T^2 p^4 q^2} \left\| \sum_{t=1}^T \{ \bE_t \bE_t^\trans \1_p \1_p^\trans - \E(\bE_t \bE_t^\trans \1_p \1_p^\trans) \} \right\|_F^2
    + \frac{1}{T^2 p^2 q^4} \left\| \sum_{t=1}^T \{ \bE_t\1_q\1_q^\trans \bE_t^\trans - \E(\bE_t\1_q\1_q^\trans \bE_t^\trans) \} \right\|_F^2 \\
    &\;\quad
    + \frac{1}{T^2 p^6 q^2} \left\| \sum_{t=1}^T \{ \1_p^\trans \bE_t \bE_t^\trans \1_p - \E(\1_p^\trans \bE_t \bE_t^\trans \1_p) \} \cdot \1_p \1_p^\trans \right\|_F^2 \\
    &\;\quad
    + \frac{1}{T^2 p^6 q^4} \left\| \sum_{t=1}^T \{(\1_q^\trans \bE_t^\trans \1_p)^2 - \E((\1_q^\trans \bE_t^\trans \1_p)^2)\} \cdot \1_p \1_p^\trans \right\|_F^2 
    = \cO_P\left( \frac{1}{Tq} + \frac{1}{p^2q^2} \right) ,
\end{align*}
where the last equality used Lemma~\ref{lemma: rate_corollary_lemma2_LamCen} and the fact that all the terms in the second line are of order $\cO_P(1/(Tq))$ by the same argument below \eqref{eqn: sample_covariance_row}.
Thus, the term $\zeta_{n,p}$ in condition (C2) in Lemma~B.4~(i) of \cite{Barigozzietal2025_robust} is $1/\sqrt{Tq} + 1/(pq)$ here. Lastly, for their condition (C3), note that $\|\bM_p\|, \|\bM_q\| =1$ (see \eqref{eqn: M_norm}) and hence there exists some constant $c>0$ such that
\begin{align*}
    &\;\quad
    \left\| p \left( \Ddot{\bGamma}_r - \frac{1}{Tpq} \sum_{t=1}^T \bC_t\bC_t^\trans \right) \right\| 
    = \left\| \frac{1}{Tq} \sum_{t=1}^T \bM_p \E\left( \bE_t \bM_q \bE_t^\trans \right) \bM_p \right\| \\
    &\leq
    \frac{1}{Tq} \left\| \E\left( \sum_{t=1}^T \bE_t \bM_q \bE_t^\trans \right) \right\|
    \leq
    \frac{1}{Tq} \left\| \E\left( \sum_{t=1}^T \bE_t \bE_t^\trans \right) \right\| \leq c,
\end{align*}
where the second line used Lemma~\ref{lemma: inequality_sandwich}~(i) and the result on $\left\| \E\left( \sum_{t=1}^T \bE_t \bE_t^\trans \right) \right\|$ in Lemma~B.5~(ii) of \cite{ZhangCenCho2026}.

With all the above verified, we may apply Lemma~B.4 of \cite{Barigozzietal2025_robust}, and by the fact that $\wt\bQ$ and $\Ddot\bQ$ are respectively $\sqrt{p}$ times the first $k_r$ eigenvectors of $(Tpq)^{-1} \sum_{t=1}^T \wh\bL_t \wh\bL_t^\trans$ and $\Ddot{\bGamma}_r$, or equivalently, the first $k_r$ eigenvectors of $(Tq)^{-1} \sum_{t=1}^T \wh\bL_t \wh\bL_t^\trans$ and $p \Ddot{\bGamma}_r$, we conclude that there exists some orthogonal matrix $\bJ_r\in \R^{k_r\times k_r}$ such that
\[
\frac{1}{\sqrt{p}} \left\| \wt\bQ - \Ddot\bQ \bJ_r \right\|_F = \cO_P\Big( \frac{1}{\sqrt{Tq}} + \frac{1}{pq} \Big) ,
\]
as desired. This completes the proof of the lemma.
\end{proof}

\begin{lemma}\label{lemma: rate_ELEK_general}
Let Assumption~\ref{ass: noise} hold.
Let $\bLambda \in \R^{q\times q}$ be any matrix independent of $\{\bE_t\}$, and $\bK \in \R^{p\times k}$ be any deterministic matrix with fixed $k$ and $\max_{i\in[p],j\in[k]} |K_{i,j}| =\cO(1)$. Then we have
\[
\left\| \sum_{t=1}^T \bE_t \bLambda \bE_t^\trans \bK \right\|_F^2 = \cO_P\Big( T^2 p \cdot \{ \E[\tr(\bLambda)] \}^2 + T p^2 \cdot \E(\|\bLambda\|_F^2) \Big) .
\]
In particular, when $\bLambda = \bLambda_0 \bLambda_0^\trans$ for some $\bLambda_0$, we have
\[
\left\| \sum_{t=1}^T \bE_t \bLambda \bE_t^\trans \bK \right\|_F^2 = \cO_P\Big( (T^2 p + T p^2) \cdot \E(\|\bLambda_0\|_F^4) \Big) .
\]
\end{lemma}

\begin{proof}[Proof of Lemma~\ref{lemma: rate_ELEK_general}]

Note that by the condition of $\bK$, we have
\begin{equation}
\label{eqn: ELambdaEK_decomp}
\begin{split}
    \left\| \sum_{t=1}^T \bE_t \bLambda \bE_t^\trans \bK \right\|_F^2
    &=
    \sum_{i=1}^{p} \left\| \sum_{j=1}^{p} \bK_{j\cdot} \sum_{t=1}^T \bE_{t,i\cdot}^\trans \bLambda \bE_{t,j\cdot} \right\|^2 
    \lesssim
    \sum_{i=1}^{p} \sum_{j=1}^{p} \left( \sum_{t=1}^T \bE_{t,i\cdot}^\trans \bLambda \bE_{t,j\cdot} \right)^2 .
\end{split}
\end{equation}


Consider $\bE_t$. From Assumption~\ref{ass: noise}, it holds that
\begin{align*}
    \bE_t = \sum_{j\geq 0} a_{e,j} \bA_{e,r} \bX_{e,t-j} \bA_{e,c}^\trans + \sum_{j\geq 0} a_{\epsilon,j} \left( \bSigma_{e} \circ \bX_{\epsilon,t-j} \right) .
\end{align*}
Then for any $i,j\in[p]$, we have
\begin{align*}
    \E\left( \sum_{t=1}^T \bE_{t,i\cdot}^\trans \bLambda \bE_{t,j\cdot} \right)
    &=
    \E\Bigg\{ \sum_{t=1}^T \Big( \sum_{w\geq 0} a_{e,w} \bA_{e,r,i\cdot}^\trans \bX_{e,t-w} \bA_{e,c}^\trans + \sum_{h\geq 0} a_{\epsilon,h} \bSigma_{e,i\cdot}^\trans \circ \bX_{\epsilon,t-h,i\cdot}^\trans \Big) \bLambda \\
    &\;\quad
    \cdot \Big( \sum_{w\geq 0} a_{e,w} \bA_{e,c} \bX_{e,t-w}^\trans \bA_{e,r,j\cdot} + \sum_{h\geq 0} a_{\epsilon,h} \bSigma_{e,j\cdot} \circ \bX_{\epsilon,t-h,j\cdot} \Big) \Bigg\} \\
    &=
    \E\Bigg( \sum_{t=1}^T \sum_{w\geq 0} \sum_{m\geq 0} a_{e,w} a_{e,m} \bA_{e,r,i\cdot}^\trans \bX_{e,t-w} \bA_{e,c}^\trans \bLambda \bA_{e,c} \bX_{e,t-m}^\trans \bA_{e,r,j\cdot} \Bigg) \\
    &\quad
    + \E\Bigg( \sum_{t=1}^T \sum_{h\geq 0} \sum_{\ell\geq 0} a_{\epsilon,h} a_{\epsilon,\ell} (\bSigma_{e,i\cdot}^\trans \circ \bX_{\epsilon,t-h,i\cdot}^\trans) \bLambda (\bSigma_{e,j\cdot} \circ \bX_{\epsilon,t-\ell,j\cdot}) \Bigg) \\
    &\lesssim
    \E[\tr(\bA_{e,c}^\trans \bLambda \bA_{e,c})] \cdot T \cdot \bA_{e,r,i\cdot}^\trans \bA_{e,r,j\cdot}
    + \E[\tr(\bLambda)] \cdot T \cdot \mathbbm{I}_{\{i=j\}} ,
\end{align*}
so that together with the sparsity of $\bA_{e,r}$ and $\bA_{e,c}$ from \Cref{ass: noise}, it holds that
\begin{equation}
\label{eqn: ELambdaEK_mean}
\sum_{i=1}^{p} \sum_{j=1}^{p} \left\{ \E\left( \sum_{t=1}^T \bE_{t,i\cdot}^\trans \bLambda \bE_{t,j\cdot} \right) \right\}^2 \lesssim T^2 \cdot \{ \E[\tr(\bLambda)] \}^2 \cdot p .
\end{equation}
Next, for any $i,j\in[p]$, consider
\begin{align*}
    &\;\quad
    \var\left( \sum_{t=1}^T \bE_{t,i\cdot}^\trans \bLambda \bE_{t,j\cdot} \right)
    = \var\left( \sum_{t=1}^T \sum_{\ell=1}^{q} \sum_{h=1}^{q} E_{t,i\ell} \Lambda_{\ell h} E_{t,j h} \right) \\
    &=
    \var\Bigg\{ \sum_{t=1}^T \sum_{\ell=1}^{q} \sum_{h=1}^{q} \Lambda_{\ell h} \Big( \sum_{w\geq 0} a_{e,w} \bA_{e,r,i\cdot}^\trans \bX_{e,t-w} \bA_{e,c,\ell\cdot} + \sum_{m\geq 0} a_{\epsilon,m} \Sigma_{e,i\ell} X_{\epsilon,t-m,i\ell} \Big) \\
    &\;\quad
    \cdot \Big( \sum_{c\geq 0} a_{e,c} \bA_{e,r,j\cdot}^\trans \bX_{e,t-c} \bA_{e,c,h\cdot} + \sum_{z\geq 0} a_{\epsilon,z} \Sigma_{e,jh} X_{\epsilon,t-z,jh} \Big) \Bigg\} \\
    &\lesssim
    T \cdot \sum_{\ell=1}^{q} \sum_{h=1}^{q} \E(\Lambda_{\ell h}^2) \left( \sum_{w\geq 0} a_{e,w}^2 \left\Vert \bA_{e,r,i\cdot} \right\Vert^{2} \left\Vert \bA_{e,c,\ell\cdot} \right\Vert^{2} + \sum_{m\geq 0} a_{\epsilon,m}^2 \Sigma_{e,i\ell}^{2} \right) \\
    &\;\quad
    \cdot \left( \sum_{c\geq 0} a_{e,c}^2 \left\Vert \bA_{e,r,j\cdot} \right\Vert^{2} \left\Vert \bA_{e,c,h\cdot} \right\Vert^{2} + \sum_{z\geq 0} a_{\epsilon,z}^2 \Sigma_{e,jh}^{2} \right)
    \lesssim T \cdot \E(\|\bLambda\|_F^2) .
\end{align*}
Hence, we have
\begin{equation}
\label{eqn: ELambdaEK_var}
\begin{split}
    &\quad
    \sum_{i=1}^{p} \sum_{j=1}^{p} \var\left( \sum_{t=1}^T \bE_{t,i\cdot}^\trans \bLambda \bE_{t,j\cdot} \right)
    \lesssim T p^2 \cdot \E(\|\bLambda\|_F^2) .
\end{split}
\end{equation}
Combining \eqref{eqn: ELambdaEK_decomp}, \eqref{eqn: ELambdaEK_mean}, and \eqref{eqn: ELambdaEK_var}, the proof of this lemma is complete.
\end{proof}

\begin{lemma}\label{lemma: consistency_proj_decomp_rate}
Let all the assumptions in Lemma~\ref{lemma: consistency_PCA} hold.
Then we have
\begin{align}
     &\left\| \sum_{t=1}^T \bM_p \bE_t \bM_q \wt\bU \wt\bU^\trans \bC_t^\trans \bQ \right\|_F^2 = \cO_P\left( Tp^3q^3 + p^2 q^4  \right) ,
     \label{eqn: rate_cII_Q} \\
     &\left\| \sum_{t=1}^T \bC_t \wt\bU \wt\bU^\trans \bM_q \bE_t^\trans \bM_p \bQ \right\|_F^2 = \cO_P\left( Tp^3q^3 + p^2 q^4  \right) ,
     \label{eqn: rate_cIII_Q} \\
     &\left\| \sum_{t=1}^T \bM_p \bE_t \bM_q \wt\bU \wt\bU^\trans \bM_q \bE_t^\trans \bM_p \bQ \right\|_F^2 = \cO_P\left(
     p^2 q^3 + T p q^4
     + T^2 p q^2 + T p^2 q^2 \right) .
     \label{eqn: rate_cIV_Q} 
\end{align}
\end{lemma}

\begin{proof}[Proof of Lemma~\ref{lemma: consistency_proj_decomp_rate}]

First, by Lemma~\ref{lemma: consistency_PCA} and Assumption~\ref{ass: loadings},
\begin{align}
    &\quad
    \frac{1}{q} \wt\bU \wt\bU^\trans = \frac{1}{q} (\wt\bU - \bU \wt\bH_c) \wt\bU^\trans + \frac{1}{q} \bU \wt\bH_c \wt\bU^\trans \notag \\
    &=
    \frac{1}{q} (\wt\bU - \bU \wt\bH_c) \wt\bU^\trans
    + \frac{1}{q} \bU \wt\bH_c (\wt\bU - \bU \wt\bH_c)^\trans
    + \frac{1}{q} \bU \wt\bH_c \wt\bH_c^\trans \bU^\trans
    = \frac{1}{q} \bU \bU^\trans + o_P(1) .
    \label{eqn: wt_UU}
\end{align}
Similarly, using Lemma~\ref{lemma: population_cov_PCA_consistency}, we also have
\begin{align}
    &\quad
    \frac{1}{q} \wt\bU \wt\bU^\trans =
    \frac{1}{q} (\wt\bU - \Ddot\bU \bJ_c) \wt\bU^\trans
    + \frac{1}{q} \Ddot\bU \bJ_c (\wt\bU - \Ddot\bU \bJ_c)^\trans
    + \frac{1}{q} \Ddot\bU \Ddot\bU^\trans 
    =
    \frac{1}{q} \Ddot\bU \Ddot\bU^\trans + \cO_P\Big( \frac{1}{\sqrt{Tp}} + \frac{1}{pq} \Big) .
    \label{eqn: wt_UU_Ddot}
\end{align}

Consider \eqref{eqn: rate_cII_Q}. By the definition of $\bM_p$ we have
\begin{align}
    &\;\quad
    \left\| \sum_{t=1}^T \bM_p \bE_t \bM_q \wt\bU \wt\bU^\trans \bC_t^\trans \bQ \right\|_F^2 
    = \left\| \sum_{t=1}^T (\bI_p - p^{-1} \1_p\1_p^\trans) \bE_t (\bI_q - q^{-1} \1_q\1_q^\trans) \wt\bU \wt\bU^\trans \bC_t^\trans \bQ \right\|_F^2 \notag \\
    &\leq
    \left\| \sum_{t=1}^T \bE_t \wt\bU \wt\bU^\trans \bC_t^\trans \bQ \right\|_F^2
    + \frac{1}{p^2 q^2} \left\| \sum_{t=1}^T \1_p\1_p^\trans \bE_t \1_q\1_q^\trans \wt\bU \wt\bU^\trans \bC_t^\trans \bQ \right\|_F^2 \notag \\
    &\;\quad
    + \frac{1}{p^2} \left\| \sum_{t=1}^T \1_p\1_p^\trans \bE_t \wt\bU \wt\bU^\trans \bC_t^\trans \bQ \right\|_F^2
    + \frac{1}{q^2} \left\| \sum_{t=1}^T \bE_t \1_q\1_q^\trans \wt\bU \wt\bU^\trans \bC_t^\trans \bQ \right\|_F^2 \notag \\
    &=: \cI_{EC} + \cII_{EC} + \cIII_{EC} + \cIV_{EC} .
    \label{eqn: decomp_EUUC}
\end{align}
For $\cI_{EC}$, we may apply Lemma~\ref{lemma: inequality_sandwich}~(ii) and obtain
\begin{align*}
    &\;\quad
    \cI_{EC} = q^2 \left\| \sum_{t=1}^T \bE_t \frac{\wt\bU \wt\bU^\trans}{q} \bC_t^\trans \bQ \right\|_F^2 \\
    &\leq
    q^2 \left\| \sum_{t=1}^T \bE_t \frac{\Ddot\bU \Ddot\bU^\trans}{q} \bC_t^\trans \bQ \right\|_F^2
    + q^2\cdot \cO_P\Big( \frac{1}{Tp} + \frac{1}{p^2q^2} \Big) \cdot \sum_{i=1}^p \sum_{j=1}^q \left\| \sum_{t=1}^T E_{t,ij} \bC_t \right\|_F^2 \cdot \|\bQ \|_F^2 \\
    &\lesssim_P
    q^2 \cdot p \cdot T^2p^2q^2 \cdot \frac{1}{Tq} + \Big( \frac{q^2}{Tp} + \frac{1}{p^2}\Big) \cdot p \cdot Tp^2q^2
    = Tp^3q^3 + p^2 q^4 + Tpq^2 
    \lesssim Tp^3q^3 + p^2 q^4 ,
\end{align*}
where the last line used Lemma~\ref{ass: loadings}, and similar arguments as in \eqref{eqn: AFBE_segment} for each segment based on Lemma~B.5~(v) in \cite{Zhangetal2025} and the proof of Proposition~1.2 in the supplement of \cite{CenLam2025}.
Consider $\cII_{EC}$. Note that $\|\1_p\1_p^\trans\| = p$, thus
\begin{align*}
    &\;\quad
    \cII_{EC} \leq
    \frac{\|\1_p\1_p^\trans\|^2}{p^2} \cdot \left\| \sum_{t=1}^T \bE_t \left( \frac{\1_q\1_q^\trans}{q} \right) \wt\bU \wt\bU^\trans \bC_t^\trans \bQ \right\|_F^2
    \lesssim_P
    Tp^3q^3 + p^2 q^4  ,
\end{align*}
where the analogous arguments for $\cI_{EC}$ hold. Similar to the above, it is also direct that $\cIII_{EC}$ and $\cIV_{EC}$ have the same rate. Hence with the decomposition in \eqref{eqn: decomp_EUUC}, we conclude \eqref{eqn: rate_cII_Q}. Note that in showing \eqref{eqn: rate_cII_Q}, the rate on $\bQ$ is all controlled by sub-multiplicativity of Frobenius norm, and hence \eqref{eqn: rate_cIII_Q} can be shown of the same rate directly.

It remains to consider \eqref{eqn: rate_cIV_Q}, for which we can write
\begin{align}
    &\;\quad
    \left\| \sum_{t=1}^T \bM_p \bE_t \bM_q \wt\bU \wt\bU^\trans \bM_q \bE_t^\trans \bM_p \bQ \right\|_F^2 \notag \\
    &\leq
    \left\| \sum_{t=1}^T \bE_t \bM_q \wt\bU \wt\bU^\trans \bM_q \bE_t^\trans \bQ \right\|_F^2
    + \left\| \sum_{t=1}^T \bE_t \bM_q \wt\bU \wt\bU^\trans \bM_q \bE_t^\trans \frac{\1_p\1_p^\trans}{p} \bQ \right\|_F^2 \notag \\
    &\leq
    \left\| \sum_{t=1}^T \bE_t \bM_q (\wt\bU -\Ddot\bU \bJ_c) \wt\bU^\trans \bM_q \bE_t^\trans \bQ \right\|_F^2
    + \left\| \sum_{t=1}^T \bE_t \bM_q \Ddot\bU \bJ_c (\wt\bU - \Ddot\bU \bJ_c)^\trans \bM_q \bE_t^\trans \bQ \right\|_F^2 \notag \\
    &\;\quad
    + \left\| \sum_{t=1}^T \bE_t \bM_q \Ddot\bU \Ddot\bU^\trans \bM_q \bE_t^\trans \bQ \right\|_F^2
    + \left\| \sum_{t=1}^T \bE_t \bM_q (\wt\bU -\Ddot\bU \bJ_c) \wt\bU^\trans \bM_q \bE_t^\trans \frac{\1_p\1_p^\trans}{p} \bQ \right\|_F^2 \notag \\
    &\;\quad
    + \left\| \sum_{t=1}^T \bE_t \bM_q \Ddot\bU \bJ_c (\wt\bU - \Ddot\bU \bJ_c)^\trans \bM_q \bE_t^\trans \frac{\1_p\1_p^\trans}{p} \bQ \right\|_F^2
    + \left\| \sum_{t=1}^T \bE_t \bM_q \Ddot\bU \Ddot\bU^\trans \bM_q \bE_t^\trans \frac{\1_p\1_p^\trans}{p} \bQ \right\|_F^2 \notag \\
    &
    =: \cI_{EE,1} + \cII_{EE,1} + \cIII_{EE,1}
    + \cI_{EE,2} + \cII_{EE,2} + \cIII_{EE,2} .
    \label{eqn: decomp_EUUE}
\end{align}

For $\cI_{EE,1}$, by Lemma~\ref{lemma: inequality_sandwich}~(i) and $\|\bM_q\| =1$ (see \eqref{eqn: M_norm}), we have
\begin{align}
    \cI_{EE,1} &\leq
    \left\| \bM_q (\wt\bU -\Ddot\bU \bJ_c) \wt\bU^\trans \bM_q \right\|^2 \cdot \left\| \sum_{t=1}^T \bE_t \bE_t^\trans  \right\|_F^2 \cdot \|\bQ\|^2 \notag \\
    &\lesssim_P
    q \cdot \left( \frac{1}{Tp} + \frac{1}{p^2q^2} \right) \cdot q \cdot (T p^2 q + T^2 p q^2) \cdot p
    =
    p^2 q^3 
    + T p q^4 + T^2 q^2 ,
    \label{eqn: rate_cI_EE1}
\end{align}
where the second line used Assumption~\ref{ass: loadings}~(ii), Lemma~\ref{lemma: rate_corollary_lemma2_LamCen} and Lemma~\ref{lemma: population_cov_PCA_consistency}. From \eqref{eqn: rate_cI_EE1}, it follows the similar arguments that
\begin{equation}
\label{eqn: rate_cI_EE_other}
\cII_{EE,1} + \cI_{EE,2} + \cII_{EE,2} =\cO_P( p^2 q^3 + T p q^4 + T^2 q^2 ) .
\end{equation}
It remains to consider $\cIII_{EE,1}$ and $\cIII_{EE,2}$. For both of them, we can apply Lemma~\ref{lemma: rate_ELEK_general} by Assumption~\ref{ass: loadings}~(i) and with $\bLambda_0$ set as $\bM_q \Ddot\bU$, thus
\begin{equation}
\label{eqn: rate_cIII_EE}
\cIII_{EE,1} + \cIII_{EE,2} =\cO_P\left( (T^2 p + T p^2) \cdot \E(\|\bM_q \Ddot\bU\|_F^4) \right)
= \cO_P( T^2 p q^2 + T p^2 q^2 ) .
\end{equation}
Then \eqref{eqn: rate_cIV_Q} follows by combining \eqref{eqn: decomp_EUUE}, \eqref{eqn: rate_cI_EE1}, \eqref{eqn: rate_cI_EE_other}, and \eqref{eqn: rate_cIII_EE}. This hence completes the proof of this lemma.
\end{proof}





\subsection{Auxiliary results for the proof of Theorem~\ref{thm: consistency_change_ME}}

\begin{lemma}\label{lemma: kernel_moment_approx}
Let Assumption~\ref{ass: kernel} hold, and recall the notations $\kappa_i$, $i=0,1,2$, therein.
Then for $S_i$, $i=0,1,2$, in \eqref{eqn: def_hat_alpha_pm}, we have the following.
\begin{align}
    S_0 &= Tb \cdot \kappa_0 + \cO(1) , \label{eqn: S_0_kappa} \\
    S_1 &= - Tb^2 \kappa_1 + \cO(b) , \label{eqn: S_1_kappa}\\
    S_2 &= Tb^3 \kappa_2 + \cO(b^2) \label{eqn: S_2_kappa} .
\end{align}
where the hidden constants in all $\cO$-terms only depend on the kernel $K$.
\end{lemma}

\begin{proof}[Proof of \Cref{lemma: kernel_moment_approx}]
Define first
$L_{\kappa_0}:= \vert K'\vert_{\infty}$, 
$L_{\kappa_1}:= \vert K\vert_{\infty} + \vert K'\vert_{\infty}$, 
and $L_{\kappa_1}:= 2 \vert K\vert_{\infty} + \vert K'\vert_{\infty}$.
Consider $S_0$. From Assumption~\ref{ass: kernel}, $\vert K'\vert_{\infty} <\infty$ and hence $K$ is Lipschitz with constant $L_{\kappa_0}$.
Define $x_i:= i/(Tb)$ for $i\in[Tb]$. 
Then for each interval $[x_{i-1}, x_i]$, the discretization error of the Riemann sum is
\begin{align*}
    \left| \int_{x_{i-1}}^{x_i} K(x) \, \mathrm{d} x - K(x_i) \frac{1}{Tb} \right|
    \leq
    \int_{x_{i-1}}^{x_i} | K(x) - K(x_i)| \, \mathrm{d} x
    \leq
    L_{\kappa_0} \int_{x_{i-1}}^{x_i} (x_i - x) \, \mathrm{d} x
    = 
    \frac{L_{\kappa_0}}{2 (Tb)^2} .
\end{align*}
By the triangle inequality and aggregating the above over $[0,1]$, we have
\begin{align*}
    \left| \int_{0}^{1} K(x) \, \mathrm{d} x - \sum_{i=1}^{Tb} K(x_i) \frac{1}{Tb} \right| 
    \leq
    \sum_{i=1}^{Tb} \left| \int_{x_{i-1}}^{x_i} K(x) \, \mathrm{d} x - K(x_i) \frac{1}{Tb} \right|
    \leq
    \sum_{i=1}^{Tb} \frac{L_{\kappa_0}}{2 (Tb)^2} 
    = 
    \frac{L_{\kappa_0}}{2 T b} .
\end{align*}
Recall $b\to 0$ according to Assumption~\ref{ass: kernel} and hence $\sum_{i=1}^{T} K(x_i)= \sum_{i=1}^{Tb} K(x_i)$ for small enough~$b$ (e.g, $b\leq 1$).
Hence by definition, we have for small enough $b$ that
\begin{align*}
    S_0 &= \sum_{i=1}^T K\Big(\frac{i}{Tb}\Big)
    = \sum_{i=1}^T K(x_i)
    = \sum_{i=1}^{Tb} K(x_i)
    = Tb \left\{ \sum_{i=1}^{Tb} K(x_i) \frac{1}{Tb} \right\} \\
    &=
    Tb \left\{ \int_0^{1} K(x) \, \mathrm{d} x + \left( \sum_{i=1}^{Tb} K(x_i) \frac{1}{Tb} - \int_{0}^{1} K(x) \, \mathrm{d} x \right) \right\} 
    =
    Tb \cdot \kappa_0 + O(1) ,
\end{align*}
thus showing \eqref{eqn: S_0_kappa} (noting that the $\cO$-term depends on the kernel $K$). 

Similarly, note that the function $f_{\kappa_1}(x) := xK(x)$ is Lipschitz on $[0,1]$ due to $f_{\kappa_1}'(x) = K(x) + xK'(x)$ which are uniformly bounded by $\vert K\vert_{\infty}, \vert K'\vert_{\infty} <\infty$ from Assumption~\ref{ass: kernel}, and the fact that $x$ is in a bounded interval. The Lipschitz constant is at most $L_{\kappa_1}$. Thus,
\begin{align*}
    &\;\quad
    \left| \int_{0}^{1} x K(x) \, \mathrm{d} x - \sum_{i=1}^{Tb} x_i K(x_i) \frac{1}{Tb} \right| 
    \leq
    \sum_{i=1}^{Tb} \left| \int_{x_{i-1}}^{x_i} x K(x) \, \mathrm{d} x - x_i K(x_i) \frac{1}{Tb} \right| \\
    &\leq
    \sum_{i=1}^{Tb} \int_{x_{i-1}}^{x_i} \big| x K(x) - x_i K(x_i) \big| \, \mathrm{d} x
    \leq
    \sum_{i=1}^{Tb} L_{\kappa_1} \int_{x_{i-1}}^{x_i} (x_i - x) \, \mathrm{d} x
    = 
    \frac{L_{\kappa_1}}{2 Tb} .
\end{align*}
Then for small enough $b$ we can also show \eqref{eqn: S_1_kappa} by
\begin{align*}
    S_1 &= - \sum_{i=1}^{Tb} \frac{i}{T} \cdot K\Big(\frac{i}{Tb}\Big)
    = - b \sum_{i=1}^{Tb} x_i  K(x_i)
    = - Tb^2 \sum_{i=1}^{Tb} x_i  K(x_i) \frac{1}{Tb} \\
    &=
    - Tb^2 \left( \int_{0}^{1} x K(x) \, \mathrm{d} x + \sum_{i=1}^{Tb} x_i  K(x_i) \frac{1}{Tb} - \int_{0}^{1} x K(x) \, \mathrm{d} x \right)
    =
    - Tb^2 \kappa_1 + \cO(b) .
\end{align*}
Lastly, \eqref{eqn: S_2_kappa} can be shown similarly, with $f_{\kappa_2}(x) := x^2 K(x)$ being Lipschitz on $[0,1]$ from $f_{\kappa_2}'(x) = 2x K(x) + x^2 K'(x)$, thus with constant $L_{\kappa_2}$.
This concludes the proof of this lemma.
\end{proof}

\begin{lemma}\label{lemma: weight_prop}
Under Assumption~\ref{ass: kernel}, there exists some constant 
$C_w \geq 1$
only depending on the kernel $K$ such that:
\begin{itemize}
    \item [(i)] $\sum_{s=1}^{bT} w_s = 1$, and
    $\sum_{s=0}^{bT} s \cdot w_s = \sum_{s=1}^{bT} s \cdot w_s = 0$.
    \item [(ii)] For any $s\in \{0,1,\dots, bT\}$, we have $|w_{s}| \leq C_w /Tb$.
    \item [(iii)] $\sum_{s= 1}^{bT} |w_{s}| s^2/T^2 \leq C_w b^2$.
\end{itemize}
\end{lemma}

\begin{proof}[Proof of \Cref{lemma: weight_prop}]
Consider part~(i) first. By definitions, we have
\begin{align*}
    \sum_{s=1}^{bT} w_s &= 
    \sum_{s=1}^{bT} \frac{K\{s/(Tb)\} (S_2 +S_1 s/T)}{S_2 S_0 - S_1^2} 
    =
    \frac{S_2}{S_2 S_0 - S_1^2} \sum_{s=1}^{bT} K\left(\frac{s}{Tb}\right) 
    + \frac{S_1}{S_2 S_0 - S_1^2} \sum_{s=1}^{bT} K\left(\frac{s}{Tb}\right) \frac{s}{T} \\
    &=
    \frac{S_2}{S_2 S_0 - S_1^2} \sum_{s=1}^{T} K\left(\frac{s}{Tb}\right) 
    + \frac{S_1}{S_2 S_0 - S_1^2} \sum_{s=1}^{T} K\left(\frac{s}{Tb}\right) \frac{s}{T} \\
    &=
    \frac{S_2}{S_2 S_0 - S_1^2} \cdot S_0
    + \frac{S_1}{S_2 S_0 - S_1^2} \cdot (-S_1)
    = 1 ,
\end{align*}
where the third equality used the fact that $b\to 0$ from Assumption~\ref{ass: kernel} (in fact, we let $b\in(0,1)$ directly) and the kernel $K$ has support $[-1,1]$.
Similarly,
\begin{align*}
    \sum_{s=1}^{bT} w_s \cdot \frac{s}{T} &= 
    \frac{S_2}{S_2 S_0 - S_1^2} \sum_{s=1}^{T} K\left(\frac{s}{Tb}\right) \frac{s}{T} 
    + \frac{S_1}{S_2 S_0 - S_1^2} \sum_{s=1}^{T} K\left(\frac{s}{Tb}\right) \frac{s^2}{T^2} \\
    &=
    \frac{S_2}{S_2 S_0 - S_1^2} \cdot (-S_1)
    + \frac{S_1}{S_2 S_0 - S_1^2} \cdot S_2
    = 0 .
\end{align*}
This completes the proof of part~(i).

For part~(ii), note that $Tb\to \infty$ by Assumption~\ref{ass: kernel}, and hence Lemma~\ref{lemma: kernel_moment_approx} gives
\[
S_0 \asymp Tb \cdot \kappa_0 , \quad
S_1 \asymp - Tb^2 \cdot \kappa_1 , \quad
S_2 \asymp Tb^3 \cdot \kappa_2 .
\]
From Assumption~\ref{ass: kernel}, there exists some constant $C >0$ that only depends on the kernel $K$, such that $\vert K\vert_{\infty} \leq C$. 
Therefore, by definition we have
\begin{align*}
    |w_{0}| &= \left| \frac{S_2}{S_2 S_0 - S_1^2} \cdot K(0) \right|
    \leq
    \vert K\vert_{\infty} \left| \frac{S_2}{S_2 S_0 - S_1^2} \right|
    \leq
    C \left| \frac{Tb^3 \kappa_2}{T^2 b^4 \kappa_2 \kappa_0 - T^2b^4 \kappa_1^2} \right|
    = \frac{1}{Tb} \frac{C\kappa_2}{\kappa_2 \kappa_0 - \kappa_1^2} .
\end{align*}
Noting that $|S_1 s/T| \leq |S_1 b|$ for any $s\leq Tb$, by analogous arguments, the above result also holds for $|w_{s}|$ for any $s\in[Tb]$.
Set $C_w = C\kappa_2/(\kappa_2 \kappa_0 - \kappa_1^2)$ finishes this proof of part~(ii).

Lastly, by $s/T \leq b$ for any $s\leq Tb$ again, we have
\begin{align*}
    \sum_{s= 1}^{bT} |w_{s}| \cdot \frac{s^2}{T^2}
    &\leq
    \left( \max_{s\in[bT]} |w_{s}| \right)  \sum_{s= 1}^{bT} b^2 
    \leq
    \frac{C_w}{Tb} \cdot bT \cdot b^2 = C_w b^2 ,
\end{align*}
thereby completing the proof of this lemma, where note that $C_w \geq 1$ by (i) and (ii).
\end{proof}

\begin{lemma}\label{lemma: smooth_trend}
Let Assumptions~\ref{ass: trend_fun} and \ref{ass: kernel} hold.
Then we have
\begin{align*}
    & \max_{t\in[bT +1, T-bT]}
    \left\| \sum_{u= t+1}^{t+bT} w_{u-t} \bmu_{\alpha,u} - \sum_{u= t-bT}^{t-1} w_{t-u} \bmu_{\alpha,u} \right\|_{\infty} = 
    \cO\left( b^2 \right) , 
\end{align*}
where the $\cO$-term has hidden constants only depending on $\bmu_\alpha$ and $K$.
Note that if the trend $\{\bmu_{\alpha,t}\}$ is at a constant level (so that the row main effects are piecewise-constant), then the right-hand side above is exactly 0.
\end{lemma}

\begin{proof}[Proof of \Cref{lemma: smooth_trend}]
Fix any $t\in[bT +1, T-bT]$, denote by $x_t:=t/T$ and we may rewrite
\begin{align}
    &\;\quad
    \sum_{u= t+1}^{t+bT} w_{u-t} \bmu_{\alpha,u} - \sum_{u= t-bT}^{t-1} w_{t-u} \bmu_{\alpha,u} =
    \sum_{u= t+1}^{t+bT} w_{u-t} \bmu_{\alpha}\left( \frac{u}{T} \right) - \sum_{u= t-bT}^{t-1} w_{t-u} \bmu_{\alpha}\left( \frac{u}{T} \right) \notag \\
    &=
    \sum_{s= 1}^{bT} w_{s} \bmu_{\alpha}\left( \frac{t+s}{T} \right) - \sum_{s= 1}^{bT} w_{s} \bmu_{\alpha}\left( \frac{t-s}{T} \right) 
    =
    \sum_{s= 1}^{bT} w_{s} \bmu_{\alpha}\left( x_t + \frac{s}{T} \right) - \sum_{s= 1}^{bT} w_{s} \bmu_{\alpha}\left( x_t - \frac{s}{T} \right) .
    \label{eqn: smooth_trend_decomp}
\end{align}

Consider $\bmu_{\alpha}(x_t + s/T)$. From Assumption~\ref{ass: trend_fun}, we may apply Taylor's theorem and expand each component to the second order:
\begin{align*}
    \bmu_{\alpha}\left( x_t + \frac{s}{T} \right) &=
    \bmu_{\alpha}(x_t) + \bmu_{\alpha}'(x_t) \frac{s}{T} + \frac{1}{2} \bmu_{\alpha}''(\xi_s^+) \frac{s^2}{T^2} \\
    &=
    \bmu_{\alpha}(x_t) + \bmu_{\alpha}'(x_t) \frac{s}{T} + \frac{1}{2} \bmu_{\alpha}''(x_t) \frac{s^2}{T^2} + \frac{s^2}{2 T^2} \left\{ \bmu_{\alpha}''(\xi_s^+) - \bmu_{\alpha}''(x_t) \right\} ,
\end{align*}
where $\xi_s^+ \in [x_t, x_t+ s/T]$. Similarly, for $\bmu_{\alpha}(x_t - s/T)$ we have
\begin{align*}
    \bmu_{\alpha}\left(x_t - \frac{s}{T} \right) &=
    \bmu_{\alpha}(x_t) - \bmu_{\alpha}'(x_t) \frac{s}{T} + \frac{1}{2} \bmu_{\alpha}''(\xi_s^-) \frac{s^2}{T^2} \\
    &=
    \bmu_{\alpha}(x_t) - \bmu_{\alpha}'(x_t) \frac{s}{T} + \frac{1}{2} \bmu_{\alpha}''(x_t) \frac{s^2}{T^2} + \frac{s^2}{2 T^2} \left\{ \bmu_{\alpha}''(\xi_s^-) - \bmu_{\alpha}''(x_t) \right\} ,
\end{align*}
where $\xi_s^- \in [x_t- s/T, x_t]$.
Altogether, plugging the above into \eqref{eqn: smooth_trend_decomp}, we have
\begin{align}
    &\;\quad
    \sum_{u= t +1}^{t+bT} w_{u-t} \bmu_{\alpha,u} - \sum_{u= t-bT}^{t-1} w_{t-u} \bmu_{\alpha,u} =
    \sum_{s= 1}^{bT} w_{s} \bmu_{\alpha}\left( x_t + \frac{s}{T} \right) - \sum_{s= 1}^{bT} w_{s} \bmu_{\alpha}\left( x_t - \frac{s}{T} \right) \notag \\
    &=
    \bmu_{\alpha}'(x_t) \sum_{s=1}^{bT} w_s \frac{s}{T}
    + \bmu_{\alpha}'(x_t) \sum_{s=1}^{bT} w_s \frac{s}{T}
    + \frac{1}{2} \bmu_{\alpha}''(x_t) \sum_{s=1}^{bT} w_s \frac{s}{T} 
    - \frac{1}{2} \bmu_{\alpha}''(x_t) \sum_{s=1}^{bT} w_s \frac{s}{T} \notag \\
    &\;\quad
    + \sum_{s= 1}^{bT} \frac{s^2 w_{s}}{2T^2}  \left\{ \bmu_{\alpha}''(\xi_s^+) - \bmu_{\alpha}''(x_t) \right\} 
    - \sum_{s= 1}^{bT} \frac{s^2 w_{s}}{2T^2} \left\{ \bmu_{\alpha}''(\xi_s^-) - \bmu_{\alpha}''(x_t) \right\} \notag \\
    &=
    \sum_{s= 1}^{bT} \frac{s^2 w_{s}}{2T^2} \left\{ \bmu_{\alpha}''(\xi_s^+) - \bmu_{\alpha}''(\xi_s^-) \right\} ,
    \label{eqn: smooth_trend_final_decomp}
\end{align}
where the last equality used Lemma~\ref{lemma: weight_prop}~(i) and some straightforward computation.

By $|(\bmu_{\alpha})_i^{''}|_\infty \leq c$, $i\in[p]$, from Assumption~\ref{ass: trend_fun}, we have
\begin{align}
    &\;\quad
    \left\| \sum_{s= 1}^{bT} \frac{s^2 w_{s}}{2T^2} \left\{ \bmu_{\alpha}''(\xi_s^+) - \bmu_{\alpha}''(\xi_s^-) \right\} \right\|_\infty \notag \\
    &\leq
    \sum_{s= 1}^{bT} \frac{s^2}{2T^2} \cdot |w_{s}| \cdot \left\|  \bmu_{\alpha}''(\xi_s^+) - \bmu_{\alpha}''(\xi_s^-) \right\|_\infty
    \leq
    \frac{1}{2} C_w b^2 \cdot 2c 
    = C_w c \cdot b^2 ,
    \label{eqn: smooth_trend_final_decomp_2}
\end{align}
where the second inequality used Lemma~\ref{lemma: weight_prop}~(iii).
We conclude the proof of this lemma by applying triangle inequality with \eqref{eqn: smooth_trend_final_decomp} and \eqref{eqn: smooth_trend_final_decomp_2}.
\end{proof}

\begin{lemma}\label{lemma: smooth_E_entry_bound}
Let Assumptions~\ref{ass: noise} and \ref{ass: kernel} hold, with $\nu=\nu_e\geq 4$ and $\eta =\eta_e > 1/2 - 1/\nu$.
For any $i\in[p]$, $t\in[bT+1, T-bT]$, 
when $a \geq a_0 \sqrt{q/(Tb)}$ for some constant $a_0>0$, we have
\begin{align*}
    \P\left(\left| \sum_{j=1}^q \left( \sum_{u= t+1}^{t+bT} w_{u-t} E_{u,ij}  - \sum_{u= t-bT}^{t-1} w_{t-u} E_{u,ij} \right) \right| > a \right)
    &\leq
    \frac{c_{e,1} (Tb)^{1-\nu} q}{a^\nu} + 8 \exp\left( 1- \frac{c_{e,2} a^2 Tb}{q} \right) ,
\end{align*}
where $a_0$ satisfies that $c_{e,2} a^2 Tb /q \geq 1$ for all $a\geq a_0$, and $c_{e,1}$ and $c_{e,2}$ are some positive constants that only depending on the kernel and the distribution of $\{\bE_t\}_{t \in \mathbb{Z}}$.
\end{lemma}

\begin{proof}[Proof of \Cref{lemma: smooth_E_entry_bound}]
Fix any $i\in[p]$, $t\in[bT+1, T-bT]$, and $a>0$. Define $\wt{w}_{i} := w_i$ for $i> 0$, $\wt{w}_{i} := -w_{-i}$ for $i< 0$, and $\wt{w}_{0} := 0$.
Then Assumption~\ref{ass: noise} gives
\begin{align}
    &\;\quad
    \P\left(\left| \sum_{j=1}^q \left( \sum_{u= t+1}^{t+bT} w_{u-t} E_{u,ij}  - \sum_{u= t-bT}^{t-1} w_{t-u} E_{u,ij} \right) \right| > a \right) 
    =
    \P\left( \left| \sum_{j=1}^q \sum_{u= t-bT}^{t+bT} \wt{w}_{u-t} E_{u,ij} \right| > a \right) \notag \\
    &=
    \P\left(\left| \sum_{j=1}^q \sum_{u= t-bT}^{t+bT} \wt{w}_{u-t} \left( \bA_{e,r, i\cdot}^\trans \bF_{e,u} \bA_{e,c, j\cdot} + \Sigma_{e, ij} E_{e,u, ij} \right) \right| > a \right) \notag \\
    &\leq
    \P\left( \left| \sum_{j=1}^q \sum_{u= t-bT}^{t+bT} \wt{w}_{u-t} \bA_{e,r, i\cdot}^\trans \bF_{e,u} \bA_{e,c, j\cdot} \right| > a/2 \right)
    + \P\left( \left| \sum_{j=1}^q \sum_{u= t-bT}^{t+bT} \wt{w}_{u-t} \Sigma_{e, ij} E_{e,u, ij} \right| > a/2 \right) \notag \\
    &=:
    \cI_{E,it} + \cII_{E,it} .
    \label{eqn: smooth_E_entry_decomp}
\end{align}


Consider $\cII_{E,it}$. From Assumption~\ref{ass: noise}, $E_{e,u, ij} = \sum_{h\geq 0} a_{\epsilon,h} X_{\epsilon,u-h, ij}$ with $X_{\epsilon,t, ij}$ i.i.d.\ over $t,i,j$.
This allows us to apply a version of Theorem~2~(ii) in \cite{Liuetal2013} extended to non-stationarity as discussed in their Section~4. 
To this end, define the process $\{\wt{E}_z\}_{z\in[q(2bT+1)]}$ such that $\wt{E}_z := Tb \cdot \wt{w}_{u-t} \Sigma_{e, ij} \sum_{h\geq 0} a_{\epsilon,h} X_{\epsilon,u-h, ij}$,
with $z = (j-1)q + u - (t - bT) + 1$ for each multi-index $(j,u) \in [q] \times [t-bT, t+bT]$.
Then we have
\[
\sum_{z=1}^{q(2bT+1)} \wt{E}_z = Tb \left( \sum_{j=1}^q \sum_{u= t-bT}^{t+bT} \wt{w}_{u-t} \Sigma_{e, ij} E_{e,u, ij} \right) .
\]
First, with $\nu$ from Assumption~\ref{ass: noise}~(i), we verify that for any $z\in[q(2bT+1)]$ (let $(j,u)$ be the multi-index corresponding to this index $z$):
\begin{align*}
    \E(|\wt{E}_z|^\nu) &=
    (Tb)^\nu \cdot \E\left( \left| \wt{w}_{u-t} \Sigma_{e,ij} \sum_{h\geq 0} a_{\epsilon,h} X_{\epsilon,u-h,ij} \right|^\nu \right)
    \leq
    (Tb)^\nu \cdot \left( \frac{C_w}{Tb} \right)^\nu C_{\Sigma,e}^\nu c^\nu C_e < \infty ,
\end{align*}
where $C_w$ is from Lemma~\ref{lemma: weight_prop}~(ii), $C_e >0$ denotes the uniform upper bound of the $\nu$th moment of $X_{e,u, ij}$ due to Assumption~\ref{ass: noise}~(i), $C_{\Sigma,e}>0$ is the uniform upper bound of elements in $\bSigma_{e}$, and $c$ is some absolute constant due to Assumption~\ref{ass: noise}.
Similarly, for any $m\geq 0$, noting that $\wt{E}_{z_1}$ and $\wt{E}_{z_2}$ would be independent for any $|z_1-z_2|\geq 2bT+1$, we compute the uniform functional dependence measure of $\{\wt{E}_z\}_{z\in[q(2bT+1)]}$ as (using notations in \cite{Liuetal2013}, let $p$ be our $\nu$ in Assumption~\ref{ass: noise}):
\begin{align*}
    &\;\quad
    \theta_{m,\nu}^* = \mathbbm{1}\{m<2bT+1\} \cdot \max_{u\in[t-bT, t+bT], j\in[q]} \Big\| Tb \cdot \wt{w}_{u-t} \Sigma_{e, ij} a_{\epsilon,m} X_{\epsilon,u-m, ij} \Big\|_\nu \\
    &\leq
    \mathbbm{1}\{m<2bT+1\} \cdot Tb \cdot \frac{C_w}{Tb} \cdot C_{\Sigma,e} C_e^{1/\nu} |a_{\epsilon,m}| 
    =
    \mathbbm{1}\{m<2bT+1\} \cdot C_w C_{\Sigma,e} C_e^{1/\nu} |a_{\epsilon,m}| ,
\end{align*}
and thus the corresponding tail sum $\Theta_{m,\nu}^* \leq C_w C_{\Sigma,e} C_e^{1/\nu} \sum_{\ell=m}^\infty |a_{\epsilon,\ell}|$.
By Assumption~\ref{ass: noise}, we have $\Theta_{m,\nu}^* = \cO(m^{-\eta})$. Given that $\eta > 1/2 - 1/\nu$ from the condition of $\eta$ in the statement of this lemma, we may apply Theorem~2~(ii) of \cite{Liuetal2013} with $n=q(2bT+1)$ therein.
Namely, with $x=Tb(a/2)$ from \eqref{eqn: smooth_E_entry_decomp}, there exists some positive constants $C_1$ and $C_2$ such that under $x \geq \sqrt{q Tb}\, \Theta_{0,\nu}^* /C_2$ 
(which can be satisfied given $a \geq a_0 \sqrt{q/Tb}$ for some constant $a_0>0$), 
we have
\begin{align}
    &\;\quad
    \P\left( \left| \sum_{z=1}^{q(2bT+1)} \wt{E}_{z} \right| > x \right)
    \leq
    \frac{C_1 c^\nu qTb}{x^\nu} + 4 \exp\left( 1- \frac{C_2^2 x^2}{c^2 qTb} \right) \cdot \sum_{\ell=1}^\infty \exp\left( -\ell^{1-2/\nu} \right) \notag \\
    &\leq
    \frac{C_1 c^\nu qTb}{x^\nu} + 4 \exp\left( 1- \frac{C_2^2 x^2}{c^2 qTb} \right) \cdot
    \int_{0}^{\infty} \exp\left( -\ell^{1-2/\nu} \right) \mathrm{d}\ell \notag \\
    &=
    \frac{C_1 c^\nu qTb}{x^\nu} + 4 \exp\left( 1- \frac{C_2^2 x^2}{c^2 qTb} \right) \cdot
    \frac{\nu}{\nu-2} \Gamma\left( \frac{\nu}{\nu-2} \right)
    \leq
    \frac{C_1 c^\nu qTb}{x^\nu} + 8 \exp\left( 1- \frac{C_2^2 x^2}{c^2 qTb} \right) ,
    \label{eqn: sum_tilde_Ez_concen}
\end{align}
where $\Gamma(\cdot)$ is the Gamma function, and the last inequality used the fact that $\nu/(\nu-2) \in [1,2]$ for $\nu\geq 4$ and hence also $\Gamma(\nu/(\nu-2)) \leq 1$.
The term $\cI_{E,it}$ in \eqref{eqn: smooth_E_entry_decomp} can be bounded analogously.
This concludes the proof of this lemma.
\end{proof}

\begin{lemma}\label{lemma: smooth_E_vec1_bound}
Let all assumptions in Lemma~\ref{lemma: smooth_E_entry_bound} hold.
Then we have
\begin{align*}
    &\;\quad
    \max_{t\in[bT+1, T-bT]} \left\| \frac{1}{q} \left(\sum_{u= t +1}^{t+bT} w_{u-t} \bE_u  - \sum_{u= t-bT}^{t-1} w_{t-u} \bE_u \right) \1_q \right\|_\infty \\
    &=
    \cO_P\left( (pT)^{1/\nu} (Tb)^{(1-\nu)/\nu} q^{(1-\nu)/\nu} + \sqrt{\frac{\log(pT)}{qTb}} \right) .
\end{align*}
\end{lemma}

\begin{proof}[Proof of Lemma~\ref{lemma: smooth_E_vec1_bound}]
With Lemma~\ref{lemma: smooth_E_entry_bound} and the notations therein, by the union bound we have
\begin{align*}
    &\;\quad
    \P\left[ \max_{t\in[bT+1, T-bT]} \left\| \frac{1}{q} \left(\sum_{u= t +1}^{t+bT} w_{u-t} \bE_u  - \sum_{u= t-bT}^{t-1} w_{t-u} \bE_u \right) \1_q \right\|_\infty > z \right] \\
    &=
    \P\left[ \max_{i\in[p], t\in[bT+1, T-bT]}
    \left| \frac{1}{q} \sum_{j=1}^q \left( \sum_{u= t +1}^{t+bT} w_{u-t} E_{u,ij}  - \sum_{u= t-bT}^{t-1} w_{t-u} E_{u,ij} \right) \right| > z \right] \\
    &\leq
    p T \cdot  \P\left[
    \left| \sum_{j=1}^q \left( \sum_{u= bT+2}^{2bT+1} w_{u-bT-1} E_{u,1j}  - \sum_{u= 1}^{bT} w_{bT+1-u} E_{u,1j} \right) \right| > qz \right] \\
    &\leq
    p T \left[ \frac{c_{e,1} (Tb)^{1-\nu} q}{(qz)^\nu} + 8 \exp\left( 1- \frac{c_{e,2} (qz)^2 Tb}{q} \right) \right] \\
    &=
    c_{e,1} pT (Tb)^{1-\nu} q^{1-\nu} z^{-\nu}
    + 8 \exp\left( 1 + \log(pT) - c_{e,2} q Tb z^2 \right) ,
\end{align*}
for any $qz \geq a_0 \sqrt{q/Tb}$ (so as to apply Lemma~\ref{lemma: smooth_E_entry_bound}) such that $1 - c_{e,2} q Tb z^2 \leq 0$. By setting $\log(pT) - c_{e,2} q Tb z^2$ as 0, and the first term above as some constant, we conclude the desired $\cO_P$-term and hence completed the proof.
\end{proof}

\begin{lemma}\label{lemma: smooth_E_entry_bound_local}
Let Assumptions~\ref{ass: noise} and \ref{ass: kernel} hold, with $\nu=\nu_e\geq 4$ and $\eta =\eta_e > 1/2 - 1/\nu$.
For any $i\in[p]$, 
$\ell \in [Q_\alpha]$,
$t$ such that $\theta_{\alpha,\ell} - Tb \leq t \leq \theta_{\alpha,\ell} -1$,
when $a \geq a_0 \sqrt{q}/(Tb)$ for some constant $a_0>0$, we have
\begin{align*}
    \P\left(\left|
    \frac{1}{\sqrt{\theta_{\alpha,\ell} -t}}
    \sum_{j=1}^q \sum_{u=t-bT}^{\theta_{\alpha,\ell} -bT -1} w_{t-u} E_{u,ij}  \right| > a \right)
    &\leq
    \frac{c_{e,1} (Tb)^{-\nu} q (\theta_{\alpha,\ell} -t)^{1 - \nu/2} }{a^\nu} \\
    &\;\quad
    + 8 \exp\left( 1- \frac{c_{e,2} a^2 T^2b^2}{q} \right) ,
\end{align*}
where $a_0$ satisfies that $c_{e,2} a^2 T^2b^2 /q \geq 1$ for all $a\geq a_0$, and $c_{e,1}$ and $c_{e,2}$ are some positive constants that only depending on the kernel and the distribution of components in the noise. 
\end{lemma}

\begin{proof}[Proof of \Cref{lemma: smooth_E_entry_bound_local}]
The proof follows analogous arguments as in the proof of Lemma~\ref{lemma: smooth_E_entry_bound}. In particular, we use $\sqrt{\theta_{\alpha,\ell} -t} \wt{E}_z$ in place of $\wt{E}_z$ therein, and we set $n=q(\theta_{\alpha,\ell} -t)$ when applying Theorem~2~(ii) of \cite{Liuetal2013}.
Then using notations in \eqref{eqn: sum_tilde_Ez_concen}, it holds that
with $x= \sqrt{\theta_{\alpha,\ell} -t} \cdot Tb(a/2)$
there exists some positive constants $C_1$ and $C_2$ such that under $x \geq \sqrt{q(\theta_{\alpha,\ell} -t)}\, \Theta_{0,\nu}^* /C_2$ 
(which can be satisfied given $a \geq a_0 \sqrt{q}/(Tb)$ for some constant $a_0>0$), we have
\begin{align}
    &\;\quad
    \P\left( \left| \sum_{z=1}^{q(\theta_{\alpha,\ell} -t)} 
    \sqrt{\theta_{\alpha,\ell} -t}
    \wt{E}_{z} \right| > x \right)
    \leq
    \frac{C_1 c^\nu q (\theta_{\alpha,\ell} -t) }{x^\nu} + 8 \exp\left( 1- \frac{C_2^2 x^2}{c^2 q (\theta_{\alpha,\ell} -t) } \right) ,
    \label{eqn: sum_tilde_Ez_concen_local}
\end{align}
which gives the desired result after plugging in $x$. This concludes the proof of this lemma.
\end{proof}

\begin{lemma}\label{lemma: smooth_E_vec1_bound_local}
Let all assumptions in Lemma~\ref{lemma: smooth_E_entry_bound_local} hold.
Then we have
\begin{align*}
    &\;\quad
    \max_{\ell \in[Q_\alpha], t\in [\theta_{\alpha,\ell} -Tb, \theta_{\alpha,\ell} -1]} \left\| \frac{1}{q} \frac{1}{|\theta_{\alpha,\ell} -t|^{1/2}} \left( \sum_{u=t-bT}^{\theta_{\alpha,\ell} -bT -1} w_{t-u} \bE_u \right) \1_q \right\|_\infty  \\
    &=
    \cO_P\left( p^{1/\nu} (Tb)^{(1-\nu)/\nu} q^{(1-\nu)/\nu} 
    + \sqrt{\frac{\log(pTb)}{q T^2b^2}} \right) .
\end{align*}
\end{lemma}

\begin{proof}[Proof of Lemma~\ref{lemma: smooth_E_vec1_bound_local}]
With Lemma~\ref{lemma: smooth_E_entry_bound_local} and the notations therein, by the union bound and linearly spaced change points in $\Theta_\alpha$, we have
\begin{align*}
    &\;\quad
    \P\left[ \max_{\ell \in[Q_\alpha], t\in [\theta_{\alpha,\ell} -Tb, \theta_{\alpha,\ell} -1]} \left\| \frac{1}{q} \frac{1}{|\theta_{\alpha,\ell} -t|^{1/2}} \left( \sum_{u=t-bT}^{\theta_{\alpha,\ell} -bT -1} w_{t-u} \bE_u \right) \1_q \right\|_\infty  > z \right] \\
    &=
    \P\left[ \max_{i\in[p], \ell \in[Q_\alpha], t\in [\theta_{\alpha,\ell} -Tb, \theta_{\alpha,\ell} -1]} \left| \frac{1}{q} \frac{1}{|\theta_{\alpha,\ell} -t|^{1/2}} \sum_{j=1}^{q} \sum_{u=t-bT}^{\theta_{\alpha,\ell} -bT -1} w_{t-u} E_{u,ij} \right| > z \right] \\
    &\leq
    Q_\alpha \cdot p Tb \cdot \P\left[  \left| \frac{1}{|\theta_{\alpha,\ell} -t|^{1/2}} \sum_{j=1}^{q} \sum_{u=t-bT}^{\theta_{\alpha,\ell} -bT -1} w_{t-u} E_{u,ij} \right| > q z \right] 
    \tag{for any $i\in[p], \ell \in[Q_\alpha], t\in [\theta_{\alpha,\ell} -Tb, \theta_{\alpha,\ell} -1]$}
    \\
    &\leq
    Q_\alpha \cdot p Tb 
    \max_{t\in [\theta_{\alpha,\ell} -Tb, \theta_{\alpha,\ell} -1]}
    \left[ \frac{c_{e,1} (Tb)^{-\nu} q (\theta_{\alpha,\ell} -t)^{1 - \nu/2} }{q^\nu z^\nu} 
    + 8 \exp\left( 1- \frac{c_{e,2} q^2 T^2b^2 z^2}{q} \right) \right] \\
    &=
    Q_\alpha \left\{
    c_{e,1} p (Tb)^{1-\nu} q^{1-\nu} 
    z^{-\nu}
    + 8 \exp\left( 1 + \log(pTb) - c_{e,2} q T^2b^2 z^2 \right) \right\} ,
\end{align*}
for any $qz \geq a_0 \sqrt{q}/(Tb)$ (so as to apply Lemma~\ref{lemma: smooth_E_entry_bound_local}) such that $1 - c_{e,2} q T^2b^2 z^2 \leq 0$. By setting $\log(pTb) - c_{e,2} q T^2b^2 z^2$ as 0, and the first term above as some constant, we conclude the desired $\cO_P$-term and hence completed the proof.
\end{proof}

\begin{lemma}\label{lemma: smooth_E_entry_bound_local_general}
Let Assumptions~\ref{ass: noise} and \ref{ass: kernel} hold, with $\nu=\nu_e\geq 4$ and $\eta =\eta_e > 1/2 - 1/\nu$.
Let $m \in [1, bT]$ be any given integer, 
and also $\xi \in [1,bT]$ be any given integer.
Then we have
\begin{align*}
    &\;\quad
    \P\left[ \max_{\ell \in[Q_\alpha], t\in [\theta_{\alpha,\ell} -Tb, \theta_{\alpha,\ell} -\xi]} \left\| \frac{1}{q} \frac{1}{|\theta_{\alpha,\ell} -t|^{1/2}} \left( \sum_{u=t-bT}^{t -bT +m -1} w_{t-u} \bE_u \right) \1_q \right\|_\infty  > z \right] \\
    &\leq
    \max_{\ell \in[Q_\alpha], 
    t\in [\theta_{\alpha,\ell} -Tb, \theta_{\alpha,\ell} -\xi]}
    C_1 \Bigg[ \frac{C_2 p (Tb)^{-\nu} q^{1-\nu} (Tb -\xi) (\theta_{\alpha,\ell} -t)^{- \nu/2} m }{z^\nu}  \\
    &\;\quad
    + 8 \exp\left( 1 + \log(p(Tb -\xi)) - \frac{C_3 z^2 q (\theta_{\alpha,\ell} -t) T^2b^2}{m} \right) \Bigg] ,
\end{align*}
for any $z$ with $c_{e,2} q (\theta_{\alpha,\ell} -t) T^2b^2 z^2 \geq m$,
where $C_1$, $C_2$ and $C_3$ are some positive constants that only depending on the kernel and the distribution of components in the noise.
\end{lemma}

\begin{proof}[Proof of Lemma~\ref{lemma: smooth_E_entry_bound_local_general}]
We first show that for any $i\in[p]$,
$\ell \in [Q_\alpha]$,
$t$ such that $\theta_{\alpha,\ell} - Tb \leq t \leq \theta_{\alpha,\ell} -1$, 
when $a \geq a_0 \sqrt{qm/ (\theta_{\alpha,\ell} -t)} /(Tb)$ for some constant $a_0>0$, we have
\begin{align}
    \P\left(\left|
    \frac{1}{\sqrt{\theta_{\alpha,\ell} -t}}
    \sum_{j=1}^q \sum_{u=t-bT}^{t- bT + m-1} w_{t-u} E_{u,ij}  \right| > a \right) \notag
    &\leq 
    \frac{c_{e,1} (Tb)^{-\nu} q (\theta_{\alpha,\ell} -t)^{- \nu/2} m }{a^\nu}  \notag \\
    &\;\quad
    + 8 \exp\left( 1- \frac{c_{e,2} a^2 (\theta_{\alpha,\ell} -t) T^2b^2}{q m} \right) ,
    \label{eqn: smooth_E_entry_bound_local_general_step}
\end{align}
where $a_0$ satisfies that $c_{e,2} a^2 (\theta_{\alpha,\ell} -t) T^2b^2 /(qm) \geq 1$ for all $a\geq a_0$, and $c_{e,1}$ and $c_{e,2}$ are some positive constants that only depending on the kernel and the distribution of components in the noise.
The proof for \eqref{eqn: smooth_E_entry_bound_local_general_step} follows repetitive arguments in the proof of Lemma~\ref{lemma: smooth_E_entry_bound_local} with the same notations therein, except that we apply Theorem~2~(ii) of \cite{Liuetal2013} with sample size $n = qm$ and hence \eqref{eqn: sum_tilde_Ez_concen_local} is replaced by
\begin{align*}
    &\;\quad
    \P\left( \left| \sum_{z=1}^{qm} 
    \sqrt{\theta_{\alpha,\ell} -t}
    \wt{E}_{z} \right| > x \right)
    \leq
    \frac{C_1 c^\nu q m }{x^\nu} + 8 \exp\left( 1- \frac{C_2^2 x^2}{c^2 q m } \right) ,
\end{align*}
thus concluding \eqref{eqn: smooth_E_entry_bound_local_general_step}.
The remainder follows by modifying the proof of Lemma~\ref{lemma: smooth_E_vec1_bound_local}. Namely,
\begin{align*}
    &\;\quad
    \P\left[ \max_{\ell \in[Q_\alpha], t\in [\theta_{\alpha,\ell} -Tb, \theta_{\alpha,\ell} -\xi]} \left\| \frac{1}{q} \frac{1}{|\theta_{\alpha,\ell} -t|^{1/2}} \left( \sum_{u=t-bT}^{t -bT +m -1} w_{t-u} \bE_u \right) \1_q \right\|_\infty  > z \right] \\
    &=
    \P\left[ \max_{i\in[p], \ell \in[Q_\alpha], t\in [\theta_{\alpha,\ell} -Tb, \theta_{\alpha,\ell} -\xi]} \left| \frac{1}{q} \frac{1}{|\theta_{\alpha,\ell} -t|^{1/2}} \sum_{j=1}^{q} \sum_{u=t-bT}^{t -bT +m -1} w_{t-u} E_{u,ij} \right| > z \right] \\
    &\leq
    Q_\alpha \cdot p (Tb -\xi) \cdot \P\left[  \left| \frac{1}{|\theta_{\alpha,\ell} -t|^{1/2}} \sum_{j=1}^{q} \sum_{u=t-bT}^{t -bT +m -1} w_{t-u} E_{u,ij} \right| > q z \right] 
    \tag{for any $i\in[p], \ell \in[Q_\alpha], t\in [\theta_{\alpha,\ell} -Tb, \theta_{\alpha,\ell} -\xi]$}
    \\
    &\leq
    Q_\alpha \cdot p (Tb -\xi) 
    \max_{\ell \in[Q_\alpha],
    t\in [\theta_{\alpha,\ell} -Tb, \theta_{\alpha,\ell} -\xi]}
    \Bigg[ \frac{c_{e,1} (Tb)^{-\nu} q (\theta_{\alpha,\ell} -t)^{- \nu/2} m }{q^\nu z^\nu}  \\
    &\;\quad
    + 8 \exp\left( 1- \frac{c_{e,2} z^2 q (\theta_{\alpha,\ell} -t) T^2b^2}{m} \right) \Bigg] \\
    &=
    Q_\alpha \cdot \max_{\ell \in[Q_\alpha],
    t\in [\theta_{\alpha,\ell} -Tb, \theta_{\alpha,\ell} -\xi]}
    \Bigg[ \frac{c_{e,1} p (Tb)^{-\nu} q^{1-\nu} (Tb -\xi) (\theta_{\alpha,\ell} -t)^{- \nu/2} m }{z^\nu}  \\
    &\;\quad
    + 8 \exp\left( 1 + \log(p(Tb -\xi)) - \frac{c_{e,2} z^2 q (\theta_{\alpha,\ell} -t) T^2b^2}{m} \right) \Bigg] ,
\end{align*}
for any $qz \geq a_0 \sqrt{qm/ (\theta_{\alpha,\ell} -t)} /(Tb)$ (so as to apply \eqref{eqn: smooth_E_entry_bound_local_general_step} in the last inequality) such that $c_{e,2} q (\theta_{\alpha,\ell} -t) T^2b^2 z^2 \geq m$. 
This concludes the proof of this lemma.
\end{proof}

\begin{lemma}\label{lemma: smooth_E_entry_bound_local_ww}
Let Assumptions~\ref{ass: noise} and \ref{ass: kernel} hold, with $\nu=\nu_e\geq 4$ and $\eta =\eta_e > 1/2 - 1/\nu$.
Then we have
\begin{align*}
    &\;\quad
    \max_{\ell \in[Q_\alpha], t\in [\theta_{\alpha,\ell} -Tb +1, \theta_{\alpha,\ell} -1]} \left\| \frac{1}{q} \frac{1}{|\theta_{\alpha,\ell} -t|^{1/2}} \left( \sum_{u=\theta_{\alpha,\ell} +1 }^{t +bT} (w_{u-t} - w_{u- \theta_{\alpha,\ell}}) \bE_u \right) \1_q \right\|_\infty \\
    &=
    \cO_P\left( p^{1/\nu} (Tb)^{2(1-\nu)/\nu} q^{(1-\nu)/\nu} 
    + \sqrt{\frac{\log(pTb)}{q T^3b^3}} \right) .
\end{align*}
\end{lemma}

\begin{proof}[Proof of Lemma~\ref{lemma: smooth_E_entry_bound_local_ww}]
We first show that for any $i\in[p]$,
$\ell \in [Q_\alpha]$,
$t$ such that $\theta_{\alpha,\ell} - Tb +1 \leq t \leq \theta_{\alpha,\ell} -1$, 
when $a \geq a_0 \sqrt{q[Tb - (\theta_{\alpha,\ell} -t)] (\theta_{\alpha,\ell} -t)} /(T^2b^2)$ for some constant $a_0>0$, we have
\begin{align}
    &\;\quad
    \P\left(\left|
    \frac{1}{\sqrt{\theta_{\alpha,\ell} -t}}
    \sum_{j=1}^q \sum_{u=\theta_{\alpha,\ell} +1 }^{t +bT} (w_{u-t} - w_{u- \theta_{\alpha,\ell}}) E_{u,ij}  \right| > a \right) \notag \\
    &= 
    \frac{c_{e,1} (\theta_{\alpha,\ell} -t)^{\nu/2} (Tb)^{-2\nu} q  [Tb - (\theta_{\alpha,\ell} -t)] }{a^\nu} 
    + 8 \exp\left( 1- \frac{c_{e,2} a^2  T^4b^4}{q [Tb - (\theta_{\alpha,\ell} -t)] (\theta_{\alpha,\ell} -t)} \right) ,
    \label{eqn: smooth_E_entry_bound_local_step_ww}
\end{align}
where $a_0$ satisfies that $c_{e,2} a^2 (\theta_{\alpha,\ell} -t) T^2b^2 /(qm) \geq 1$ for all $a\geq a_0$, and $c_{e,1}$ and $c_{e,2}$ are some positive constants that only depending on the kernel and the distribution of components in the noise.
To this end, firstly, from Lemma~1 of \cite{Chenetal2022}, under Assumption~\ref{ass: kernel}, for any integer $m\in[0,Tb]$, the kernel weights $\{w_1,\dots, w_{Tb}\}$ satisfy that $\max_{|\iota -\ell| \leq m} |w_\iota - w_\ell| \leq c m/(Tb)^2$ for some strictly positive constant $c$ that only depends on the kernel $K$.
We then use analogous arguments as in the proof of Lemma~\ref{lemma: smooth_E_entry_bound}, except for 
$\wt{E}_z$ therein
is replaced by
(due to the scaling by $w_{u-t} - w_{u-\theta_{\alpha,\ell}}$):
\begin{align*}
    Tb (\theta_{\alpha,\ell} -t)^{-1/2} \cdot
    Tb \cdot (w_{u-t} - w_{u- \theta_{\alpha,\ell}}) \Sigma_{e, ij} \sum_{h\geq 0} a_{\epsilon,h} X_{\epsilon,u-h, ij} ,
\end{align*}
and we set $n=q [Tb - (\theta_{\alpha,\ell} -t)]$ when applying Theorem~2~(ii) of \cite{Liuetal2013}.
Then using notations in \eqref{eqn: sum_tilde_Ez_concen}, it holds that
with $x= (\theta_{\alpha,\ell} -t)^{-1/2} \cdot T^2b^2 (a/2)$,
there exist some positive constants $C_1$ and $C_2$ such that under $x \geq \sqrt{q[Tb - (\theta_{\alpha,\ell} -t)]}\, \Theta_{0,\nu}^* /C_2$ 
(which can be satisfied given $a \geq a_0 \sqrt{q[Tb - (\theta_{\alpha,\ell} -t)] (\theta_{\alpha,\ell} -t)} /(T^2b^2)$ for some constant $a_0>0$), we have
\begin{align*}
    &\;\quad
    \P\left( \left| \sum_{z=1}^{q[Tb - (\theta_{\alpha,\ell} -t)]} 
    Tb (\theta_{\alpha,\ell} -t)^{-1/2}
    \wt{E}_{z} \right| > x \right) \\
    &\leq
    \frac{C_1 c^\nu q [Tb - (\theta_{\alpha,\ell} -t)] }{x^\nu} + 8 \exp\left( 1- \frac{C_2^2 x^2}{c^2 q [Tb - (\theta_{\alpha,\ell} -t)] } \right) ,
\end{align*}
which implies \eqref{eqn: smooth_E_entry_bound_local_step_ww}.
The remainder follows by modifying the proof of Lemma~\ref{lemma: smooth_E_vec1_bound_local}. Namely,
\begin{align*}
    &\;\quad
    \P\left[ \max_{\ell \in[Q_\alpha], t\in [\theta_{\alpha,\ell} -Tb +1, \theta_{\alpha,\ell} -1]} \left\| \frac{1}{q} \frac{1}{|\theta_{\alpha,\ell} -t|^{1/2}} \left( \sum_{u=\theta_{\alpha,\ell} +1 }^{t +bT} (w_{u-t} - w_{u- \theta_{\alpha,\ell}}) \bE_u \right) \1_q \right\|_\infty  > z \right] \\
    &=
    \P\left[ \max_{i\in[p], \ell \in[Q_\alpha], t\in [\theta_{\alpha,\ell} -Tb +1, \theta_{\alpha,\ell} -1]} \left| \frac{1}{q} \frac{1}{|\theta_{\alpha,\ell} -t|^{1/2}} \sum_{j=1}^{q} \sum_{u=\theta_{\alpha,\ell} +1 }^{t +bT} (w_{u-t} - w_{u- \theta_{\alpha,\ell}}) E_{u,ij} \right| > z \right] \\
    &\lesssim
    p Tb \cdot \P\left[  \left| \frac{1}{|\theta_{\alpha,\ell} -t|^{1/2}} \sum_{j=1}^{q} \sum_{u=\theta_{\alpha,\ell} +1 }^{t +bT} (w_{u-t} - w_{u- \theta_{\alpha,\ell}}) E_{u,ij} \right| > q z \right] 
    \tag{for any $i\in[p], \ell \in[Q_\alpha], t\in [\theta_{\alpha,\ell} -Tb, \theta_{\alpha,\ell} -1]$}
    \\
    &\leq
    p Tb
    \max_{\ell \in[Q_\alpha],
    t\in [\theta_{\alpha,\ell} -Tb +1, \theta_{\alpha,\ell} -1]}
    \Bigg[ \frac{c_{e,1} (\theta_{\alpha,\ell} -t)^{\nu/2} (Tb)^{-2\nu} q  [Tb - (\theta_{\alpha,\ell} -t)] }{(qz)^\nu} \\
    &\;\quad
    + 8 \exp\left( 1- \frac{c_{e,2} (qz)^2  T^4b^4}{q [Tb - (\theta_{\alpha,\ell} -t)] (\theta_{\alpha,\ell} -t)} \right) \Bigg] \\
    &=
    \max_{\ell \in[Q_\alpha],
    t\in [\theta_{\alpha,\ell} -Tb +1, \theta_{\alpha,\ell} -1]}
    \Bigg[ \frac{c_{e,1} p (Tb)^{1-2\nu} q^{1-\nu}  [Tb - (\theta_{\alpha,\ell} -t)] (\theta_{\alpha,\ell} -t)^{\nu/2} }{z^\nu} \\
    &\;\quad
    + 8 \exp\left( 1 +\log(pTb) - \frac{c_{e,2} z^2 q T^4b^4}{[Tb - (\theta_{\alpha,\ell} -t)] (\theta_{\alpha,\ell} -t)} \right) \Bigg] \\
    &\leq
    \frac{c_{e,1} p (Tb)^{2-2\nu} q^{1-\nu} }{z^\nu} 
    + 8 \exp\left( 1 +\log(pTb) - c_{e,2} z^2 q T^3b^3 \right)  .
\end{align*}
This concludes the proof of this lemma.
\end{proof}

\subsection{Auxiliary results for the proof of Theorem~\ref{thm: consistency_change_fac}}

To simplify notations in this subsection, define $\bd_t := \vec(\bC_t + \bE_t)$ for any $t\in[T]$.

\begin{lemma}\label{lemma: max_ZhangWu}
Under Assumptions~\ref{ass: loadings}, \ref{ass: core_factor}, \ref{ass: noise}, and \ref{ass: change_loc}, for any $i,j\in[pq]$ and $a>0$, we have
\begin{align*}
    \P\Bigg( \Bigg| \sum_{t = 1}^{\gamma} \Big\{ \bd_t \bd_t^\trans - \E[\bd_t \bd_t^\trans] \Big\}_{i,j} \Bigg| > a \Bigg) 
    &\leq
    C_{\nu,\eta} a^{-\nu/2} (3 c_\eta)^{\nu} 
    \gamma
    + C \exp\left( -\frac{a^2}{81 C_\eta c_\eta^4 \gamma } \right) ,
\end{align*}
where 
$\nu=\min\{\nu_f, \nu_e\}$, $\eta=\min\{\eta_f, \eta_e\}$
are from Assumptions~\ref{ass: core_factor} and \ref{ass: noise}, and satisfy $\nu>4$ and 
$\eta> 1$, respectively,
while $C$, $C_{\nu,\eta}$, and $C_\eta$ are constants only depending on $(\nu,\eta)$.
\end{lemma}

\begin{proof}[Proof of \Cref{lemma: max_ZhangWu}]

For any $i\in[p], j\in[q]$, Assumption~\ref{ass: noise} allows us to write the $(i,j)$th element of $\bE_t$ as
\[
E_{t,ij} = \bA_{e,r, i\cdot}^\trans \bF_{e,t} \bA_{e,c, j\cdot} + \Sigma_{e, ij} E_{e,t, ij} ,
\]
in which case the corresponding entry in $\vec(\bE_t)$, denoted by $\vec(\bE_t)_{w}$, $w\in[pq]$, can be read as
\[
\vec(\bE_t)_{w} = (\bA_{e,c, j\cdot} \otimes \bA_{e,r, i\cdot})^\trans \vec(\bF_{e,t}) + \Sigma_{e, ij} E_{e,t, ij} .
\]

Similarly to \eqref{eqn: CC_limit}, we may also write the entry $\vec(\bC_t)_w$, $w\in[pq]$ as
\begin{align*}
    \vec(\bC_t)_w &= \vec\left( \sum_{k=1}^{Q_f +1} \bQ \bR_{k-1} \bF_t \bV_{k-1}^\trans \bU^\trans \cdot \mathbb{I}\{\theta_{f,k-1} <t\le \theta_{f,k}\} \right)_w \\
    &=
    (\bU \otimes \bQ)_{w\cdot}^\trans \sum_{k=1}^{Q_f +1} (\bV_{k-1}^\trans \otimes \bR_{k-1}) \vec(\bF_t) \cdot \mathbb{I}\{\theta_{f,k-1} <t\le \theta_{f,k}\} .
\end{align*}

We will apply Proposition~6.5 of \cite{ZhangWu2021} with $B=0$ and $T=\gamma$ therein.
By Assumption~\ref{ass: noise}, we can write the process $\{\bd_t\}$ (defined at the beginning of this subsection) in the form using notations in (1.1) of \cite{ZhangWu2021}, such that
\[
\bd_t \equiv \vec(\bC_t + \bE_t) = G(t/T,\cF_t) := \left( g_1(t/T, \cF_t), \dots, g_{pq}(t/T, \cF_t) \right) ,
\]
where $\mathcal{F}_t = \{((X_{f,t, ij})_{i\in[p], j\in[q]}, (X_{e,t, ij})_{i\in[p], j\in[q]}, (X_{\epsilon,t, ij})_{i\in[p], j\in[q]} )\}$ and any $g_w(t/T, \cF_t)$ (under the relationship $\bd_w = \vec(\bC_t + \bE_t)_w= C_{t,ij} + E_{t,ij}$) is
\begin{align}
    &\;\quad
    g_w(t/T, \cF_t) := \bd_w = \vec(\bC_t + \bE_t)_w  \notag \\
    &=
    (\bU_{j\cdot} \otimes \bQ_{i\cdot})^\trans \sum_{k=1}^{Q_f +1} (\bV_{k-1}^\trans \otimes \bR_{k-1}) \vec(\bF_t) \cdot \mathbb{I}\{\theta_{f,k-1} <t\le \theta_{f,k}\} \notag \\
    &\;\quad
    + (\bA_{e,c, j\cdot} \otimes \bA_{e,r, i\cdot})^\trans \vec(\bF_{e,t}) + \Sigma_{e, ij} E_{e,t, ij} \notag \\
    &=
    (\bU_{j\cdot} \otimes \bQ_{i\cdot})^\trans \sum_{k=1}^{Q_f +1} (\bV_{k-1}^\trans \otimes \bR_{k-1}) \sum_{s\geq 0} a_{f,s} \vec(\bX_{f,t-s}) \cdot \mathbb{I}\{\theta_{f,k-1} <t\le \theta_{f,k}\} \notag \\
    &\;\quad
    \sum_{s\geq 0} a_{e,s} (\bA_{e,c, j\cdot} \otimes \bA_{e,r, i\cdot})^\trans \vec(\bX_{e,t-s}) + \Sigma_{e, ij} \cdot \sum_{s\geq 0} a_{\epsilon,s} X_{\epsilon, t -s, ij} .
    \label{eqn: ZhangWu_g_Ft}
\end{align}
To verify the condition of uniform stochastic Lipschitz continuity (cf.\ Equation~(2.1) in \cite{ZhangWu2021}), note that it is automatically satisfied for the part $\{\vec(\bE_t)\}$ in our considered process since $\{\bE_t\}$ is strictly stationary, and it is also straightforward for the part $\{\vec(\bC_t)\}$ which is piecewise strictly stationary (with non-stationarity attributed solely to the heteroscedasticity) by Assumptions~\ref{ass: loadings}~(i), \ref{ass: core_factor} and \ref{ass: change_loc}.

Next, we compute the element-wise functional dependence measure (cf.\ (2.1) in \cite{ZhangWu2021}) of the process in \eqref{eqn: ZhangWu_g_Ft}. By Minkowski's inequality, for any $1 \leq q^* \leq \nu$ and $w\in[pq]$ (with the relationship $\bd_w = \vec(\bC_t + \bE_t)_w= C_{t,ij} + E_{t,ij}$), we have
\begin{align*}
    \delta_{t,q*,w} &:= \Bigg\{ \E \Bigg( \Bigg| 
    (\bU_{j\cdot} \otimes \bQ_{i\cdot})^\trans \sum_{k=1}^{Q_f +1} (\bV_{k-1}^\trans \otimes \bR_{k-1}) a_{f,t} \vec(\bX_{f,0} - \bX_{f,0}') \cdot \mathbb{I}\{\theta_{f,k-1} <t\le \theta_{f,k}\} \notag \\
    &\;\quad
    + a_{e,t} (\bA_{e,c, j\cdot} \otimes \bA_{e,r, i\cdot})^\trans \vec(\bX_{e,0} - \bX_{e,0}') + \Sigma_{e, ij} a_{\epsilon,t} (X_{\epsilon, 0, ij} - X_{\epsilon, 0, ij}') \Bigg|^{q^*} \Bigg) \Bigg\}^{1/q^*} \notag \\
    &\lesssim
    a_{f,t} + a_{e,t} + a_{\epsilon,t} ,
\end{align*}
which then gives a dependence adjusted norm (d.a.n.), for any $1 \leq \alpha^* \leq \eta$, as
\begin{align*}
    \Big\| \vec((\bC + \bE)_{\cdot w}) \Big\|_{q^*, \alpha^*} := \sup_{m\geq 0} (m+1)^{\alpha^*} \sum_{t=m}^\infty \delta_{t,q*,w} \leq 3 c_\eta ,
\end{align*}
uniformly over $w\in[pq]$. The uniform d.a.n.\ (cf.\ (2.5) in \cite{ZhangWu2021}) is thus
\[
\Phi_{q^*, \alpha^*} := \max_{1\leq w \leq pq} \Big\| \vec((\bC + \bE)_{\cdot w}) \Big\|_{q^*, \alpha^*}
\leq 3 c_\eta .
\]

Hereafter, we directly take $(q^*, \alpha^*) = (\nu, \eta)$. Note that for any $\nu\geq 0$, we may read $D_{T,B}$ in Proposition~6.5 of \cite{ZhangWu2021} as 
$D_{T,B} = \gamma$ due to $\eta>1$.
Then combining all the above quantities and Proposition~6.5 of \cite{ZhangWu2021} with $B=0$ and $T=\gamma$, we have for any $w,z\in[pq]$ and $a>0$,
\begin{align*}
    &\;\quad
    \P\left( \left| \sum_{t = 1}^{\gamma} \Big\{ \bd_t \bd_t^\trans - \E[\bd_t \bd_t^\trans] \Big\}_{w,z} \right| > a \right) 
    =
    \P\left( \left| \sum_{t = 1}^{\gamma} \Big\{ (\bd_t)_w (\bd_t)_z - \E[(\bd_t)_w (\bd_t)_z] \Big\} \right| > a \right) \\
    &\leq
    C_{\nu,\eta} a^{-\nu/2} (3 c_\eta)^{\nu} 
    \gamma
    + C \exp\left( -\frac{a^2}{81 C_\eta c_\eta^4 \gamma } \right) ,
\end{align*}
for some constants $C_{\nu,\eta}$, $C$, and $C_\eta$.
This completes the proof of this lemma.
\end{proof}

\begin{lemma}\label{lemma: rate_vech_GG}
Let all the assumptions in Lemma~\ref{lemma: max_ZhangWu} hold. Recalling that $\bd_t = \vec(\bC_t + \bE_t)$ for any $t\in[T]$, under the condition on $\gamma$ in Theorem~\ref{thm: consistency_change_fac}, 
we have
\begin{align*}
    & \max_{0\leq \tau\leq T-\gamma} \frac{1}{\sqrt{\gamma}} \Bigg\| \sum_{t=\tau+1}^{\tau+\gamma} (qp)^{-2} \L_{k_rk_c} \left\{ (\wh\bU \otimes \wh\bQ) \otimes (\wh\bU \otimes \wh\bQ) \right\}^\trans \left\{ (\bM_p \otimes \bM_q) \otimes (\bM_p \otimes \bM_q) \right\} \notag \\
    &\quad\quad\quad\quad\quad\quad\quad\quad\quad\quad\quad\quad\quad\quad\quad\quad\quad\quad\quad
    \cdot \vec\Big\{ \bd_t \bd_t^\trans - \E[\bd_t \bd_t^\trans] \Big\} \Bigg\|
    = \cO_P\left( \sqrt{\log(T)} \right) .
\end{align*}
\end{lemma}

\begin{proof}[Proof of \Cref{lemma: rate_vech_GG}]
First, for any $\delta>0$, we have
\begin{align*}
    &\;\quad
    \E\Bigg[ \bigg( \max_{0\leq \tau\leq T-\gamma} \frac{1}{\sqrt{\gamma}} \bigg\| \sum_{t=\tau+1}^{\tau+\gamma} \Big\{ \bd_t \bd_t^\trans - \E[\bd_t \bd_t^\trans] \Big\} \bigg\|_F \bigg)^2 \Bigg] \\
    &=
    \sum_{i,j=1}^{pq} \E\Bigg( \max_{0\leq \tau\leq T-\gamma} 
    \left| \frac{1}{\sqrt{\gamma}} \sum_{t=\tau+1}^{\tau+\gamma} \Big\{ \bd_t \bd_t^\trans - \E[\bd_t \bd_t^\trans] \Big\}_{i,j} \right|^2 \Bigg) \\
    &=
    \sum_{i,j=1}^{pq} \int_{0}^\infty \P\Bigg( \max_{0\leq \tau\leq T-\gamma}
    \Bigg| \frac{1}{\sqrt{\gamma}} \sum_{t=\tau+1}^{\tau+\gamma} \Big\{ \bd_t \bd_t^\trans - \E[\bd_t \bd_t^\trans] \Big\}_{i,j} \Bigg|^2 > z \Bigg) \mathrm{d}z \\
    &=
    \sum_{i,j=1}^{pq} \int_{0}^\delta \P\Bigg( \max_{0\leq \tau\leq T-\gamma} 
    \Bigg| \frac{1}{\sqrt{\gamma}} \sum_{t=\tau+1}^{\tau+\gamma} \Big\{ \bd_t \bd_t^\trans - \E[\bd_t \bd_t^\trans] \Big\}_{i,j} \Bigg|^2 > z \Bigg) \mathrm{d}z \\
    &\;\quad
    + \sum_{i,j=1}^{pq} \int_{\delta}^\infty \P\Bigg( \max_{0\leq \tau\leq T-\gamma} 
    \Bigg| \frac{1}{\sqrt{\gamma}} \sum_{t=\tau+1}^{\tau+\gamma} \Big\{ \bd_t \bd_t^\trans - \E[\bd_t \bd_t^\trans] \Big\}_{i,j} \Bigg|^2 > z \Bigg) \mathrm{d}z \\
    &\leq
    (pq)^2 \delta 
    + \sum_{i,j=1}^{pq} \int_{\delta}^\infty
    \P\Bigg( \max_{0\leq \tau\leq T-\gamma} 
    \Bigg| \frac{1}{\sqrt{\gamma}} \sum_{t=\tau+1}^{\tau+\gamma} \Big\{ \bd_t \bd_t^\trans - \E[\bd_t \bd_t^\trans] \Big\}_{i,j} \Bigg|^2 > z \Bigg)
    \mathrm{d}z \\
    &\lesssim
    (pq)^2 \delta 
    + \sum_{i,j=1}^{pq} \int_{\delta}^\infty
    (T - \gamma) 
    \cdot \P\Bigg( \Bigg| \sum_{t = 1}^{\gamma} \Big\{ \bd_t \bd_t^\trans - \E[\bd_t \bd_t^\trans] \Big\}_{i,j} \Bigg| > \sqrt{z \gamma } \Bigg)
    \mathrm{d}z \\
    &\lesssim
    (pq)^2 \delta 
    + \sum_{i,j=1}^{pq} \int_{\delta}^\infty
    (T - \gamma) \cdot (z \gamma)^{-\nu/4} 
    \gamma
    \mathrm{d}z \\
    &\;\quad
    + \sum_{i,j=1}^{pq} \int_{\delta}^\infty
    (T - \gamma) \cdot \exp\left( -\frac{z \gamma}{81 C_\eta c_\eta^4 \gamma } \right)
    \mathrm{d}z \\
    &\lesssim
    (pq)^2 \left\{ \delta
    + \int_{\delta}^\infty T z^{-\nu/4} 
    \gamma^{1- \nu/4}
    \mathrm{d}z
    + \int_{\delta}^\infty T \exp\left( -\frac{z}{81 C_\eta c_\eta^4 } \right) \mathrm{d}z \right\} \\
    &\asymp
    (pq)^2 \left\{ \delta
    + T 
    \gamma^{1- \nu/4}
    \delta^{1 -\nu/4}
    + T \exp\left( -\frac{\delta}{81 C_\eta c_\eta^4 } \right) \right\} ,
\end{align*}
where the third and fourth last lines used Lemma~\ref{lemma: max_ZhangWu}. Since the above holds for any $\delta\geq 0$, by taking $\delta = 81 C_\eta c_\eta^4 \log(T/\log(T))$, we have
\begin{align*}
    &\;\quad
    \E\Bigg[ \bigg( \max_{0\leq \tau\leq T-\gamma} \frac{1}{\sqrt{\gamma}} \bigg\| \sum_{t=\tau+1}^{\tau+\gamma} \Big\{ \bd_t \bd_t^\trans - \E[\bd_t \bd_t^\trans] \Big\} \bigg\|_F \bigg)^2 \Bigg] \\
    &\lesssim
    (pq)^2 \left( \log(T)
    + \frac{T}{\log(T)^{\nu/4} \gamma^{\nu/4 - 1}} \log(T)
    + \log(T) \right)
    \lesssim (pq)^2 \log(T) ,    
\end{align*}
by the condition on $\gamma$ in~\eqref{cond:thm: consistency_change_fac:one} of Theorem~\ref{thm: consistency_change_fac}.
Then together with the fact that we have
\[
\|\L_{k_rk_c}\|, \big\| p^{-1/2} \wh\bQ \big\|, \big\| q^{-1/2} \wh\bU \big\|, \|\bM_p\|, \|\bM_q\| = \cO_P(1) ,
\]
by the definitions and \eqref{eqn: M_norm}, the proof of this lemma is concluded.
\end{proof}

\begin{lemma}\label{lemma: U_kron_Q_converge_I}
Let all the assumptions in Theorem~\ref{thm: consistency_proj} hold. We have
\[
\left\| \frac{\bU^\trans \wh\bU \wh\bU^\trans \bU}{q^2} \otimes \frac{\bQ^\trans \wh\bQ \wh\bQ^\trans \bQ}{p^2} - \bI_{k_rk_c} \right\| = o_P(1) .
\]
\end{lemma}

\begin{proof}[Proof of \Cref{lemma: U_kron_Q_converge_I}]

First, we may write
\[
\frac{1}{q} \wh\bU^\trans \bU = \frac{(\wh\bU - \bU \wh\bH_c)^\trans}{\sqrt{q}} \frac{\bU}{\sqrt{q}} + \wh\bH_c^\trans \left( \frac{\bU^\trans \bU}{q} - \bI_{k_c} \right) + \wh\bH_c^\trans ,
\]
which, by Theorem~\ref{thm: consistency_proj} and Assumption~\ref{ass: loadings}~(ii), gives
\begin{equation}
\label{eqn: hat_U_U_converge}
\left\| \frac{1}{q} \wh\bU^\trans \bU - \wh\bH_c^\trans \right\| = o_P(1) ,
\text{\ and similarly \ }
\left\| \frac{1}{p} \wh\bQ^\trans \bQ - \wh\bH_r^\trans \right\| = o_P(1) .
\end{equation}

Note that we also have
\begin{align*}
    &\;\quad
    \frac{\bU^\trans \wh\bU \wh\bU^\trans \bU}{q^2} \otimes \frac{\bQ^\trans \wh\bQ \wh\bQ^\trans \bQ}{p^2} \\
    &=
    \Bigg\{ \frac{\bU^\trans \wh\bU}{q} \Bigg( \frac{1}{q} \wh\bU^\trans \bU - \wh\bH_c^\trans \Bigg) \Bigg\} \otimes \frac{\bQ^\trans \wh\bQ \wh\bQ^\trans \bQ}{p^2}
    + \Bigg\{ \Bigg( \frac{\bU^\trans \wh\bU}{q} - \wh\bH_c \Bigg) \wh\bH_c^\trans \Bigg\} \otimes \frac{\bQ^\trans \wh\bQ \wh\bQ^\trans \bQ}{p^2} \\
    &\;\quad
    + (\wh\bH_c \wh\bH_c^\trans - \bI_{k_c}) \otimes \frac{\bQ^\trans \wh\bQ \wh\bQ^\trans \bQ}{p^2}
    + \bI_{k_c} \otimes \Bigg\{ \frac{\bQ^\trans \wh\bQ}{p} \Bigg( \frac{\wh\bQ^\trans \bQ}{p} - \wh\bH_r^\trans \Bigg) \Bigg\} \\
    &\;\quad
    + \bI_{k_c} \otimes \Bigg\{ \Bigg( \frac{\bQ^\trans \wh\bQ}{p} - \wh\bH_r \Bigg) \wh\bH_r^\trans \Bigg\} 
    + \bI_{k_c} \otimes (\wh\bH_r \wh\bH_r^\trans - \bI_{k_r}) 
    + \bI_{k_r k_c} ,
\end{align*}
which concludes the proof of this lemma by combining the triangle inequality, \eqref{eqn: hat_U_U_converge}, Theorem~\ref{thm: consistency_proj}, and Assumption~\ref{ass: loadings}~(ii).
\end{proof}

\begin{lemma}\label{lemma: ZZ_2+epsilon}
Let Assumptions~\ref{ass: loadings}, \ref{ass: core_factor}, \ref{ass: noise}, and \ref{ass: change_loc} hold, with $\nu = \min\{\nu_f, \nu_e\} >4$, $\eta = \min\{\eta_f, \eta_e\} > 1$.
Then for the constant $\epsilon$ satisfying $\epsilon \in (0, \nu/2 - 2)$, we have for any $i\in[p^2 q^2]$ and any $1\leq s<e \leq T$:
\begin{equation*}
\E\left( \left\vert \sum_{t = s + 1}^e \left[ \left\{ (\bM_p \otimes \bM_q) \otimes (\bM_p \otimes \bM_q) \right\} \vec\Big\{ \bd_t \bd_t^\trans - \E[\bd_t \bd_t^\trans] \Big\} \right]_i \right\vert^{2 + \epsilon} \right) \le c_1 (e - s)^{1 + \epsilon/2} .
\end{equation*}
where $c_1 >0 $ is some constant that only depends on $\nu$.
\end{lemma}

\begin{proof}[Proof of Lemma~\ref{lemma: ZZ_2+epsilon}]
We first show that for any $i,j\in[pq]$ and any $1\leq s<e \leq T$:
\begin{equation}
\E\left( \left\vert \sum_{t = s + 1}^e \left[ (\bd_t)_i (\bd_t)_j - \E\{ (\bd_t)_i (\bd_t)_j \} \right] \right\vert^{2 + \epsilon} \right) \le c_1' (e - s)^{1 + \epsilon/2} ,
\label{eqn: ZZ_2+epsilon_part}
\end{equation}
where $\epsilon$ is the same constant from the statement of this lemma, and $c_1'$ is some constant that only depends on $\nu$.
In view of \eqref{eqn: Gt_define}, Assumptions~\ref{ass: loadings}~(i), \ref{ass: core_factor} (with $\nu>4$, $\eta> 1$), \ref{ass: noise} (with $\nu>4$, $\eta> 1$), and \ref{ass: change_loc},
the above result can be shown analogously (and hence omitted) by applying Proposition~4 of \cite{Berkesetal2011} as in the proof of Lemma~B.22 in \cite{ZhangCenCho2026}.
Next, note that for any positive integer $a$, the matrix $\bM_a$ satisfies
\begin{equation}
\label{eqn: M_norm_infty}
\|\bM_a\|_1 = \|\bM_a\|_{\infty} = 2(a-1)/a \leq 2 .
\end{equation}
Hence with Lemma~B.15~(i) of \cite{ZhangCenCho2026}, the arguments for \eqref{eqn: ZZ_2+epsilon_part} hold similarly.
\end{proof}

\begin{lemma}\label{lemma: max_ZZ_2+epsilon}
Let all assumptions in Lemma~\ref{lemma: ZZ_2+epsilon} hold.
Then for the same constant $\epsilon$ in \Cref{lemma: ZZ_2+epsilon}, for any $\kappa_T\to\infty$ arbitrarily slowly, define the event
\begin{align*}
    \cM_T^{(z)} &:= \Bigg\{ \max_{1 \leq w \leq Q_f} \max_{\omega_{f,w}^{-2} \kappa_T^2 \leq \ell \leq \theta_{f,w} - \theta_{f,w - 1}} \frac{\sqrt{\omega_{f,w}^{-2}\kappa_T^2}}{\ell}
    \Bigg\Vert  \sum_{t = \theta_{f,w} + z\gamma - \ell + 1}^{\theta_{f,w} + z\gamma} \frac{\L_{k_rk_c}}{p^2 q^2} \left\{ (\wh\bU \otimes \wh\bQ) \otimes (\wh\bU \otimes \wh\bQ) \right\}^\trans \\
    &\quad\quad\quad\quad\quad\quad
    \quad\quad\quad\quad\quad\quad
    \cdot \Big\{ (\bM_p \otimes \bM_q) \otimes (\bM_p \otimes \bM_q) \Big\} \vec\Big\{ \bd_t \bd_t^\trans - \E[\bd_t \bd_t^\trans] \Big\} \Bigg\Vert \geq \kappa_T \Bigg\} ,
\end{align*}
for $z\in\{-1,0,1\}$. Then as $\min\{T, p, q\} \to \infty$, it holds that
\begin{equation*}
\begin{split}
    &
    \sum_{z\in\{-1,0,1\}} \P\left( \cM_T^{(z)} \right) \to 0 .
\end{split}
\end{equation*}
\end{lemma}

\begin{proof}[Proof of Lemma~\ref{lemma: max_ZZ_2+epsilon}]
Consider the event $\cM_T^{(0)}$.
To ease notations, for any $t\in[T]$, define
\begin{align*}
    \wh\bXi &:= \frac{\L_{k_rk_c}}{p^2 q^2} \left\{ (\wh\bU \otimes \wh\bQ) \otimes (\wh\bU \otimes \wh\bQ) \right\}^\trans , \\
    \wt\bd_t &:= \Big\{ (\bM_p \otimes \bM_q) \otimes (\bM_p \otimes \bM_q) \Big\} \vec\Big\{ \bd_t \bd_t^\trans - \E[\bd_t \bd_t^\trans] \Big\} .
\end{align*}
Immediately, by Cauchy--Schwarz inequality and the definitions of $\wh\bU$ and $\wh\bQ$, we have 
\begin{align*}
    \|\wh\bXi\|_\infty^2 &= \max_{i\in[k_r^2 k_c^2]} \left( \sum_{j=1}^{p^2 q^2} |(\wh\bXi)_{ij}| \right)^2
    \leq
    p^2 q^2 \max_{i\in[k_r^2 k_c^2]} \sum_{j=1}^{p^2 q^2} (\wh\bXi)_{ij}^2 \lesssim 1 .
\end{align*}
Then by Minkowski's inequality, we have for any $w\in[Q_f]$ and $i\in [k_rk_c (k_rk_c +1)/2]$,
\begin{align}
    &\;\quad
    \E\left[ \left( \max_{\omega_{f,w}^{-2} \kappa_T^2 \leq \ell \leq \theta_{f,w} - \theta_{f,w - 1}} \frac{\sqrt{\omega_{f,w}^{-2}\kappa_T^2}}{\ell} \left\vert  \sum_{t = \theta_{f,w} - \ell + 1}^{\theta_{f,w}} (\wh\bXi \wt\bd_t)_i \right\vert \right)^{2 + \epsilon} \right] \notag \\
    &=
    \E\left[ \left( \max_{\omega_{f,w}^{-2} \kappa_T^2 \leq \ell \leq \theta_{f,w} - \theta_{f,w - 1}} \frac{\sqrt{\omega_{f,w}^{-2}\kappa_T^2}}{\ell} \left\vert  \sum_{j=1}^{p^2 q^2} (\wh\bXi)_{i,j} \sum_{t = \theta_{f,w} - \ell + 1}^{\theta_{f,w}} (\wt\bd_t)_j \right\vert \right)^{2 + \epsilon} \right] \notag \\
    &\leq
    \E\left[ \left( \max_{\omega_{f,w}^{-2} \kappa_T^2 \leq \ell \leq \theta_{f,w} - \theta_{f,w - 1}} \frac{\sqrt{\omega_{f,w}^{-2}\kappa_T^2}}{\ell} 
    \left( \sum_{j=1}^{p^2 q^2} |(\wh\bXi)_{i,j}| \right)
    \max_{j\in [p^2 q^2]}
    \left\vert  \sum_{t = \theta_{f,w} - \ell + 1}^{\theta_{f,w}} (\wt\bd_t)_j \right\vert \right)^{2 + \epsilon} \right] \notag \\
    &\lesssim
    \E\left[ \left( \max_{j\in [p^2 q^2]} \max_{\omega_{f,w}^{-2} \kappa_T^2 \leq \ell \leq \theta_{f,w} - \theta_{f,w - 1}} \frac{\sqrt{\omega_{f,w}^{-2}\kappa_T^2}}{\ell} 
    \left\vert  \sum_{t = \theta_{f,w} - \ell + 1}^{\theta_{f,w}} (\wt\bd_t)_j \right\vert \right)^{2 + \epsilon} \right] \notag \\
    &\leq
    \E\left[ \left( \max_{j\in [p^2 q^2]} \max_{\omega_{f,w}^{-2} \kappa_T^2 \leq \ell \leq \theta_{f,w} - \theta_{f,w - 1}} \frac{\sqrt{\omega_{f,w}^{-2} \kappa_T^2}}{\ell} \left\vert  \sum_{t = \theta_{f,w} - \ell + 1}^{\theta_{f,w} - \omega_{f,w}^{-2}\kappa_T^2 + 1} (\wt\bd_t)_j \right\vert  \right)^{2 + \epsilon} \right] \notag \\
    &\;\quad
    + \E\left[ \left( \max_{j\in [p^2 q^2]} \max_{\omega_{f,w}^{-2} \kappa_T^2 \leq \ell \leq \theta_{f,w} - \theta_{f,w - 1}} \frac{\sqrt{\omega_{f,w}^{-2} \kappa_T^2}}{\ell} \left\vert  \sum_{t = \theta_{f,w} - \omega_{f,w}^{-2} \kappa_T^2 + 2}^{\theta_{f,w}} (\wt\bd_t)_j \right\vert  \right)^{2 + \epsilon} \right] \notag \\
    &\lesssim
    \left( \omega_{f,w}^{-2} \kappa_T^2 \right)^{1+\epsilon/2} \sum_{\ell= \omega_{f,w}^{-2} \kappa_T^2}^{\theta_{f,w} - \theta_{f,w - 1}} \frac{1}{\ell^{2+\epsilon}} \cdot \ell^{\epsilon/2}
    + \left\{ \left( \omega_{f,w}^{-2} \kappa_T^2 \right)^{-1} \left( \omega_{f,w}^{-2} \kappa_T^2 -1 \right) \right\}^{1+\epsilon/2}
    = \cO(1) ,
    \label{eqn: max_ZZ_2+epsilon_entry}
\end{align}
where we used Lemma~\ref{lemma: ZZ_2+epsilon} and applied Theorem~B.3 in \cite{Kirch2006} to obtain the last line from the second last; the first term in the last line can be upper bounded by an integral inequality; and the second term is trivially bounded, with $\cO(1)$ that does not depend on $i$ nor $w$. 

Then by Jensen's inequality, we have for any $w\in[Q_f]$:
\begin{align*}
    &\;\quad
    \E\left[ \left( \max_{\omega_{f,w}^{-2} \kappa_T^2 \leq \ell \leq \theta_{f,w} - \theta_{f,w - 1}} \frac{\sqrt{\omega_{f,w}^{-2}\kappa_T^2}}{\ell} \left\Vert \sum_{t = \theta_{f,w} - \ell + 1}^{\theta_{f,w}} \wh\bXi \wt\bd_t \right\Vert  \right)^{2 + \epsilon} \right] \\
    &\lesssim
    \E\left[ \left( \max_{\omega_{f,w}^{-2} \kappa_T^2 \leq \ell \leq \theta_{f,w} - \theta_{f,w - 1}} \frac{1}{k_rk_c (k_rk_c +1)/2} \sum_{i=1}^{k_rk_c (k_rk_c +1)/2} \left\vert \frac{\sqrt{\omega_{f,w}^{-2}\kappa_T^2}}{\ell} \sum_{t = \theta_{f,w} - \ell + 1}^{\theta_{f,w}} (\wh\bXi \wt\bd_t)_i \right\vert^2 \right)^{1 + \epsilon/2} \right] \\
    &\leq
    \frac{1}{k_rk_c (k_rk_c +1)/2} \sum_{i=1}^{k_rk_c (k_rk_c +1)/2}
    \E\left[ \max_{\omega_{f,w}^{-2} \kappa_T^2 \leq \ell \leq \theta_{f,w} - \theta_{f,w - 1}} \left\vert \frac{\sqrt{\omega_{f,w}^{-2}\kappa_T^2}}{\ell} \sum_{t = \theta_{f,w} - \ell + 1}^{\theta_{f,w}} (\wh\bXi \wt\bd_t)_i \right\vert^{2 + \epsilon} \right] \\
    &= \cO(1) ,
\end{align*}
where the last equality used~\eqref{eqn: max_ZZ_2+epsilon_entry}. Together with the Markov's inequality, we have
\begin{align*}
    &\;\quad
    \P\big( \cM_T^{(0)} \big) =
    \P\Bigg\{ \max_{1 \leq w \leq Q_f} \max_{\omega_{f,w}^{-2} \kappa_T^2 \leq \ell \leq \theta_{f,w} - \theta_{f,w - 1}} \frac{\sqrt{\omega_{f,w}^{-2}\kappa_T^2}}{\ell}
    \Bigg\Vert  \sum_{t = \theta_{f,w} - \ell + 1}^{\theta_{f,w}} \wh\bXi \wt\bd_t \Bigg\Vert \geq \kappa_T \Bigg\}  \\
    &\leq
    \frac{Q_f}{(\kappa_T)^{2+\epsilon}} \max_{1 \leq w \leq Q_f} 
    \E\Bigg[ \Bigg( \max_{\omega_{f,w}^{-2} \kappa_T^2 \leq \ell \leq \theta_{f,w} - \theta_{f,w - 1}} \frac{\sqrt{\omega_{f,w}^{-2}\kappa_T^2}}{\ell} \Bigg\Vert \sum_{t = \theta_{f,w} - \ell + 1}^{\theta_{f,w}} \wh\bXi \wt\bd_t \Bigg\Vert  \Bigg)^{2 + \epsilon} \Bigg] 
    \lesssim
    \frac{Q_f}{(\kappa_T)^{2+\epsilon}} = o(1),
\end{align*}
where the last equality holds as $Q_f$ is fixed due to linearly spaced change points and $\kappa_T$ goes to infinity (arbitrarily slowly).
By repetitive arguments, we may also conclude $\P\big( \cM_T^{(-1)} \big) = o(1)$ and $\P\big( \cM_T^{(1)} \big) = o(1)$, thus completing the proof of this lemma.
\end{proof}

\newpage
\subsection{Proof of theorems}

\begin{proof}[Proof of \Cref{thm: consistency_proj}]
Let us define
\begin{align*}
    \wh\bD_r &:= \frac{1}{Tp^2q^2} \wh\bQ^\trans \sum_{t=1}^T \wh\bL_t \wt\bU \wt\bU^\trans \wh\bL_t^\trans \wh\bQ , \quad
    \wh\bH_r := \frac{1}{Tpq^2} \sum_{t=1}^T \bG_t \bU^\trans \wt\bU \wt\bU^\trans \bU \bG_t^\trans \bQ^\trans \wh\bQ \wh\bD_r^{-1}, \\
    \wh\bD_c &:= \frac{1}{Tp^2q^2} \wh\bU^\trans \sum_{t=1}^T \wh\bL_t^\trans \wt\bQ \wt\bQ^\trans \wh\bL_t \wh\bU , \quad
    \wh\bH_c := \frac{1}{Tp^2q} \sum_{t=1}^T \bG_t^\trans \bQ^\trans \wt\bQ \wt\bQ^\trans \bQ \bG_t \bU^\trans \wh\bU \wh\bD_c^{-1} .
\end{align*}

We only show the results for $\wh\bQ$ and $\wh\bH_r$, and those for $\wh\bU$ and $\wh\bH_c$ follow analogously.
First, note that by the definition, $\wh\bD_r$ is the diagonal matrix consisting of the leading $k_r$ eigenvalues of $(Tpq^2)^{-1} \sum_{t=1}^T \wh\bL_t \wt\bU \wt\bU^\trans \wh\bL_t^\trans$. We then prove the statement:
\begin{itemize}
    \item [(D1)] When $\min\{T,p,q\} \to \infty$, the eigenvalues of $\wh\bD_r$ satisfy $\lambda_i(\wh\bD_r) = \lambda_i(\bSigma_{R}) + o_P(1)$ for any $i\leq k_r$, where $\bSigma_{R}$ is defined in Lemma~\ref{lemma: eigenvalue_PCA_consistency}~(i). 
\end{itemize}
To see this, similar to \eqref{eqn: sample_covariance_row}, consider
\begin{align}
    &\;\quad
    \wh\bL_t \wt\bU \wt\bU^\trans \wh\bL_t^\trans = \bM_p \bX_t \bM_q \wt\bU \wt\bU^\trans \bM_q \bX_t^\trans \bM_p
    = \bM_p (\bC_t + \bE_t) \bM_q \wt\bU \wt\bU^\trans \bM_q (\bC_t^\trans + \bE_t^\trans) \bM_p \notag \\
    &=
    \bM_p \bC_t \bM_q \wt\bU \wt\bU^\trans \bM_q \bC_t^\trans \bM_p
    + \bM_p \bE_t \bM_q \wt\bU \wt\bU^\trans \bM_q \bC_t^\trans \bM_p 
    + \bM_p \bC_t \bM_q \wt\bU \wt\bU^\trans \bM_q \bE_t^\trans \bM_p \notag \\
    &\;\quad
    + \bM_p \bE_t \bM_q \wt\bU \wt\bU^\trans \bM_q \bE_t^\trans \bM_p
    =: \cI_{Q,t} + \cII_{Q,t} + \cIII_{Q,t} + \cIV_{Q,t} .
    \label{eqn: LUUL_decomp}
\end{align}
Further simplify $\cI_{Q,t}$ as
\begin{align*}
    \cI_{Q,t} &= (\bI_p - p^{-1} \1_p \1_p^\trans) \bC_t \bM_q \wt\bU \wt\bU^\trans \bM_q \bC_t^\trans (\bI_p - p^{-1} \1_p \1_p^\trans)
    = \bC_t \bM_q \wt\bU \wt\bU^\trans \bM_q \bC_t^\trans \\
    &=
    \bC_t (\bI_q - q^{-1} \1_q \1_q^\trans) \wt\bU \wt\bU^\trans (\bI_q - q^{-1} \1_q \1_q^\trans) \bC_t^\trans
    = \bC_t \wt\bU \wt\bU^\trans \bC_t^\trans ,
\end{align*}
where we used throughout Assumption~\ref{ass: identification} that $\bC_t \1_q=\0$ and $\bC_t^\trans \1_p=\0$. Note that similarly,
\[
\cII_{Q,t} = \bM_p \bE_t \bM_q \wt\bU \wt\bU^\trans \bC_t^\trans
\text{ \ and \ }
\cIII_{Q,t} = \bC_t \wt\bU \wt\bU^\trans \bM_q \bE_t^\trans \bM_p .
\]

Then analogous to \eqref{eqn: CC_limit}, by Assumptions~\ref{ass: loadings}, \ref{ass: core_factor}, \ref{ass: change_loc}~(iii), Lemma~\ref{lemma: FMF_expectation}, and the fact that $\|\wt\bU\|_F^2=q$ by definition, we have
\begin{align}
    &\;\quad
    \frac{1}{Tpq^2} \sum_{t=1}^T \cI_{Q,t}
    =
    \frac{1}{Tp} \sum_{w=1}^{Q_f +1} \sum_{t=\theta_{f,w-1} +1}^{\theta_{f,w}} \bQ \bR_{w-1} \bF_t \bV_{w-1}^\trans \left( \frac{\bU^\trans \wt\bU}{q} \right) \left( \frac{\wt\bU^\trans \bU}{q} \right) \bV_{w-1} \bF_t^\trans \bR_{w-1}^\trans \bQ^\trans \notag \\
    &\xrightarrow{P}
    \frac{1}{p} \bQ \left\{ \frac{1}{T} \sum_{w=1}^{Q_f +1} \|\bV_{w-1} \|_F^2 \cdot (\theta_{f,w} - \theta_{f,w-1}) \cdot \bR_{w-1} \bR_{w-1}^\trans \right\} \bQ^\trans
    \to
    \frac{1}{p} \bQ \bSigma_{R} \bQ^\trans ,
    \label{eqn: CUUC_limit}
\end{align}
which gives rise to statement~(D1) using Assumption~\ref{ass: loadings} again.
Then similarly to the proof of \Cref{lemma: consistency_PCA}, by \eqref{eqn: LUUL_decomp} and triangle inequality, we may arrive at
\begin{align*}
    &\;\quad
    \frac{1}{p} \left\| \wh\bQ - \bQ \wh\bH_r \right\|_F^2 \lesssim_P \frac{1}{p} \left\| \frac{1}{Tpq^2} \sum_{t=1}^T (\cII_{Q,t} + \cIII_{Q,t} + \cIV_{Q,t}) \bQ \right\|_F^2 \\
    &\leq
    \frac{1}{T^2p^3q^4} \left\| \sum_{t=1}^T \cII_{Q,t} \bQ \right\|_F^2
    + \frac{1}{T^2p^3q^4} \left\| \sum_{t=1}^T \cIII_{Q,t} \bQ \right\|_F^2
    + \frac{1}{T^2p^3q^4} \left\| \sum_{t=1}^T \cIV_{Q,t} \bQ \right\|_F^2 \\
    &=
    \frac{1}{T^2p^3q^4} \cdot \cO_P\left( Tp^3q^3 + p^2 q^4 + Tpq^2 + p^2 q^3 + T p q^4 + T^2 p q^2 + T p^2 q^2 \right) \\
    &=
    \cO_P\left( \frac{1}{Tq} + \frac{1}{T^2p} + \frac{1}{Tp^2} + \frac{1}{p^2q^2} \right) .
\end{align*}
where the first equality used Lemma~\ref{lemma: consistency_proj_decomp_rate}. Finally, we have
\begin{align*}
    &\;\quad
    \bI_{k_r} = \frac{1}{p} \wh\bQ^\trans \wh\bQ = \frac{1}{p} \wh\bQ^\trans (\wh\bQ - \bQ \wh\bH_r) + \frac{1}{p} \wh\bQ^\trans \bQ \wh\bH_r \\
    &=
    \frac{1}{p} \wh\bQ^\trans (\wh\bQ - \bQ \wh\bH_r) + \frac{1}{p} (\wh\bQ - \bQ \wh\bH_r)^\trans \bQ \wh\bH_r + \frac{1}{p} \wh\bH_r^\trans \bQ^\trans \bQ \wh\bH_r 
    =
    \wh\bH_r \wh\bH_r^\trans + o_P(1),
\end{align*}
where the last equality used Assumption~\ref{ass: loadings}~(ii), Lemma~\ref{lemma: inequality_sandwich}~(i), and the rate of $\|\wh\bQ - \bQ \wh\bH_r\|_F$.
This ends this proof.
\end{proof}


\begin{proof}[Proof of Theorem~\ref{thm: consistency_change_ME}]
It suffices to consider the row main effects, as the proof for the column main effects is analogous. 

\paragraph{Detection consistency.}
Without loss of generality, we set $\bW_{\alpha, \tau} = \bI_p$ throughout this proof.
First, note that for any positive integer $a$, the matrix $\bM_a$ is symmetric and idempotent with rank $a-1$. Thus, $\bM_a$ has eigenvalues 1 repeated $a-1$ times and eigenvalue 0, thus indicating
\begin{equation}
\label{eqn: M_norm}
    \|\bM_a\| = 1.
\end{equation}

Using \eqref{eqn: mu_hat} and the representation $\1_p^\trans \bX_t \1_q = pq\mu_t + q\1_p^\trans \be_{\alpha,t} + p\1_q^\trans \be_{\beta,t} + \1_p^\trans \bE_t \1_q$,
\begin{align*}
    \wh\mu_t -\mu_t = \frac{1}{pq} \1_p^\trans \bX_t \1_q - \mu_t = \frac{1}{pq} \Big( q\1_p^\trans \be_{\alpha,t} + p\1_q^\trans \be_{\beta,t} + \1_p^\trans \bE_t \1_q \Big) .
\end{align*}
Together with \eqref{eqn: mu_alpha_hat}, we have
\begin{align*}
    &\;\quad
    \wh\balpha_t - \balpha_t =
    (\mu_t - \wh\mu_t) \1_p + q^{-1}\bE_t\1_q \\
    &=
    \frac{1}{q} \bE_t\1_q - \frac{1}{p} \1_p \1_p^\trans \be_{\alpha,t} - \frac{1}{q} \1_p \1_q^\trans \be_{\beta,t} - \frac{1}{pq} \1_p \1_p^\trans \bE_t \1_q
    = \frac{1}{q} \bM_p \bE_t\1_q - \1_p \Big(\1_p^\trans \be_{\alpha,t} /p + \1_q^\trans \be_{\beta,t} /q \Big) ,
\end{align*}
and hence $\bM_p(\wh\balpha_t - \balpha_t) = \frac{1}{q} \bM_p \bE_t\1_q$.
Note that Assumption~\ref{ass: identification} implies $\bM_p \bmu_{\alpha,t} = \bmu_{\alpha,t}$ and $\bM_p \bgamma_{\alpha,i} = \bgamma_{\alpha,i}$
for any $t\in[T]$ and $i\in[Q_\alpha]$.
Then by \eqref{eqn: mean_change} and \eqref{eqn: def_hat_alpha_pm}, for any $t\in[bT+1, T-bT]$, we have
\begin{align}
    &\;\quad
    \bM_p \, \wh\bj_{\alpha}(t) = \sum_{u= t +1}^{t+bT} w_{u-t} \bM_p \wh\balpha_u - \sum_{u= t-bT}^{t-1} w_{t-u} \bM_p \wh\balpha_u \notag \\
    &=
    \sum_{u= t +1}^{t+bT} w_{u-t} \bM_p \balpha_u - \sum_{u= t-bT}^{t-1} w_{t-u} \bM_p \balpha_u
    + \sum_{u= t +1}^{t+bT} w_{u-t} \bM_p (\wh\balpha_u - \balpha_u) - \sum_{u= t-bT}^{t-1} w_{t-u} \bM_p (\wh\balpha_u - \balpha_u) \notag \\
    &=
    \left( \sum_{u= t +1}^{t+bT} w_{u-t} \bmu_{\alpha,u} - \sum_{u= t-bT}^{t-1} w_{t-u} \bmu_{\alpha,u} \right) \notag \\
    &\;\quad
    + \left( \sum_{u= t +1}^{t+bT} w_{u-t} \sum_{i=1}^{Q_\alpha} \bgamma_{\alpha,i} \cdot \mathbb{I}\{u \geq \theta_{\alpha,i} \} - \sum_{u= t-bT}^{t-1} w_{t-u} \sum_{i=1}^{Q_\alpha} \bgamma_{\alpha,i} \cdot \mathbb{I}\{u \geq \theta_{\alpha,i} \} \right) \notag \\
    &\;\quad
    + \bM_p \left( \sum_{u= t +1}^{t+bT} w_{u-t} \be_{\alpha,u}  - \sum_{u= t-bT}^{t-1} w_{t-u} \be_{\alpha,u} \right) \notag \\
    &\;\quad
    + \frac{1}{q} \bM_p \left( \sum_{u= t +1}^{t+bT} w_{u-t} \bE_u  - \sum_{u= t-bT}^{t-1} w_{t-u} \bE_u \right) \1_q
    =:
    \cI_{\alpha,t} + \cII_{\alpha,t} + \cIII_{\alpha,t} + \cIV_{\alpha,t} .
    \label{eqn: Mj_hat_alpha_decomp}
\end{align}
Then by the definition in \eqref{eqn: mosum_main_effect} and triangle inequality we have
\begin{align}
    \cT^{(\alpha)}_{b, t} &= \left\| \bM_p\, \wh\bj_{\alpha}(t) \right\|_\infty
    = \Big\| \cI_{\alpha,t} + \cII_{\alpha,t} + \cIII_{\alpha,t} + \cIV_{\alpha,t} \Big\|_\infty \notag \\
    &\leq
    \big\| \cI_{\alpha,t} \big\|_\infty 
    + \big\| \cII_{\alpha,t} \big\|_\infty 
    + \big\| \cIII_{\alpha,t} \big\|_\infty
    + \big\| \cIV_{\alpha,t} \big\|_\infty , 
    \label{eqn: cT_alpha_upper} \\
    \cT^{(\alpha)}_{b, t} &\geq
    \big\| \cII_{\alpha,t} \big\|_\infty 
    - \left( \big\| \cI_{\alpha,t} \big\|_\infty 
    + \big\| \cIII_{\alpha,t} \big\|_\infty
    + \big\| \cIV_{\alpha,t} \big\|_\infty \right) .
    \label{eqn: cT_alpha_lower}
\end{align}

For $\cI_{\alpha,t}$, Lemma~\ref{lemma: smooth_trend} gives
\begin{align}
    \max_{t\in[bT +1, T-bT]} \|\cI_{\alpha,t}\|_\infty = \cO\left( b^2 \right) . 
    \label{eqn: cI_alpha_max}
\end{align}
Consider $\cIV_{\alpha,t}$. From the sub-multiplicativity of the max norm, we have
\begin{align}
    &\;\quad
    \max_{t\in[bT +1, T-bT]} \|\cIV_{\alpha,t}\|_\infty \notag \\
    &\leq
    \|\bM_p\|_\infty \cdot \max_{t\in[bT +1, T-bT]} \left\| \frac{1}{q} \left( \sum_{u= t +1}^{t+bT} w_{u-t} \bE_u  - \sum_{u= t-bT}^{t-1} w_{t-u} \bE_u \right) \1_q \right\|_\infty  \notag \\
    &=
    \cO_P\left( (pT)^{1/\nu} (Tb)^{(1-\nu)/\nu} q^{(1-\nu)/\nu} + \sqrt{\frac{\log(pT)}{qTb}} \right) ,
    \label{eqn: cIV_alpha_max}
\end{align}
where the last equality used Lemma~\ref{lemma: smooth_E_vec1_bound} and $\|\bM_p\|_\infty \leq 2$ from \eqref{eqn: M_norm_infty}.
Similarly, for $\cIII_{\alpha,t}$,
by Assumption~\ref{ass: ME_noise}, Lemma~\ref{lemma: smooth_E_vec1_bound} holds with $q=1$ therein, and thus
\begin{align}
    \max_{t\in[bT +1, T-bT]} \|\cIII_{\alpha,t}\|_\infty = 
    \cO_P\left( (pT)^{1/\nu} (Tb)^{(1-\nu)/\nu} + \sqrt{\frac{\log(pT)}{Tb}} \right) .
    \label{eqn: cIII_alpha_max}
\end{align}

For $\cII_{\alpha,t}$, note that each scan of the MOSUM window contains at most one true change point in $\wh\Theta_\alpha$, since $2bT/T = 2b \to 0$ from Assumption~\ref{ass: kernel} while the change is linearly spaced (from the statement of Theorem~\ref{thm: consistency_change_ME}). 
If $\min_{i\in[Q_f]} | t- \theta_{\alpha,i} | > Tb$, then $\cII_{\alpha,t} = 0$.
Now suppose $\theta_{\alpha,i-1} \leq t-bT < \theta_{\alpha,i} \leq t+bT < \theta_{\alpha,i+1}$ for some $i\in[Q_\alpha]$.
Thus we may write $\theta_{\alpha,i} -bT \leq t < \theta_{\alpha,i} + bT$, for which we have
\begin{align}
    \cII_{\alpha,t} &= 
    \sum_{u= t +1}^{t+bT} w_{u-t} \bgamma_{\alpha,i} \cdot \mathbb{I}\{u \geq \theta_{\alpha,i} \} - \sum_{u= t-bT}^{t-1} w_{t-u} \bgamma_{\alpha,i} \cdot \mathbb{I}\{u \geq \theta_{\alpha,i} \} \notag \\
    &=
    \left( \sum_{u= t + \max(|\theta_{\alpha,i}-t|, |\theta_{\alpha,i}-1 -t|)}^{t+bT} w_{u-t} \right) \bgamma_{\alpha,i}
    = \left( 1- \sum_{u= 1}^{\min(|\theta_{\alpha,i}-t|, |\theta_{\alpha,i}-1 -t|)} w_{u} \right) \bgamma_{\alpha,i} ,
    \label{eqn: cII_alpha_rewrite}
\end{align}
where the last equality used Lemma~\ref{lemma: weight_prop}~(i), and the summation in the final expression is 0 by convention when $\min(|\theta_{\alpha,i}-t|, |\theta_{\alpha,i}-1 -t|) <1$.
Combining \eqref{eqn: cT_alpha_upper}, \eqref{eqn: cI_alpha_max}, \eqref{eqn: cIV_alpha_max}, and \eqref{eqn: cIII_alpha_max} (note that $\cII_{\alpha,t}=0$ in the following case), we have
\begin{align*}
    &\;\quad
    \max_{\substack{t: |t - \theta_{\alpha,i}| > Tb \\ i\in [Q_\alpha] }} \cT^{(\alpha)}_{b, t} 
    \leq 
    \max_{t\in[bT +1, T-bT]} \|\cI_{\alpha,t}\|_\infty
    + \max_{t\in[bT +1, T-bT]} \|\cIII_{\alpha,t}\|_\infty
    + \max_{t\in[bT +1, T-bT]} \|\cIV_{\alpha,t}\|_\infty \\
    &=
    \cO_P\left( b^2
    + (pT)^{1/\nu} (Tb)^{(1-\nu)/\nu} q^{(1-\nu)/\nu} + \sqrt{\frac{\log(pT)}{qTb}}
    + (pT)^{1/\nu} (Tb)^{(1-\nu)/\nu} + \sqrt{\frac{\log(pT)}{Tb}}
    \right) \\
    &=
    \cO_P\left( b^2
    + (pT)^{1/\nu} (Tb)^{(1-\nu)/\nu} + \sqrt{\frac{\log(pT)}{Tb}}
    \right) ,
\end{align*}
which, together with $c_{T,p,q}^{(\alpha)} \to \infty$, implies that as $\min\{T, p, q\} \to \infty$:
\begin{align}
    \P\left( \max_{\substack{t: |t - \theta_{\alpha,i}| > Tb \\ i\in [Q_\alpha] }} \cT^{(\alpha)}_{b, t} > c_{T,p,q}^{(\alpha)} 
    \left[ b^2 + (pT)^{1/\nu} (Tb)^{(1-\nu)/\nu} + \sqrt{\frac{\log(pT)}{Tb}} \right]
    \right) 
    \to 0 .
    \label{eqn: outside_gamma_below_thres_ME}
\end{align}
Hence by the definition of the threshold in the statement of Theorem~\ref{thm: consistency_change_ME}, we conclude that that no estimator in $\wh\Theta_\alpha$ would be detected outside of the radius of $Tb$ from each change point in $\Theta_\alpha$, with probability approaching one.

For any $\xi\in(0,1]$, $i\in[Q_\alpha]$, define the set
\[
\cS_{T,p,q,i}^{(\alpha)} := \left\{ \cT^{(\alpha)}_{b, \theta_{\alpha,i}} \geq \max\left( \max_{t: |t - \theta_{\alpha,i}| > \xi Tb} \cT^{(\alpha)}_{b, t}, c_{T,p,q}^{(\alpha)} 
\left[ b^2 + (pT)^{1/\nu} (Tb)^{(1-\nu)/\nu} + \sqrt{\frac{\log(pT)}{Tb}} \right]
\right) \right\} ,
\]
and the set $\cS_{T,p,q}^{(\alpha)} := \bigcap_{i\in[Q_\alpha]} \cS_{T,p,q,i}^{(\alpha)}$.
Note that by \eqref{eqn: cII_alpha_rewrite}, for any $i\in[Q_\alpha]$, we have $\cII_{\alpha,\theta_{\alpha,i}} = \bgamma_{\alpha,i}$.
Also, uniformly over $i\in[Q_\alpha]$, for any $t$ such that $|t - \theta_{\alpha,i}| > \xi Tb$, \eqref{eqn: cII_alpha_rewrite} implies
\begin{align*}
    \big\| \cII_{\alpha,t} \big\|_\infty
    &\leq
    \left| 1- \sum_{u= 1}^{\min(|\theta_{\alpha,i}-t|, |\theta_{\alpha,i}-1 -t|)} w_{u} \right| \cdot \big\| \bgamma_{\alpha,i} \big\|_\infty  \\
    &\leq
    \left| 1- c' \frac{\xi Tb}{Tb} \right| \cdot \big\| \bgamma_{\alpha,i} \big\|_\infty 
    = (1- c' \xi) \big\| \bgamma_{\alpha,i} \big\|_\infty ,
\end{align*}
where the second inequality used the last result of Lemma~1 in \cite{Chenetal2022}, $c'$ is a strictly positive constant that only depends on the kernel $K$, and the last equality used the scenario $\xi \leq 1/c'$ and hence $c' \xi \leq 1$ (as $\xi\in(0,1]$, either $1/c' \geq 1$ which does not affect our argument, or we have $1/c' < 1$ which implies a better result anyway).
Hence combining the above with \eqref{eqn: cT_alpha_upper} and \eqref{eqn: cT_alpha_lower},
uniformly over $i\in[Q_\alpha]$:
\begin{align}
    & \cT^{(\alpha)}_{b, \theta_{\alpha,i}} \geq
    \big\| \bgamma_{\alpha,i} \big\|_\infty + \cO_P\left( 
    \left[ b^2 + (pT)^{1/\nu} (Tb)^{(1-\nu)/\nu} + \sqrt{\frac{\log(pT)}{Tb}} \right]
    \right) =
    \big\| \bgamma_{\alpha,i} \big\|_\infty (1 + o_P(1)) ; 
    \label{eqn: T_alpha_theta_tb_lower_bound} \\
    & \max_{\substack{t: |t - \theta_{\alpha,j}| > \xi Tb \\ j\in [Q_\alpha]}} \cT^{(\alpha)}_{b, t}
    \leq
    (1- c' \xi) \big\| \bgamma_{\alpha,i} \big\|_\infty
    + \cO_P\left( 
    \left[ b^2 + (pT)^{1/\nu} (Tb)^{(1-\nu)/\nu} + \sqrt{\frac{\log(pT)}{Tb}} \right]
    \right) \notag \\
    &\quad\quad\quad\quad\quad\quad\quad\quad
    < \big\| \bgamma_{\alpha,i} \big\|_\infty (1 + o_P(1)) ;
    \label{eqn: T_alpha_tb_upper_bound} \\
    & c_{T,p,q}^{(\alpha)}
    \left[ b^2 + (pT)^{1/\nu} (Tb)^{(1-\nu)/\nu} + \sqrt{\frac{\log(pT)}{Tb}} \right]
    \leq
    \frac{1}{2} \big\| \bgamma_{\alpha,i} \big\|_\infty
    < \big\| \bgamma_{\alpha,i} \big\|_\infty ; \notag
\end{align}
by the condition~\eqref{cond:thm: consistency_change_ME}.
Hence $\P(\cS_{T,p,q}^{(\alpha)}) \to 1$ as $\min\{T, p, q\} \to \infty$, indicating that at least some change would be detected by our detection rule, within the radius of $\xi Tb$ from each change point in $\Theta_\alpha$.
Furthermore, we can show the test statistic valued on farther timestamps is smaller than on those closer to each true change point (within the radius of $Tb$). By similar arguments for $\cS_{T,p,q}^{(\alpha)}$, we can show $\P(\wt\cS_{T,p,q}^{(\alpha)}) \to 1$ as $\min\{T, p, q\} \to \infty$, where 
\begin{align*}
    \wt\cS_{T,p,q}^{(\alpha)} &:= \bigcap_{i\in [Q_\alpha]} \wt\cS_{T,p,q,i}^{(\alpha)}, \\
    \wt\cS_{T,p,q,i}^{(\alpha)} &:= \bigcap_{0\leq j \leq \lfloor 2/\xi \rfloor -2} \Bigg[ \left\{ \cT^{(\alpha)}_{b, \theta_{\alpha,i} + j \xi Tb /2} > \max_{\theta_{\alpha,i} + (j+1) \xi Tb /2 \leq t \leq \theta_{\alpha,i} + Tb} \cT^{(\alpha)}_{b, t} \right\} \\
    &\;\quad
    \cap \left\{ \cT^{(\alpha)}_{b, \theta_{\alpha,i} - j \xi Tb /2} > \max_{\theta_{\alpha,i} - Tb \leq t \leq \theta_{\alpha,i} - (j+1) \xi Tb /2} \cT^{(\alpha)}_{b, t} \right\} \Bigg] .
\end{align*}
Altogether, we have now shown
\begin{equation}
\label{eqn: loc_row_ME_consistency_coarse}
\P\left( \wh{Q}_{\alpha} = Q_\alpha; \max_{i\in[\wh{Q}_\alpha]} |\wh\theta_{\alpha,i} - \theta_{\alpha,i}| \leq \xi Tb \right) \to 1 ,
\end{equation}
for any $\xi \in (0,1]$, as $\min\{T, p, q\} \to\infty$. We refine such result in the remainders of the proof. 

\paragraph{Rate of estimation.}
Fix any $i\in [Q_\alpha]$. On the set $\cS_{T,p,q}^{(\alpha)} \cap \wt\cS_{T,p,q}^{(\alpha)}$, \eqref{eqn: loc_row_ME_consistency_coarse} indicates $|\wh\theta_{\alpha,i} - \theta_{\alpha,i}| \leq Tb$. Without loss of generality, we only consider the case $\wh\theta_{\alpha,i} \leq \theta_{\alpha,i}$ throughout the rest of the proof; the proof for $\wh\theta_{\alpha,i} \geq \theta_{\alpha,i}$ is similar and omitted.
From \eqref{eqn: loc_row_ME_consistency_coarse} and our detection scheme, the $i$th detected change point for the row main effects can be represented as
\[
\wh\theta_{\alpha,i} 
= \argmax_{t: \theta_{\alpha,i} -Tb < t \leq \theta_{\alpha,i} } \cT^{(\alpha)}_{b, t}
= \argmax_{t: \theta_{\alpha,i} -Tb < t \leq \theta_{\alpha,i} } \big\| \bM_p \, \wh\bj_{\alpha}(t) \big\|_\infty.
\]
Define the (row) index $j_{\theta_{\alpha,i}} := \argmax_{j\in[p]} \big| \big( \bM_p \, \wh\bj_{\alpha}(\wh\theta_{\alpha,i}) \big)_j \big|$. We next show that with probability approaching one, we have
\begin{align}
    \min_{i\in[Q_\alpha]} 
    \frac{\big| (\bgamma_{\alpha,i} )_{j_{\theta_{\alpha,i}}} \big|}{ \big\| \bgamma_{\alpha,i} \big\|_\infty } 
    \geq
    \frac{1}{2} .
    \label{eqn: detect_alpha_refine_step1}
\end{align}
If the above failed to hold, i.e.,\ for some $i\in[Q_\alpha]$, $ \big| (\bgamma_{\alpha,i} )_{j_{\theta_{\alpha,i}}} \big| <  \big\| \bgamma_{\alpha,i} \big\|_\infty /2$,
then 
\begin{align*}
    \big\| \bM_p \, \wh\bj_{\alpha}(\wh\theta_{\alpha,i}) \big\|_\infty
    \leq
    \big| (\bgamma_{\alpha,i} )_{j_{\theta_{\alpha,i}}} \big| \cdot (1 + o_P(1))
    <
    \frac{1}{2} \big\| \bgamma_{\alpha,i} \big\|_\infty \cdot (1 + o_P(1))
    \leq
    \frac{1}{2} \cT^{(\alpha)}_{b, \theta_{\alpha,i}} ,
\end{align*}
where the first inequality used similar arguments as in \eqref{eqn: T_alpha_tb_upper_bound}, and the last used \eqref{eqn: T_alpha_theta_tb_lower_bound}. This implies $\cT^{(\alpha)}_{b, \theta_{\alpha,i}} > \cT^{(\alpha)}_{b, \wh\theta_{\alpha,i}}$ which is a contradiction to the definition of $\wh\theta_{\alpha,i}$. Hence \eqref{eqn: detect_alpha_refine_step1} is true.

Recalling the notations in~\eqref{eqn: Mj_hat_alpha_decomp}, we next show that
\begin{align}
    \max_{i\in[Q_\alpha]} \max_{t\in [\theta_{\alpha,i} - Tb +1, \theta_{\alpha,i} -\xi_0]}
    \frac{\big\| \cIV_{\alpha,\theta_{\alpha,i}} - \cIV_{\alpha,t} \big\|_\infty}{ |\theta_{\alpha,i}-t|^{1/2} } 
    &=
    \cO_P\left( p^{1/\nu} (Tb)^{(1-\nu)/\nu} q^{(1-\nu)/\nu}
    + \sqrt{\frac{\log(pTb)}{q T^2b^2}} \right) .
    \label{eqn: detect_alpha_refine_step2-IV}
\end{align}
Fix any $i\in[Q_\alpha]$ and $t$ such that $\theta_{\alpha,i} - Tb \leq t \leq \theta_{\alpha,i}$.
By direct calculations (cf.\ proof of Theorem~2 in \cite{Chenetal2022}), we can decompose
\begin{align*}
    &\;\quad
    \cIV_{\alpha,t} - \cIV_{\alpha,\theta_{\alpha,i}} \\
    &=
    \frac{\bM_p}{q} \Bigg\{ \Bigg( \sum_{u= t +1}^{t+bT} w_{u-t} \bE_u  - \sum_{u= t-bT}^{t-1} w_{t-u} \bE_u \Bigg) 
    - \Bigg( \sum_{u= \theta_{\alpha,i} +1}^{\theta_{\alpha,i}+bT} w_{u-\theta_{\alpha,i}} \bE_u  - \sum_{u= \theta_{\alpha,i}-bT}^{\theta_{\alpha,i}-1} w_{\theta_{\alpha,i}-u} \bE_u \Bigg) \Bigg\} \1_q \\
    &=
    \frac{\bM_p}{q} \Bigg\{
    \sum_{u=t-bT}^{\theta_{\alpha,i} -bT -1} w_{t-u} \bE_u
    + \sum_{u=\theta_{\alpha,i} -bT}^{t -1} (w_{t-u} - w_{\theta_{\alpha,i}-1}) \bE_u
    - \sum_{u=t+1}^{\theta_{\alpha,i} -1} (w_{u-t} - w_{\theta_{\alpha,i}-u}) \bE_u \\
    &\;\quad
    - \sum_{u=\theta_{\alpha,i} +1 }^{t +bT} (w_{u-t} - w_{u- \theta_{\alpha,i}}) \bE_u
    - \sum_{u=t+bT+1}^{\theta_{\alpha,i} +bT} w_{u -\theta_{\alpha,i} } \bE_u
    + w_{\theta_{\alpha,i} -t} \bE_t
    - w_{\theta_{\alpha,i} -t} \bE_{\theta_{\alpha,i}}
    \Bigg\} \1_q \\
    &=:
    \frac{\bM_p}{q} \Bigg\{ 
    \cIV_{i,t,1} + \cIV_{i,t,2} + \cIV_{i,t,3} + \cIV_{i,t,4} + \cIV_{i,t,5} + \cIV_{i,t,6} + \cIV_{i,t,7}
    \Bigg\} \1_q .
\end{align*}
By sub-multiplicativity of the max norm and \eqref{eqn: M_norm_infty}, we have
\begin{align*}
    &\;\quad
    \max_{i\in[Q_\alpha]} \max_{t\in [\theta_{\alpha,i} -Tb +1, \theta_{\alpha,i}]}
    \left\| \frac{\bM_p}{q} \frac{\cIV_{i,t,1}}{|\theta_{\alpha,i}-t|^{1/2}} \1_q \right\|_\infty
    \leq
    2 \max_{i\in[Q_\alpha]} \max_{t\in [\theta_{\alpha,i} -Tb +1, \theta_{\alpha,i}]} \left\| \frac{1}{q} \frac{\cIV_{i,t,1}}{|\theta_{\alpha,i}-t|^{1/2}} \1_q \right\|_\infty \\
    &=
    \cO_P\left( p^{1/\nu} (Tb)^{(1-\nu)/\nu} q^{(1-\nu)/\nu}
    + \sqrt{\frac{\log(pTb)}{q T^2b^2}} \right) ,
\end{align*}
where the equality used Lemma~\ref{lemma: smooth_E_vec1_bound_local}.
By the similar arguments, we also have
\begin{align*}
    &\;\quad
    \max_{i\in[Q_\alpha]} \max_{t\in [\theta_{\alpha,i} -Tb +1, \theta_{\alpha,i}]}
    \left\| \frac{\bM_p}{q} \frac{\cIV_{i,t,3}}{|\theta_{\alpha,i}-t|^{1/2}} \1_q \right\|_\infty ,
    \quad
    \max_{i\in[Q_\alpha]} \max_{t\in [\theta_{\alpha,i} -Tb +1, \theta_{\alpha,i}]}
    \left\| \frac{\bM_p}{q} \frac{\cIV_{i,t,5}}{|\theta_{\alpha,i}-t|^{1/2}} \1_q \right\|_\infty \\
    &=
    \cO_P\left( p^{1/\nu} (Tb)^{(1-\nu)/\nu} q^{(1-\nu)/\nu}
    + \sqrt{\frac{\log(pTb)}{q T^2b^2}} \right) .
\end{align*}

On the other hand, applying Lemma~\ref{lemma: smooth_E_entry_bound_local_general} with $m= 1$ and $\xi=1$ therein, we have
\begin{align*}
    &\;\quad
    \max_{i\in[Q_\alpha]} \max_{t\in [\theta_{\alpha,i} -Tb +1, \theta_{\alpha,i} -1]}
    \left\| \frac{\bM_p}{q} \frac{\cIV_{i,t,6}}{|\theta_{\alpha,i}-t|^{1/2}} \1_q \right\|_\infty ,
    \quad
    \max_{i\in[Q_\alpha]} \max_{t\in [\theta_{\alpha,i} -Tb +1, \theta_{\alpha,i} -1]}
    \left\| \frac{\bM_p}{q} \frac{\cIV_{i,t,7}}{|\theta_{\alpha,i}-t|^{1/2}} \1_q \right\|_\infty \\
    &=
    \cO_P\left( p^{1/\nu} (Tb)^{(1-\nu)/\nu} q^{(1-\nu)/\nu}
    + \sqrt{\frac{\log(pTb)}{q T^2b^2}} \right) .
\end{align*}
Finally, by Lemma~\ref{lemma: smooth_E_entry_bound_local_ww}, we have
\begin{align*}
    \max_{i\in[Q_\alpha]} \max_{t\in [\theta_{\alpha,i} -Tb +1, \theta_{\alpha,i} -1]}
    \left\| \frac{\bM_p}{q} \frac{\cIV_{i,t,4}}{|\theta_{\alpha,i}-t|^{1/2}} \1_q \right\|_\infty
    &=
    \cO_P\left( p^{1/\nu} (Tb)^{2(1-\nu)/\nu} q^{(1-\nu)/\nu} 
    + \sqrt{\frac{\log(pTb)}{q T^3b^3}} \right) ,
\end{align*}
and the same result holds similarly with $\cIV_{i,t,4}$ replaced by $\cIV_{i,t,2}$ above.
This concludes \eqref{eqn: detect_alpha_refine_step2-IV}.

Note that by Assumption~\ref{ass: ME_noise}, $\cIII_{\alpha,t}$ has the same structure as $\cIV_{\alpha,t}$ with $q=1$. Then \eqref{eqn: detect_alpha_refine_step2-IV} directly gives
\begin{align}
    \max_{i\in[Q_\alpha]} \max_{t\in [\theta_{\alpha,i} - Tb +1, \theta_{\alpha,i} -1]}
    \frac{\big\| \cIII_{\alpha,\theta_{\alpha,i}} - \cIII_{\alpha,t} \big\|_\infty}{ |\theta_{\alpha,i}-t|^{1/2} }  
    &=
    \cO_P\left( p^{1/\nu} (Tb)^{(1-\nu)/\nu}
    + \sqrt{\frac{\log(pTb)}{T^2b^2}} \right) .
    \label{eqn: detect_alpha_refine_step2-III}
\end{align}

Next, assume $(\bgamma_{\alpha,i} )_{j_{\theta_{\alpha,i}}} > 0$ without loss of generality (all the following arguments hold analogously otherwise).
With $C_w$ from Lemma~\ref{lemma: weight_prop}~(ii), set $\xi = 1/(2C_w)$ in \eqref{eqn: loc_row_ME_consistency_coarse} (note that we are guaranteed that $\xi<1$ since $C_w \geq 1$), by which we have
\[
\sum_{u= 1}^{\min(|\theta_{\alpha,i}-\wh\theta_{\alpha,i}|, |\theta_{\alpha,i}-1 -\wh\theta_{\alpha,i}|)} w_{u}
\leq
\sum_{u= 1}^{Tb/(2C_w)} |w_{u}|
\leq \frac{1}{2} ,
\]
with probability approaching one.
Thus using \eqref{eqn: Mj_hat_alpha_decomp}, \eqref{eqn: cII_alpha_rewrite}, and $(\bgamma_{\alpha,i} )_{j_{\theta_{\alpha,i}}} > 0$, we have
\begin{align*}
    &\;\quad
    \big( \bM_p \, \wh\bj_{\alpha}(\wh\theta_{\alpha,i}) \big)_{j_{\theta_{\alpha,i}}} \\
    &=
    \Bigg( 1- \sum_{u= 1}^{\min(|\theta_{\alpha,i}-\wh\theta_{\alpha,i}|, |\theta_{\alpha,i}-1 -\wh\theta_{\alpha,i}|)} w_{u} \Bigg) \big( \bgamma_{\alpha,i} \big)_{j_{\theta_{\alpha,i}}}
    + \big( \cI_{\alpha,\wh\theta_{\alpha,i}} \big)_{j_{\theta_{\alpha,i}}}
    + \big( \cIII_{\alpha,\wh\theta_{\alpha,i}} \big)_{j_{\theta_{\alpha,i}}}
    + \big( \cIV_{\alpha,\wh\theta_{\alpha,i}} \big)_{j_{\theta_{\alpha,i}}} \\
    &\geq
    \frac{1}{2} \big( \bgamma_{\alpha,i} \big)_{j_{\theta_{\alpha,i}}}
    - \left( 
    \big\| \cI_{\alpha,\wh\theta_{\alpha,i}} \big\|_\infty
    + \big\| \cIII_{\alpha,\wh\theta_{\alpha,i}} \big\|_\infty
    + \big\| \cIV_{\alpha,\wh\theta_{\alpha,i}} \big\|_\infty
    \right) ,
\end{align*}
with probability approaching one, which thus indicates that $\big( \bM_p \, \wh\bj_{\alpha}(\wh\theta_{\alpha,i}) \big)_{j_{\theta_{\alpha,i}}} > 0$ by \eqref{eqn: cI_alpha_max}, \eqref{eqn: cIV_alpha_max}, \eqref{eqn: cIII_alpha_max},
and \eqref{eqn: detect_alpha_refine_step1}.
Next, by the definition of $\cI_{\alpha,t}$ in \eqref{eqn: Mj_hat_alpha_decomp}, we have
\begin{align}
    &\;\quad
    \big\| \cI_{\alpha,\theta_{\alpha,i}} - \cI_{\alpha,\wh\theta_{\alpha,i}} \big\|_\infty \notag \\
    &\leq
    \Bigg\| \sum_{u= \theta_{\alpha,i} +1}^{\theta_{\alpha,i}+bT} w_{u-\theta_{\alpha,i}} \bmu_{\alpha,u} - \sum_{u= \wh\theta_{\alpha,i} +1}^{\wh\theta_{\alpha,i}+bT} w_{u-\wh\theta_{\alpha,i}} \bmu_{\alpha,u} \Bigg\|_\infty 
    \notag \\
    &\;\quad
    + \Bigg\| \sum_{u= \wh\theta_{\alpha,i}-bT}^{\wh\theta_{\alpha,i}-1} w_{\wh\theta_{\alpha,i}-u} \bmu_{\alpha,u} - \sum_{u= \theta_{\alpha,i}-bT}^{\theta_{\alpha,i}-1} w_{\theta_{\alpha,i}-u} \bmu_{\alpha,u} \Bigg\|_\infty 
    \notag \\
    &=
    \Bigg\| w_1 \big( \bmu_{\alpha,\theta_{\alpha,i} +1} - \bmu_{\alpha,\wh\theta_{\alpha,i} +1} \big) + \dots + 
    w_{bT} \big( \bmu_{\alpha,\theta_{\alpha,i} +bT} - \bmu_{\alpha,\wh\theta_{\alpha,i} +bT} \big)
    \Bigg\|_\infty 
    \notag \\
    &\;\quad
    + \Bigg\| \sum_{u= \wh\theta_{\alpha,i}-bT}^{\wh\theta_{\alpha,i}-1} w_{\wh\theta_{\alpha,i}-u} \bmu_{\alpha,u} - \sum_{u= \theta_{\alpha,i}-bT}^{\theta_{\alpha,i}-1} w_{\theta_{\alpha,i}-u} \bmu_{\alpha,u} \Bigg\|_\infty 
    \notag \\
    &=
    \cO\big( |\wh\theta_{\alpha,i} - \theta_{\alpha,i}|/T \big) 
    + \Bigg\| \sum_{u= \wh\theta_{\alpha,i}-bT}^{\wh\theta_{\alpha,i}-1} w_{\wh\theta_{\alpha,i}-u} \bmu_{\alpha,u} - \sum_{u= \theta_{\alpha,i}-bT}^{\theta_{\alpha,i}-1} w_{\theta_{\alpha,i}-u} \bmu_{\alpha,u} \Bigg\|_\infty 
    \notag \\
    &=
    \cO\big( |\wh\theta_{\alpha,i} - \theta_{\alpha,i}|/T \big) ,
    \label{eqn: cI_alpha_theta_diff_rate}
\end{align}
where the second last equality used the definition $\bmu_{\alpha,t}= \bmu_{\alpha}(t/T)$, Assumption~\ref{ass: trend_fun} and hence the Taylor's expansion of $\bmu_{\alpha}((\wh\theta_{\alpha,i}+1) /T)$ on $(\theta_{\alpha,i}+1) /T$, etc.,
combined with \eqref{eqn: loc_row_ME_consistency_coarse} and Lemma~\ref{lemma: weight_prop}~(i);
and the last used similar arguments.

Therefore, with probability approaching one, by \eqref{eqn: Mj_hat_alpha_decomp} and \eqref{eqn: cII_alpha_rewrite},
uniformly over $\theta_{\alpha,i} \in \Theta_\alpha$ and any $\wh\theta_{\alpha,i}$ within the $Tb$ radius around each $\theta_{\alpha,i}$, we have 
\begin{align}
    0 &\geq 
    \cT^{(\alpha)}_{b, \theta_{\alpha,i}} - \cT^{(\alpha)}_{b, \wh\theta_{\alpha,i}} \notag \\
    &=
    \big\| \bM_p \, \wh\bj_{\alpha}(\theta_{\alpha,i}) \big\|_\infty - \big\| \bM_p \, \wh\bj_{\alpha}(\wh\theta_{\alpha,i}) \big\|_\infty 
    \geq
    \big( \bM_p \, \wh\bj_{\alpha}(\theta_{\alpha,i}) \big)_{j_{\theta_{\alpha,i}}} - \big( \bM_p \, \wh\bj_{\alpha}(\wh\theta_{\alpha,i}) \big)_{j_{\theta_{\alpha,i}}} \notag \\
    &\geq
    \left( \sum_{u= 1}^{|\theta_{\alpha,i}-\wh\theta_{\alpha,i}|} w_{u} \right) \big( \bgamma_{\alpha,i} \big)_{j_{\theta_{\alpha,i}}}
    - \cO\Bigg( \frac{|\wh\theta_{\alpha,i} - \theta_{\alpha,i}|}{T} \Bigg) \notag \\
    &\;\quad
    - \big\| (\cIII_{\alpha,\theta_{\alpha,i}} - \cIII_{\alpha,\wh\theta_{\alpha,i}}) + (\cIV_{\alpha,\theta_{\alpha,i}} - \cIV_{\alpha,\wh\theta_{\alpha,i}}) \big\|_\infty \notag \\
    &\geq
    \frac{c'}{2} \cdot \frac{|\wh\theta_{\alpha,i} - \theta_{\alpha,i}|}{Tb}
    \cdot \big\| \bgamma_{\alpha,i} \big\|_\infty
    - \cO\Bigg( \frac{|\wh\theta_{\alpha,i} - \theta_{\alpha,i}|}{T} \Bigg) \notag \\
    &\;\quad
    - \big\| (\cIII_{\alpha,\theta_{\alpha,i}} - \cIII_{\alpha,\wh\theta_{\alpha,i}}) + (\cIV_{\alpha,\theta_{\alpha,i}} - \cIV_{\alpha,\wh\theta_{\alpha,i}}) \big\|_\infty ,
    \label{eqn: detect_alpha_refine_step3}
\end{align}
where the first inequality used the definition of $\wh\theta_{\alpha,i}$;
the second inequality used the fact that $\big\| \bM_p \, \wh\bj_{\alpha}(\theta_{\alpha,i}) \big\|_\infty \geq \big( \bM_p \, \wh\bj_{\alpha}(\theta_{\alpha,i}) \big)_{j_{\theta_{\alpha,i}}}$, 
and $\big\| \bM_p \, \wh\bj_{\alpha}(\wh\theta_{\alpha,i}) \big\|_\infty = \big( \bM_p \, \wh\bj_{\alpha}(\wh\theta_{\alpha,i}) \big)_{j_{\theta_{\alpha,i}}}$ by the definition of $j_{\theta_{\alpha,i}}$ and $\big( \bM_p \, \wh\bj_{\alpha}(\wh\theta_{\alpha,i}) \big)_{j_{\theta_{\alpha,i}}} > 0$ with probability approaching one;
the third used \eqref{eqn: cI_alpha_theta_diff_rate};
and the last used \eqref{eqn: detect_alpha_refine_step1} and the last result of Lemma~1 in \cite{Chenetal2022}, with $c'$ being a strictly positive constant that only depends on the kernel $K(\cdot)$.

Rearranging \eqref{eqn: detect_alpha_refine_step3},
combining \eqref{eqn: detect_alpha_refine_step2-IV}, \eqref{eqn: detect_alpha_refine_step2-III},
and the condition $b=o(\min_{i\in[Q_\alpha]} \big\| \bgamma_{\alpha,i} \big\|_\infty)$ from~\eqref{cond:thm: consistency_change_ME},
we finally conclude that uniformly over $i\in[Q_\alpha]$, we have
\begin{align*}
    \big\| \bgamma_{\alpha,i} \big\|_\infty |\wh\theta_{\alpha,i} - \theta_{\alpha,i}|^{1/2}
    &\lesssim_P
    Tb \cdot \big\| (\cIII_{\alpha,\theta_{\alpha,i}} - \cIII_{\alpha,\wh\theta_{\alpha,i}}) + (\cIV_{\alpha,\theta_{\alpha,i}} - \cIV_{\alpha,\wh\theta_{\alpha,i}}) \big\|_\infty  \\
    &\lesssim_P
    p^{1/\nu} (Tb)^{1/\nu}
    + \sqrt{\log(pTb)}  ,
\end{align*}
which further gives
\begin{align*}
    |\wh\theta_{\alpha,i} - \theta_{\alpha,i}| \lesssim_P
    \frac{(pTb)^{2/\nu}  + \log(pTb)}{\big\| \bgamma_{\alpha,i} \big\|_\infty^2} .
\end{align*}
This completes the proof of the theorem.
\end{proof}


\begin{proof}[Proof of Theorem~\ref{thm: consistency_change_fac}]

As the eigenvalues of $\bW_W$ are all bounded away from zero and infinity, we set $\bW_W = \bI_{k_rk_c(k_rk_c+1)/2}$ throughout this proof without loss of generality.
Then with the construction in \eqref{eqn: test_stat_aggreg}, consider
\begin{align*}
    \frac{1}{\sqrt{2\gamma}} \Vech\big( \wh{\bGamma}_{G ,\gamma ,\tau} \big) 
    = \frac{1}{\sqrt{2\gamma}} \left\{ \sum_{t=\tau+1}^{\tau+\gamma} \Vech\left( \vec(\wh\bG_t) \vec(\wh\bG_t)^\trans \right)
    - \sum_{t=\tau-\gamma+1}^{\tau} \Vech\left( \vec(\wh\bG_t) \vec(\wh\bG_t)^\trans \right) \right\} .
\end{align*}

\paragraph{Detection consistency.}
Recall from \eqref{eqn: Gt_hat} that $\wh\bG_t = \wh\bQ^\trans \wh\bL_t \wh\bU /(pq)$, similar to \eqref{eqn: LUUL_decomp} and the argument for $\cI_{Q,t}$ therein, together with \eqref{eqn: Gt_define}, we have
\begin{align}
    &\;\quad
    \Vech\left\{ \vec(\wh\bG_t) \vec(\wh\bG_t)^\trans \right\}
    = \L_{k_rk_c} \vec\left\{ \vec(\wh\bG_t) \vec(\wh\bG_t)^\trans \right\}  \notag \\
    &=
    (qp)^{-2} \L_{k_rk_c} \vec\left\{ (\wh\bU \otimes \wh\bQ)^\trans \vec(\wh\bL_t) \vec(\wh\bL_t)^\trans  (\wh\bU \otimes \wh\bQ) \right\} \notag
    \\
    &=
    (qp)^{-2} \L_{k_rk_c} \left\{ (\wh\bU \otimes \wh\bQ) \otimes (\wh\bU \otimes \wh\bQ) \right\}^\trans \vec\left\{ \vec(\wh\bL_t) \vec(\wh\bL_t)^\trans \right\} \notag \\
    &=
    (qp)^{-2} \L_{k_rk_c} \left\{ (\wh\bU \otimes \wh\bQ) \otimes (\wh\bU \otimes \wh\bQ) \right\}^\trans \vec\left\{ \vec[\bM_p (\bC_t + \bE_t) \bM_q] \vec[\bM_p (\bC_t + \bE_t) \bM_q]^\trans \right\} \notag \\
    &=
    (qp)^{-2} \L_{k_rk_c} \left\{ (\wh\bU \otimes \wh\bQ) \otimes (\wh\bU \otimes \wh\bQ) \right\}^\trans \notag \\
    &\;\quad
    \cdot \left\{ (\bM_p \otimes \bM_q) \otimes (\bM_p \otimes \bM_q) \right\} \vec\left\{ \vec(\bC_t + \bE_t) \vec(\bC_t + \bE_t)^\trans \right\} \notag \\
    &=
    (qp)^{-2} \L_{k_rk_c} \left\{ (\wh\bU \otimes \wh\bQ) \otimes (\wh\bU \otimes \wh\bQ) \right\}^\trans \left\{ (\bM_p \otimes \bM_q) \otimes (\bM_p \otimes \bM_q) \right\} \notag \\
    &\;\quad
    \cdot \vec\Big\{ \vec(\bC_t + \bE_t) \vec(\bC_t + \bE_t)^\trans - \E[\vec(\bC_t + \bE_t) \vec(\bC_t + \bE_t)^\trans] \Big\} \notag \\
    &\;\quad
    + (qp)^{-2} \L_{k_rk_c} \left\{ (\wh\bU \otimes \wh\bQ) \otimes (\wh\bU \otimes \wh\bQ) \right\}^\trans \vec\Big\{ \E[ \vec(\bC_t) \vec(\bC_t)^\trans ] \Big\} \notag \\
    &\;\quad
    + (qp)^{-2} \L_{k_rk_c} \left\{ (\wh\bU \otimes \wh\bQ) \otimes (\wh\bU \otimes \wh\bQ) \right\}^\trans \left\{ (\bM_p \otimes \bM_q) \otimes (\bM_p \otimes \bM_q) \right\} \vec\Big\{ \E[\vec(\bE_t) \vec(\bE_t)^\trans] \Big\} \notag \\
    &=:
    \cI_{G,t} + \cII_{G,t} + \cIII_{G,t} .
    \label{eqn: vech_GG_decomp} 
\end{align}

Note that $\cIII_{G,t}$ is invariant over $t\in[T]$, and hence
\begin{align*}
    & \sum_{t=\tau+1}^{\tau+\gamma} \cIII_{G,t} - \sum_{t=\tau-\gamma+1}^{\tau} \cIII_{G,t} = \0 .
\end{align*}
Together with \eqref{eqn: test_stat_aggreg} and \eqref{eqn: vech_GG_decomp}, we may write
\begin{align}
    &\;\quad
    \frac{1}{\sqrt{2\gamma}} \left\| \sum_{t=\tau+1}^{\tau+\gamma} \cII_{G,t}
    - \sum_{t=\tau-\gamma+1}^{\tau} \cII_{G,t} \right\|
    - \frac{1}{\sqrt{2\gamma}} \left\| \sum_{t=\tau+1}^{\tau+\gamma} \cI_{G,t}
    - \sum_{t=\tau-\gamma+1}^{\tau} \cI_{G,t} \right\| \notag \\
    &\leq 
    \cT^{(f)}_{\gamma, \tau} \leq
    \frac{1}{\sqrt{2\gamma}} \left\| \sum_{t=\tau+1}^{\tau+\gamma} \cII_{G,t}
    - \sum_{t=\tau-\gamma+1}^{\tau} \cII_{G,t} \right\|
    + \frac{1}{\sqrt{2\gamma}} \left\| \sum_{t=\tau+1}^{\tau+\gamma} \cI_{G,t}
    - \sum_{t=\tau-\gamma+1}^{\tau} \cI_{G,t} \right\| .
    \label{eqn: T_tau_gamma_bound}
\end{align}

First, by Lemma~\ref{lemma: rate_vech_GG}, we have
\begin{align}
    &
    \max_{\gamma \leq \tau\leq T-\gamma} \frac{1}{\sqrt{2\gamma}} \left\| \sum_{t=\tau+1}^{\tau+\gamma} \cI_{G,t}
    - \sum_{t=\tau-\gamma+1}^{\tau} \cI_{G,t} \right\| 
    = \cO_P\left( \sqrt{\log(T)} \right) .
    \label{eqn: rate_cI_Gt}
\end{align}
In addition, note that by \eqref{eqn: Gt_define}, we may further decompose
\begin{align*}
    &\;\quad
    \cII_{G,t} = (qp)^{-2} \L_{k_rk_c} \left\{ (\wh\bU^\trans \bU \otimes \wh\bQ^\trans \bQ) \otimes (\wh\bU^\trans \bU \otimes \wh\bQ^\trans \bQ) \right\}^\trans \vec\Big\{ \E[ \vec(\bG_t) \vec(\bG_t)^\trans ] \Big\} ,
\end{align*}
which, together with \eqref{eqn: def_jump_factor}, gives
\begin{align}
    &\;\quad
    \frac{1}{\sqrt{2\gamma}} \Bigg\| \sum_{t=\tau+1}^{\tau+\gamma} \cII_{G,t}
    - \sum_{t=\tau-\gamma+1}^{\tau} \cII_{G,t} \Bigg\| \notag \\
    &=
    \frac{1}{\sqrt{2\gamma}} \Bigg\| (qp)^{-2} \L_{k_rk_c} \left\{ (\wh\bU^\trans \bU \otimes \wh\bQ^\trans \bQ) \otimes (\wh\bU^\trans \bU \otimes \wh\bQ^\trans \bQ) \right\}^\trans \notag \\
    &\;\quad
    \cdot \vec\Bigg( \sum_{t=\tau+1}^{\tau+\gamma} \E[ \vec(\bG_t) \vec(\bG_t)^\trans ]
    - \sum_{t=\tau-\gamma+1}^{\tau} \E[ \vec(\bG_t) \vec(\bG_t)^\trans ] \Bigg) \Bigg\| \notag \\
    &=
    \frac{1}{\sqrt{2\gamma}} \Bigg\| \Vech\Bigg\{ \Bigg( \frac{\wh\bU^\trans \bU}{q} \otimes \frac{\wh\bQ^\trans \bQ}{p} \Bigg)^\trans  \notag \\
    &\;\quad
    \cdot \Bigg( \sum_{t=\tau+1}^{\tau+\gamma} \E[ \vec(\bG_t) \vec(\bG_t)^\trans ]
    - \sum_{t=\tau-\gamma+1}^{\tau} \E[ \vec(\bG_t) \vec(\bG_t)^\trans ] \Bigg) \Bigg( \frac{\wh\bU^\trans \bU}{q} \otimes \frac{\wh\bQ^\trans \bQ}{p} \Bigg) \Bigg\} \Bigg\| \notag \\
    &=
    \begin{cases}
        \frac{\gamma - |\tau - \theta_{f,w}|}{\sqrt{2\gamma}} \Big\| \Vech\Big\{ \Big( \frac{\wh\bU^\trans \bU}{q} \otimes \frac{\wh\bQ^\trans \bQ}{p} \Big)^\trans \Delta_{G,w} \Big( \frac{\wh\bU^\trans \bU}{q} \otimes \frac{\wh\bQ^\trans \bQ}{p} \Big) \Big\} \Big\|  
        & \text{if }
        \tau \in [\theta_{f,w} - \gamma + 1, \theta_{f,w} + \gamma - 1]
        , \\
        0  & \text{if }
        \min_{w\in [Q_f]} |\tau - \theta_{f,w}| \geq \gamma
        .
    \end{cases}
    \label{eqn: cII_G_tau}
\end{align}
Note that in the last equality, we also used $2\gamma \leq \min_{w\in[Q_f +1]}\{\theta_{f,w} - \theta_{f,w-1}\}$ (from \eqref{cond:thm: consistency_change_fac:one} in the statement of the theorem)
to guarantee the scanning window contains at most one true change each time.
For any $w\in[Q_f]$, define
\begin{equation}
\label{eqn: def_tilde_Delta_Gw}
\wt\Delta_{G,w} := \left\| \Vech\left\{ \left( \frac{\wh\bU^\trans \bU}{q} \otimes \frac{\wh\bQ^\trans \bQ}{p} \right)^\trans \Delta_{G,w} \left( \frac{\wh\bU^\trans \bU}{q} \otimes \frac{\wh\bQ^\trans \bQ}{p} \right) \right\} \right\| .
\end{equation}

Then combining Weyl's inequality and the fact that $\|\Vech(\bA)\| \geq \|\bA\|_F/\sqrt{2} \geq \|\bA\|/\sqrt{2}$ for any symmetric matrix $\bA$, by Lemma~\ref{lemma: U_kron_Q_converge_I}, there exists some event $\cH_{T,p,q}$ satisfying $\P(\cH_{T,p,q}) \to 1$ as $\min\{T, p, q\} \to \infty$ such that on $\cH_{T,p,q}$:
\begin{align}
    \wt\Delta_{G,w} &\geq
    \frac{1}{\sqrt{2}} \Bigg\| \Bigg( \frac{\wh\bU^\trans \bU}{q} \otimes \frac{\wh\bQ^\trans \bQ}{p} \Bigg)^\trans \Delta_{G,w} \Bigg( \frac{\wh\bU^\trans \bU}{q} \otimes \frac{\wh\bQ^\trans \bQ}{p} \Bigg) \Bigg\| \notag \\
    &\geq
    \frac{1}{\sqrt{2}} \omega_{f,w} \cdot \sigma_{\text{min}}\Bigg\{ \Bigg( \frac{\wh\bU^\trans \bU}{q} \otimes \frac{\wh\bQ^\trans \bQ}{p} \Bigg)^\trans \Bigg( \frac{\wh\bU^\trans \bU}{q} \otimes \frac{\wh\bQ^\trans \bQ}{p} \Bigg) \Bigg\} \notag \\
    &\geq
    \frac{1}{\sqrt{2}} \omega_{f,w} \Bigg\{ \sigma_{\text{min}}(\bI_{k_rk_c}) - \Bigg\| \bI_{k_rk_c} - \frac{\bU^\trans \wh\bU \wh\bU^\trans \bU}{q^2} \otimes \frac{\bQ^\trans \wh\bQ \wh\bQ^\trans \bQ}{p^2} \Bigg\| \Bigg\}
    \geq 
    \frac{1}{2} \omega_{f,w} ,
    \label{eqn: Vech_UQUQ_lower_bound}
\end{align}
where $\sigma_{\text{min}}(\cdot)$ denotes the minimum singular of a matrix, and only the last inequality involves the set $\cH_{T,p,q}$.
By the fact that $\|\Vech(\bA)\| \leq \sqrt{n} \|\bA\|$ for any $n\times n$ symmetric matrix $\bA$, we may similarly obtain on the same set $\cH_{T,p,q}$:
\begin{align}
    \wt\Delta_{G,w} \leq
    \sqrt{k_r k_c} (2-1/\sqrt{2}) \omega_{f,w} .
    \label{eqn: Vech_UQUQ_upper_bound}
\end{align}

Combining \eqref{eqn: T_tau_gamma_bound}, \eqref{eqn: rate_cI_Gt}, and \eqref{eqn: cII_G_tau}, we have
\begin{align}
    \P\left( \max_{\substack{\tau: |\tau - \theta_{f,w}| \geq \gamma \\ w\in [Q_f] }} \cT^{(f)}_{\gamma, \tau} > c_{T,p,q}^{(f)} \sqrt{\log(T)} \right) 
    \to 0 \text{ \ as \ } \min\{T, p, q\} \to \infty ,
    \label{eqn: outside_gamma_below_thres}
\end{align}
since $c_{T,p,q}^{(f)} \to\infty$ arbitrarily slowly, thus indicating that no estimator in $\wh\Theta_f$ would be detected outside of the radius of $\gamma$ from each change point in $\Theta_f$.

With the set $\cH_{T,p,q}$ from \eqref{eqn: Vech_UQUQ_lower_bound} and \eqref{eqn: Vech_UQUQ_upper_bound}, for any $\xi\in(0,1]$, $w\in[Q_f]$, define the set
\[
\cS_{T,p,q,w} := \left\{ \cT^{(f)}_{\gamma, \theta_{f,w}} \geq \max\left( \max_{\tau: |\tau - \theta_{f,w}| > \xi \gamma} \cT^{(f)}_{\gamma, \tau}, c_{T,p,q}^{(f)} \sqrt{\log(T)} \right) \right\} \cap \cH_{T,p,q} .
\]
Further define $\cS_{T,p,q} := \bigcap_{w\in[Q_f]} \cS_{T,p,q,w}$.
Combining \eqref{eqn: T_tau_gamma_bound}, \eqref{eqn: rate_cI_Gt}, \eqref{eqn: cII_G_tau}, \eqref{eqn: Vech_UQUQ_lower_bound}, and \eqref{eqn: Vech_UQUQ_upper_bound}, we conclude that on the set $\cH_{T,p,q}$ from \eqref{eqn: Vech_UQUQ_lower_bound} and \eqref{eqn: Vech_UQUQ_upper_bound}, uniformly over $w\in[Q_f]$:
\begin{align*}
    \cT^{(f)}_{\gamma, \theta_{f,w}} & \geq
    \sqrt{\frac{\gamma}{2}} \cdot \wt\Delta_{G,w} + \cO_P\left( \sqrt{\log(T)} \right) =
    \sqrt{\frac{\gamma}{2}} \cdot \wt\Delta_{G,w}  (1 + o_P(1)) ; \\
    \max_{\substack{\tau: |\tau - \theta_{f,j}| > \xi \gamma \\ j\in [Q_f]}} \cT^{(f)}_{\gamma, \tau} 
    & \leq
    (1-\xi) \sqrt{\frac{\gamma}{2}} \cdot \wt\Delta_{G,w}
    + \cO_P\left( \sqrt{\log(T)} \right)
    < \sqrt{\frac{\gamma}{2}} \cdot \wt\Delta_{G,w} (1 + o_P(1)) ; \\
    c_{T,p,q}^{(f)} \sqrt{\log(T)}
    & \leq
    \frac{\sqrt{\gamma} \, \omega_{f,w}}{4} 
    <
    \sqrt{\frac{\gamma}{2}} \cdot \frac{\omega_{f,w}}{2}
    \leq 
    \sqrt{\frac{\gamma}{2}} \cdot \wt\Delta_{G,w} ;
\end{align*}
by the conditions on $\gamma$ in \eqref{cond:thm: consistency_change_fac:two} in the statement of this theorem.
The above implies $\P(\cS_{T,p,q}) \to 1$ as $\min\{T, p, q\} \to \infty$, which by our detection rule indicates that at least some change would be detected within the radius of $\xi \gamma$ from each change point in $\Theta_f$.
Next, we show that with probability approaching one, there would not be more than one detected change point within the radius of $\xi \gamma$ from each change point in $\Theta_f$. To this end, we will argue that the test statistics evaluated at farther time stamps are smaller than on those evaluated closer to each true change point (within the radius of $\gamma$). For each $w\in[Q_f]$, consider the set
\begin{align*}
    \wt\cS_{T,p,q,w} &:= \bigcap_{0\leq j \leq \lfloor 2/\xi \rfloor -2} \Bigg[ \left\{ \cT^{(f)}_{\gamma, \theta_{f,w} + j \xi \gamma /2} > \max_{\theta_{f,w} + (j+1) \xi \gamma /2 \leq \tau \leq \theta_{f,w} + \gamma} \cT^{(f)}_{\gamma, \tau} \right\} \\
    &\;\quad
    \cap \left\{ \cT^{(f)}_{\gamma, \theta_{f,w} - j \xi \gamma /2} > \max_{\theta_{f,w} - \gamma \leq \tau \leq \theta_{f,w} - (j+1) \xi \gamma /2} \cT^{(f)}_{\gamma, \tau} \right\} \Bigg] \cap \cH_{T,p,q} .
\end{align*}
Then with analogous arguments for $\cS_{T,p,q}$, we may conclude that $\P(\wt\cS_{T,p,q}) \to 1$ as $\min\{T, p, q\} \to \infty$, where $\wt\cS_{T,p,q} := \bigcap_{w\in [Q_f]} \wt\cS_{T,p,q,w}$, thereby 
showing that
\begin{equation}
\label{eqn: loc_fac_consistency_coarse}
\P\left( \wh{Q}_f = Q_f; \max_{w\in[\wh{Q}_f]} |\wh\theta_{f,w} - \theta_{f,w}| \leq \xi \gamma \right) \to 1 ,
\end{equation}
for any $\xi \in (0,1]$, as $\min\{T, p, q\} \to\infty$.

\paragraph{Rate of estimation.}
We refine such result in the remainder of this proof. 

Fix any $w\in [Q_f]$. On the set $\cS_{T,p,q} \cap \wt\cS_{T,p,q}$, we have $|\wh\theta_{f,w} - \theta_{f,w}| \leq \gamma$. Without loss of generality, we only show the case for $\wh\theta_{f,w} \leq \theta_{f,w}$; the proof for $\wh\theta_{f,w} \geq \theta_{f,w}$ is similar and omitted.
Define
$\wt\cT_{\gamma, \tau,w} := \big(\cT^{(f)}_{\gamma, \tau}\big)^2 - \big(\cT^{(f)}_{\gamma, \theta_{f,w}}\big)^2$ (note that $\cT^{(f)}_{\gamma, \tau} \geq 0$ for any $\tau$, yet squaring the test statistics aids the proof arguments later).
Further note that $\argmax_{\tau \in \cI} \cT^{(f)}_{\gamma, \tau} = \argmax_{\tau \in \cI} \wt\cT_{\gamma, \tau,w}$ for any interval $\cI$, thus $\argmax_{\tau \in \cI} \big(\cT^{(f)}_{\gamma, \tau}\big)^2 = \argmax_{\tau \in \cI} \wt\cT_{\gamma, \tau,w}^2$.
Hence with $\kappa_T$ from Lemma~\ref{lemma: max_ZZ_2+epsilon}, for some fixed constant $c_0\in (0, \infty)$, we have
\begin{align}
    &\;\quad
    \left\{ c_0 \, \omega_{f,w}^2 (\wh\theta_{f,w} - \theta_{f,w}) \leq - \kappa_T^2 \right\} \notag \\
    &\subseteq
    \left\{ \max_{\theta_{f,w} - \gamma + 1 \leq \tau \leq \theta_{f,w} - c_0^{-1} \, \omega_{f,w}^{-2} \kappa_T^2} \cT^{(f)}_{\gamma, \tau} \geq \max_{\theta_{f,w} - c_0^{-1} \, \omega_{f,w}^{-2} \kappa_T^2 +1 \leq \tau \leq \theta_{f,w} + \gamma} \cT^{(f)}_{\gamma, \tau} \right\} \notag \\
    &=
    \left\{ \max_{\theta_{f,w} - \gamma + 1 \leq \tau \leq \theta_{f,w} - c_0^{-1} \, \omega_{f,w}^{-2} \kappa_T^2} \wt\cT_{\gamma, \tau,w} \geq \max_{\theta_{f,w} - c_0^{-1} \, \omega_{f,w}^{-2} \kappa_T^2 +1 \leq \tau \leq \theta_{f,w} + \gamma} \wt\cT_{\gamma, \tau,w} \right\} \notag \\
    &\subseteq
    \left\{ \max_{\theta_{f,w} - \gamma + 1 \leq \tau \leq \theta_{f,w} - c_0^{-1} \, \omega_{f,w}^{-2} \kappa_T^2} \wt\cT_{\gamma, \tau,w} \geq 0 \right\}  
    , \label{eqn: fac_cp_subset_implication}
\end{align}
where the last line used the fact that the interval $\big[ \theta_{f,w} - c_0^{-1} \, \omega_{f,w}^{-2} \kappa_T^2 +1 , \theta_{f,w} + \gamma \big]$ covers $\theta_{f,w}$ at which $\wt\cT_{\gamma,\theta_{f,w},w} =0$.

Next, for any $t\in[\gamma, T-\gamma]$, define
\begin{align*}
    \wt\cI_{G,t} &:= \frac{1}{\sqrt{2\gamma}} \left( \sum_{t=\tau+1}^{\tau+\gamma} \cI_{G,t}
    - \sum_{t=\tau-\gamma+1}^{\tau} \cI_{G,t} \right) , \quad
    \wt\cII_{G,t} := \frac{1}{\sqrt{2\gamma}} \left( \sum_{t=\tau+1}^{\tau+\gamma} \cII_{G,t}
    - \sum_{t=\tau-\gamma+1}^{\tau} \cII_{G,t} \right) .
\end{align*}
where we have suppressed the dependency on $\gamma$ for the ease of notation.
Then using \eqref{eqn: vech_GG_decomp} and the expansion that $\|\bu\|^2 - \|\bv\|^2 = (\bu+\bv)^\trans (\bu-\bv)$ for any vectors $\bu, \bv$, consider the decomposition
\begin{align*}
    \wt\cT_{\gamma, \tau,w} &=
    (\cT^{(f)}_{\gamma, \tau} + \cT^{(f)}_{\gamma, \theta_{f,w}}) (\cT^{(f)}_{\gamma, \tau} - \cT^{(f)}_{\gamma, \theta_{f,w}}) \\
    &=
    ( \wt\cI_{G,\tau} + \wt\cI_{G,\theta_{f,w}} + \wt\cII_{G,\tau} + \wt\cII_{G,\theta_{f,w}} )^\trans
    ( \wt\cI_{G,\tau} - \wt\cI_{G,\theta_{f,w}} + \wt\cII_{G,\tau} - \wt\cII_{G,\theta_{f,w}} ) \\
    &=
    (\wt\cII_{G,\tau} + \wt\cII_{G,\theta_{f,w}})^\trans (\wt\cII_{G,\tau} - \wt\cII_{G,\theta_{f,w}}) 
    + (\wt\cI_{G,\tau} - \wt\cI_{G,\theta_{f,w}})^\trans (\wt\cII_{G,\tau} + \wt\cII_{G,\theta_{f,w}}) \\
    &\;\quad
    +(\wt\cII_{G,\tau} - \wt\cII_{G,\theta_{f,w}})^\trans
    (\wt\cI_{G,\tau} + \wt\cI_{G,\theta_{f,w}})
    + (\wt\cI_{G,\tau} - \wt\cI_{G,\theta_{f,w}})^\trans (\wt\cI_{G,\tau} + \wt\cI_{G,\theta_{f,w}}) \\
    &=: 
    \wt\cT_{\gamma, \tau,w,1} + \wt\cT_{\gamma, \tau,w,2} + \wt\cT_{\gamma, \tau,w,3} + \wt\cT_{\gamma, \tau,w,4} .
\end{align*}
Given $\tau \in [\theta_{f,w} - \gamma + 1, \theta_{f,w} - c_0^{-1} \, \omega_{f,w}^{-2} \kappa_T^2]$, note that $\wt\cT_{\gamma, \tau,w,1} = \|\wt\cII_{G,\tau}\|^2 - \|\wt\cII_{G,\theta_{f,w}}\|^2 < 0$ by \eqref{eqn: cII_G_tau}. 
With $\wt\Delta_{G,w}$ defined in \eqref{eqn: def_tilde_Delta_Gw}, we further have
\begin{align*}
    & |\wt\cT_{\gamma, \tau,w,1}| =
    \frac{2\gamma - |\tau - \theta_{f,w}|}{\sqrt{2\gamma}}
    \cdot \frac{|\tau - \theta_{f,w}|}{\sqrt{2\gamma}}
    \cdot \wt\Delta_{G,w} ^2 
    \geq \frac{|\tau - \theta_{f,w}|}{2} \cdot \wt\Delta_{G,w}^2 ; \\
    & \| \wt\cII_{G,\tau} - \wt\cII_{G,\theta_{f,w}} \| =
    \frac{|\tau - \theta_{f,w}|}{\sqrt{2\gamma}} \cdot \wt\Delta_{G,w} ; \\
    &\| \wt\cII_{G,\tau} + \wt\cII_{G,\theta_{f,w}} \| \leq
    \sqrt{2\gamma} \cdot \wt\Delta_{G,w} . 
\end{align*}
Then by similar arguments in the proof of Theorem~3.2 in \cite{EichingerKirch2018}, using Cauchy--Schwarz inequality, we have
\begin{align}
    &\;\quad
    \P\left( \max_{\theta_{f,w} - \gamma + 1 \leq \tau \leq \theta_{f,w} - c_0^{-1} \, \omega_{f,w}^{-2} \kappa_T^2} \wt\cT_{\gamma, \tau,w} \geq 0, \cS_{T,p,q} \cap \wt\cS_{T,p,q} \right) \notag \\
    &=
    \P\left( \max_{\theta_{f,w} - \gamma + 1 \leq \tau \leq \theta_{f,w} - c_0^{-1} \, \omega_{f,w}^{-2} \kappa_T^2} 
    \left| \frac{\wt\cT_{\gamma, \tau,w,2}}{\wt\cT_{\gamma, \tau,w,1}} + \frac{\wt\cT_{\gamma, \tau,w,3}}{\wt\cT_{\gamma, \tau,w,1}} + \frac{\wt\cT_{\gamma, \tau,w,4}}{\wt\cT_{\gamma, \tau,w,1}} \right|
    \geq 1, \cS_{T,p,q} \cap \wt\cS_{T,p,q} \right) \notag \\
    &\leq
    \P\left( \max_{\theta_{f,w} - \gamma + 1 \leq \tau \leq \theta_{f,w} - c_0^{-1} \, \omega_{f,w}^{-2} \kappa_T^2} 
    \frac{|\wt\cT_{\gamma, \tau,w,2}|}{|\wt\cT_{\gamma, \tau,w,1}|}
    \geq \frac{1}{3}, \cS_{T,p,q} \cap \wt\cS_{T,p,q} \right) \notag \\
    &\;\quad
    + \P\left( \max_{\theta_{f,w} - \gamma + 1 \leq \tau \leq \theta_{f,w} - c_0^{-1} \, \omega_{f,w}^{-2} \kappa_T^2} 
    \frac{|\wt\cT_{\gamma, \tau,w,3}|}{|\wt\cT_{\gamma, \tau,w,1}|}
    \geq \frac{1}{3}, \cS_{T,p,q} \cap \wt\cS_{T,p,q} \right) \notag \\
    &\;\quad
    + \P\left( \max_{\theta_{f,w} - \gamma + 1 \leq \tau \leq \theta_{f,w} - c_0^{-1} \, \omega_{f,w}^{-2} \kappa_T^2} 
    \frac{|\wt\cT_{\gamma, \tau,w,4}|}{|\wt\cT_{\gamma, \tau,w,1}|}
    \geq \frac{1}{3}, \cS_{T,p,q} \cap \wt\cS_{T,p,q} \right) \notag \\
    &\leq
    \P\left( \max_{\theta_{f,w} - \gamma + 1 \leq \tau \leq \theta_{f,w} - c_0^{-1} \, \omega_{f,w}^{-2} \kappa_T^2} 
    \frac{\sqrt{2\gamma} \|\wt\cI_{G,\tau} - \wt\cI_{G,\theta_{f,w}}\|}{|\tau - \theta_{f,w}| \wt\Delta_{G,w}}
    \geq \frac{1}{6}, \cS_{T,p,q} \cap \wt\cS_{T,p,q} \right) \notag \\
    &\;\quad
    + \P\left( \max_{\theta_{f,w} - \gamma + 1 \leq \tau \leq \theta_{f,w} - c_0^{-1} \, \omega_{f,w}^{-2} \kappa_T^2} 
    \frac{\|\wt\cI_{G,\tau} + \wt\cI_{G,\theta_{f,w}}\|}{\wt\Delta_{G,w} \sqrt{2\gamma}}
    \geq \frac{1}{6}, \cS_{T,p,q} \cap \wt\cS_{T,p,q} \right) \notag \\
    &\;\quad
    + \P\left( \max_{\theta_{f,w} - \gamma + 1 \leq \tau \leq \theta_{f,w} - c_0^{-1} \, \omega_{f,w}^{-2} \kappa_T^2} 
    \frac{|(\wt\cI_{G,\tau} - \wt\cI_{G,\theta_{f,w}})^\trans (\wt\cI_{G,\tau} + \wt\cI_{G,\theta_{f,w}})|}{|\tau - \theta_{f,w}| \wt\Delta_{G,w}^2}
    \geq \frac{1}{6}, \cS_{T,p,q} \cap \wt\cS_{T,p,q} \right) \notag \\
    &\leq
    2 \P\left( \max_{\theta_{f,w} - \gamma + 1 \leq \tau \leq \theta_{f,w} - c_0^{-1} \, \omega_{f,w}^{-2} \kappa_T^2} 
    \frac{\sqrt{2\gamma} \|\wt\cI_{G,\tau} - \wt\cI_{G,\theta_{f,w}}\|}{|\tau - \theta_{f,w}| \wt\Delta_{G,w}}
    \geq \frac{1}{6}, \cS_{T,p,q} \cap \wt\cS_{T,p,q} \right) \notag \\
    &\;\quad
    + 2 \P\left( \max_{\theta_{f,w} - \gamma + 1 \leq \tau \leq \theta_{f,w} - c_0^{-1} \, \omega_{f,w}^{-2} \kappa_T^2} 
    \frac{\|\wt\cI_{G,\tau} + \wt\cI_{G,\theta_{f,w}}\|}{\wt\Delta_{G,w} \sqrt{2\gamma}}
    \geq \frac{1}{6}, \cS_{T,p,q} \cap \wt\cS_{T,p,q} \right) 
    , \label{eqn: refine_tilde_T_step}
\end{align}
where the last inequality used $\P(ab>cd) \leq \P(\{a>c\} \cup \{b>d\})$ for any non-negative $a,b,c,d$.

By direct computations (cf.\ Equation (5.6) in \cite{EichingerKirch2018}), for any $w\in[Q_f]$, $0< c < \gamma$, $\theta_{f,w} - \gamma \leq \tau \leq \theta_{f,w} - c$, we have
\begin{equation}
\label{eqn: refine_noise_mosum_diff}
\wt\cI_{G,\tau} - \wt\cI_{G,\theta_{f,w}}
= \frac{1}{\sqrt{2\gamma}} \left( \sum_{t = \tau + \gamma +1}^{\theta_{f,w} + \gamma} \cI_{G,t} + \sum_{t = \tau - \gamma +1}^{\theta_{f,w} - \gamma} \cI_{G,t} - 2 \sum_{t = \tau +1}^{\theta_{f,w}} \cI_{G,t}  \right) .
\end{equation}

From \eqref{eqn: rate_cI_Gt} and \eqref{eqn: Vech_UQUQ_lower_bound}, on the set $\cS_{T,p,q} \cap \wt\cS_{T,p,q}$ we have
\begin{align*}
    \max_{\theta_{f,w} - \gamma + 1 \leq \tau \leq \theta_{f,w} - c_0^{-1} \, \omega_{f,w}^{-2} \kappa_T^2} \frac{\|\wt\cI_{G,\tau} + \wt\cI_{G,\theta_{f,w}}\|}{\wt\Delta_{G,w} \sqrt{2\gamma}}
    \leq 
    \max_{\gamma \leq \tau \leq T - \gamma} 
    \frac{2\|\wt\cI_{G,\tau}\|}{\wt\Delta_{G,w} \sqrt{2\gamma}}
    =
    \cO_P\left( \frac{2\sqrt{\log(T)}}{\omega_{f,w} \sqrt{2\gamma}} \right)
    = o_P(1) ,
\end{align*}
where the last equality used the rate condition in \eqref{cond:thm: consistency_change_fac:two} from the statement of Theorem~\ref{thm: consistency_change_fac}.
In addition, from \eqref{eqn: Vech_UQUQ_lower_bound}, \eqref{eqn: refine_noise_mosum_diff}, and Lemma~\ref{lemma: max_ZZ_2+epsilon}, we have as $\min\{T, p, q\} \to \infty$:
\begin{align*}
    &\;\quad
    \P\left( \max_{\theta_{f,w} - \gamma + 1 \leq \tau \leq \theta_{f,w} - c_0^{-1} \, \omega_{f,w}^{-2} \kappa_T^2} 
    \frac{\sqrt{2\gamma} \|\wt\cI_{G,\tau} - \wt\cI_{G,\theta_{f,w}}\|}{|\tau - \theta_{f,w}| \wt\Delta_{G,w}}
    \geq \frac{1}{6}, \cS_{T,p,q} \cap \wt\cS_{T,p,q} \right) \\
    &\leq
    \P\left( \max_{\theta_{f,w} - \gamma + 1 \leq \tau \leq \theta_{f,w} - c_0^{-1} \, \omega_{f,w}^{-2} \kappa_T^2} 
    \frac{\sqrt{2\gamma} \|\wt\cI_{G,\tau} - \wt\cI_{G,\theta_{f,w}}\|}{|\tau - \theta_{f,w}| \cdot \omega_{f,w}}
    \geq \frac{1}{12} \right) \\
    &\leq
    \P\left( \max_{\theta_{f,w} - \gamma + 1 \leq \tau \leq \theta_{f,w} - c_0^{-1} \, \omega_{f,w}^{-2} \kappa_T^2} 
    \frac{1}{|\tau - \theta_{f,w}| \cdot \omega_{f,w}} \left\| \sum_{t = \tau + \gamma +1}^{\theta_{f,w} + \gamma} \cI_{G,t} \right\|
    \geq \frac{1}{36} \right) \\
    &\;\quad
    + \P\left( \max_{\theta_{f,w} - \gamma + 1 \leq \tau \leq \theta_{f,w} - c_0^{-1} \, \omega_{f,w}^{-2} \kappa_T^2} 
    \frac{1}{|\tau - \theta_{f,w}| \cdot \omega_{f,w}} \left\| \sum_{t = \tau - \gamma +1}^{\theta_{f,w} - \gamma} \cI_{G,t} \right\|
    \geq \frac{1}{36} \right) \\
    &\;\quad
    + \P\left( \max_{\theta_{f,w} - \gamma + 1 \leq \tau \leq \theta_{f,w} - c_0^{-1} \, \omega_{f,w}^{-2} \kappa_T^2} 
    \frac{1}{|\tau - \theta_{f,w}| \cdot \omega_{f,w}} \left\| \sum_{t = \tau +1}^{\theta_{f,w}} \cI_{G,t} \right\|
    \geq \frac{1}{72} \right) \\
    &=
    \P\left( \max_{c_0^{-1}\, \omega_{f,w}^{-2} \kappa_T^2 \leq \ell \leq \theta_{f,w} - \theta_{f,w - 1}} 
    \frac{\sqrt{c_0^{-1}\, \omega_{f,w}^{-2} \kappa_T^2}}{\ell} \left\| \sum_{t = \theta_{f,w} - \ell + \gamma +1}^{\theta_{f,w} + \gamma} \cI_{G,t} \right\|
    \geq \frac{c_0^{-1/2}\,\kappa_T}{36\, c_0} \right) \\
    &\;\quad
    + \P\left( \max_{c_0^{-1}\, \omega_{f,w}^{-2} \kappa_T^2 \leq \ell \leq \theta_{f,w} - \theta_{f,w - 1}} 
    \frac{\sqrt{c_0^{-1}\, \omega_{f,w}^{-2} \kappa_T^2}}{\ell} \left\| \sum_{t = \theta_{f,w} - \ell - \gamma +1}^{\theta_{f,w} - \gamma} \cI_{G,t} \right\|
    \geq \frac{c_0^{-1/2}\,\kappa_T}{36\, c_0} \right) \\
    &\;\quad
    + \P\left( \max_{c_0^{-1}\, \omega_{f,w}^{-2} \kappa_T^2 \leq \ell \leq \theta_{f,w} - \theta_{f,w - 1}} 
    \frac{\sqrt{c_0^{-1}\, \omega_{f,w}^{-2} \kappa_T^2}}{\ell} \left\| \sum_{t = \theta_{f,w} - \ell +1}^{\theta_{f,w}} \cI_{G,t} \right\|
    \geq \frac{c_0^{-1/2}\,\kappa_T}{72\, c_0} \right)
    \to 0 ,
\end{align*}
where the last step holds for small enough $c_0$ (e.g.,\ $c_0 = 1/72$). 
Altogether, we have shown the right-hand-side of \eqref{eqn: refine_tilde_T_step} is bounded as $o(1)$. Further by \eqref{eqn: fac_cp_subset_implication}, we may conclude the proof of this theorem by setting $\kappa_{T,p,q}^{(f)} = c_0^{-1} \kappa_T^2$.
\end{proof}

\clearpage

\section{Additional discussion on the model}
\label{app: model}

We make additional comments on the specific structure in \eqref{eqn: mean_change}--\eqref{eqn: cov_change} here.
Suppose $\{\bX_t\}$ admits a conventional matrix factor model representation such that for each $t\in[T]$, $\bX_t = \Acute\bC_t + \Acute\bE_t$, where $\Acute\bC_t = \Acute\bA \Acute\bF_t \Acute\bB^\trans$, $\Acute\bF_t$ is the core factor, $\Acute\bA$ and $\Acute\bB$ are the row and column loading matrices with no restrictions on their spanning spaces. We may also assume $\Acute\bF_t$ can be decomposed as a stationary part $\bE_{F,t}$ and a non-stationary part with trend function $\bmu_{F,t}$ and $Q_F$ changes at $\{\theta_{F,i}\}_{i\in[Q_F]}$, \text{viz.}, $\Acute\bF_t = \bmu_{F,t} + \sum_{i=1}^{Q_F} \bgamma_{F,i} \cdot \mathbb{I}\{t \geq \theta_{F,i} \} + \bE_{F,t}$. With the matrix $\bM_m$ introduced at the end of Section~\ref{sec: introduction}, we may rewrite
\begin{align*}
    \bX_t &= \{(pq)^{-1} \1_p^\trans \Acute\bC_t \1_q\} \1_p \1_q^\trans + (q^{-1} \bM_p \Acute\bC_t \1_q) \1_q^\trans + \1_p (p^{-1} \bM_q \Acute\bC_t^\trans \1_p )^\trans + (\bM_p \Acute\bA) \Acute\bF_t (\bM_q \Acute\bB)^\trans + \Acute\bE_t ,
\end{align*}
thereby reading $\bX_t$ in model \eqref{eqn: MEFM_structure} and the structure in \eqref{eqn: mean_change}--\eqref{eqn: cov_change} with $Q_f=0$, $Q_\alpha = Q_\beta = Q_F$, $\theta_{\alpha,i} = \theta_{\beta,i} = \theta_{F,i}$ for $i\in[Q_F]$,
\begin{equation}
\label{eqn: FM_as_MEFM}
\begin{split}
    & \bmu_t = (pq)^{-1} \1_p^\trans \Acute\bC_t \1_q ,\quad
    \bF_t = \Acute\bF_t, \quad
    \bA_0 = \bM_p \Acute\bA, \quad
    \bB_0 = \bM_q \Acute\bB , \quad
    \bE_t = \Acute\bE_t , \\
    & \bmu_{\alpha,t} = q^{-1} \bM_p \Acute\bA \bmu_{F,t} \Acute\bB^\trans \1_q , \quad
    \bgamma_{\alpha,i} = q^{-1} \bM_p \Acute\bA \bgamma_{F,i} \Acute\bB^\trans \1_q , \quad
    \be_{\alpha,t} = q^{-1} \bM_p \Acute\bA \bE_{F,t} \Acute\bB^\trans \1_q , \\
    & \bmu_{\beta,t} = p^{-1} \bM_q \Acute\bB \bmu_{F,t}^\trans \Acute\bA^\trans \1_p , \quad
    \bgamma_{\beta,i} = p^{-1} \bM_q \Acute\bB \bgamma_{F,i}^\trans \Acute\bA^\trans \1_p , \quad
    \be_{\beta,t} = p^{-1} \bM_q \Acute\bB \bE_{F,t}^\trans \Acute\bA^\trans \1_p ,    
\end{split}
\end{equation}
except that we will consider mean-zero stationary factor $\bF_t$ throughout this article for convenience. Although it is possible to relax this at the expenses of more complicated arguments, the above shows that our setup is considerably general. For instance, any induced jump vector $\bgamma_{\alpha,i}$ above has such a particular structure which is not required in our formulation in \eqref{eqn: mean_change}.

We further elaborate our considered framework by the following example.

\begin{example}
Suppose that Assumption~\ref{ass: identification} is not necessarily met by $\bA_w$ or $\bB_w$, i.e.,\ $\1_p^\trans \bA_w \ne \0$ or $\1_q^\trans \bB_w \ne \0$.
Recalling the definition of $\bM_p$, we can write the model~\eqref{eqn: MEFM_structure} as
\begin{align*}
\bX_t &= \Acute\mu_t \1_p\1_q^\trans + \Acute\balpha_t \1_q^\trans + \1_p \Acute\bbeta_t + \sum_{w = 0}^{Q_f} \Acute\bA_w \bF_t \Acute\bB_w^\trans \cdot \mathbb{I}\{\theta_{f, w} < t \le \theta_{f, w + 1} \}, \quad \text{where}
\\
\Acute\mu_t &= \mu_t + \frac{1}{pq} \sum_{w = 0}^{Q_f} \1_p^\trans \bA_w \bF_t \bB_w^\trans \1_q \cdot \mathbb{I}\{\theta_{f, w} < t \le \theta_{f, w + 1} \}, \\
\Acute\balpha_t &= \balpha_t + \frac{1}{q}\sum_{w = 0}^{Q_f} \bM_p \bA_w \bF_t \bB_w^\trans \1_q \cdot \mathbb{I}\{\theta_{f, w} < t \le \theta_{f, w + 1} \},
\end{align*}
and $\Acute\bA_w = \bM_p \bA_w$, and $\Acute\bbeta_t$ and $\Acute\bB_w$ are analogously defined.
This implies that the part of (possibly) time-varying factor-driven component that is not orthogonal to the unit vectors, may get absorbed into the stochastic errors of the grand mean and the main effects, which justifies Assumption~\ref{ass: identification}. 
Theoretically, we can permit changes in the covariance of the main effects as we take the max norm to detect changes in the main effects, and the asymptotic separation of the signal and the noise will continue to be met. 
This is feasible at the cost of more complicated notations,
and hence we choose to assume trend stationarity and time-invarying covariance on the main effects for simplicity.
On a related note, we can similarly handle the non-zero mean of $\bF_t$ or any shifts therein, by incorporating them into the grand mean and the main effects.
\end{example}


\clearpage 

\section{Selection of the thresholds}
\label{subsec: thres_pi_numerical}

We simulate model~\eqref{eqn: MEFM_structure} under the null model with no change points, i.e.,\ $Q_\alpha=Q_\beta=Q_f =0$,
with all model parameters generated as in Section~\ref{subsec: simulation},
except that for flexibility and to account for more general scenarios in practice,
the elements in $\mu_t$, $\be_{\alpha,t}$, $\be_{\beta,t}$ and $\bF_t$ are independent standardized AR(3) with AR coefficients (0.5, 0.3, $-0.4$), and the elements in $\bF_{e,t}$ and $\bE_{e,t}$ are also independent standardized AR(3) with AR coefficients ($-0.5, -0.3, -0.4$) and (0.4, 0.4, 0.2), respectively.
All AR processes take i.i.d.\ $\cN(0,1)$ innovation.

We fix $k_r=k_c=2$ throughout;
local weights are used as described in Section~\ref{subsec: tuning};
the trend functions in both main effects are sinusoidal as in Section~\ref{sec:setting};
the dimensions are varied as
\begin{align*}
T&\in\{400, 800, 1200, 1600, 2000\},\\
(p, q)&\in\{(20, 20), (20, 80), (80, 80), (80, 160)\} ;
\end{align*}
and we generate $100$ realizations for each combination.
For each realization,
we take the maximum detector statistic over all intervals scanned through (i.e.,\ $\max_{\tau\in [\gamma + 1, T- \gamma)} \cT^*_{\tau}$ using notations in Algorithm~\ref{alg: detect}, where $\gamma$ is set in Section~\ref{subsec: tuning})
and record the $0.9$-empirical quantiles of these maxima for each setting,
all collected in the \verb+R+ object \verb+y+.

Consider the statistics for the row main effects.
As suggested by Theorem~\ref{thm: consistency_change_ME}, we regard the threshold $\pi^*$ as $\sqrt{\log(pT)/\gamma}$ multiplied by $\sqrt{\log(p)}$ 
(motivated by that the corresponding statistic in \eqref{eqn: mosum_main_effect} is the scaled max norm of a vector in $\R^p$), 
such that using the \verb+R+ function \verb+lm()+ (with \verb+p+ denoting $p$, \verb+T+ denoting $T$, \verb+gamma+ denoting $\gamma$, and \verb+-1+ denoting no intercepts), we regress
\begin{verbatim}
lm(y / sqrt(log(pT)/gamma) ~ sqrt(log(p)) -1) 
\end{verbatim}
with adjusted R-squared of the fitted model $R^2_{\text{adj}} = 0.9855$.
Similar regression is performed on the MOSUM statistics for the column main effects with $p$ above replaced by $q$, resulting in the corresponding adjusted R-squared $R^2_{\text{adj}} = 0.9846$.

For the second-order structure, we rely on Theorem~\ref{thm: consistency_change_fac} and let the threshold $\pi^*$ be $\sqrt{\log(T)}$ multiplied by the same functions of $T$ 
that are motivated by the approximation of the asymptotic null distribution in \citet{cho2026multivariate}, 
so that we regress
\begin{verbatim}
lm(y / sqrt(log(T)) ~ T + log(log(T))/sqrt(log(T)) -1) 
\end{verbatim}
and the adjusted R-squared of the fitted model is $R^2_{\text{adj}} = 0.9982$.

In all three fitted models,
the covariates (those functions of $T$, $p$, and $q$) are statistically significant with p-values less than $0.001$.
Furthermore,
Figure~\ref{fig: thres_reg} displays the Q-Q plots for these fitted models,
where the black tilted line denotes a comparison line by robust regression on the shown residuals; see details in the \verb+R+ function \verb+car::qqPlot()+. They all reflect that the regression models are good fits to the data,
and it is thus reasonable to use the fitted threshold.

Note that our theoretical results and thus the above rationale are built upon large dimensions.
In some practical scenarios, the sample size $T$ is not large (e.g.,\ $T< 500$), and the fitted threshold from our regression models may still be used after scaling by some reasonable constant 
$C_{\text{thres}}>0$
such as $C_{\text{thres}} = 0.7$.
On the other hand, we may set $C_{\text{thres}} = 1.4$ for heavy-tailed data such as financial returns.
In our experiments, the innovation process follows $t_6$ distribution and we set $C_{\text{thres}} = 1.4$; and $C_{\text{thres}} = 1$ is used for Gaussian innovation process.

\begin{figure}[h!t!b!]
\centering
\includegraphics[width=0.6\linewidth]{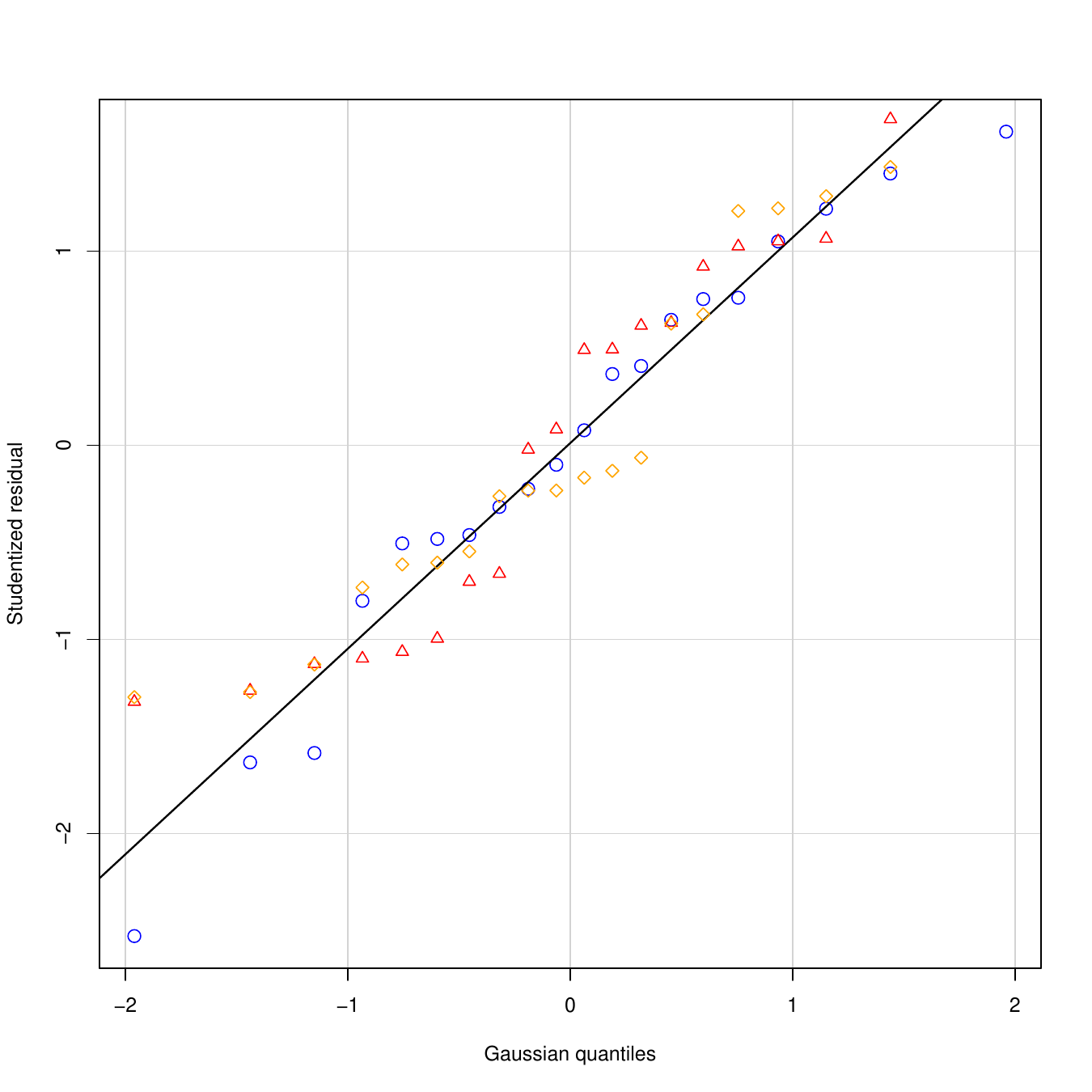}
\caption{Normal Q-Q plot for the model regressing the scaled empirical 0.9-quantiles of the maximal detector statistics on functions of $T$.
Vertical axis: standardized residuals in each fitted model; 
horizontal axis: theoretical quantiles of standard Gaussian.
Row main effects: triangle, red;
column main effects: square, orange;
factor structure: circle, blue.
}
\label{fig: thres_reg}
\end{figure}

\clearpage

\section{Additional simulation results}\label{appendix: numerical}







We report additional results obtained under~(S2) in Section~\ref{subsec:single}, 
with varying $T\in \{600,900\}$ and $(p,q) \in \{(10,10), (15,15)\}$.
The histograms are shown in Figures~\ref{fig: sim_S2_600_10}, \ref{fig: sim_S2_600_15}, \ref{fig: sim_S2_900_10}, and \ref{fig: sim_S2_900_15}, respectively.

\begin{figure}[!t]
\centering
\includegraphics[width=0.8\linewidth]{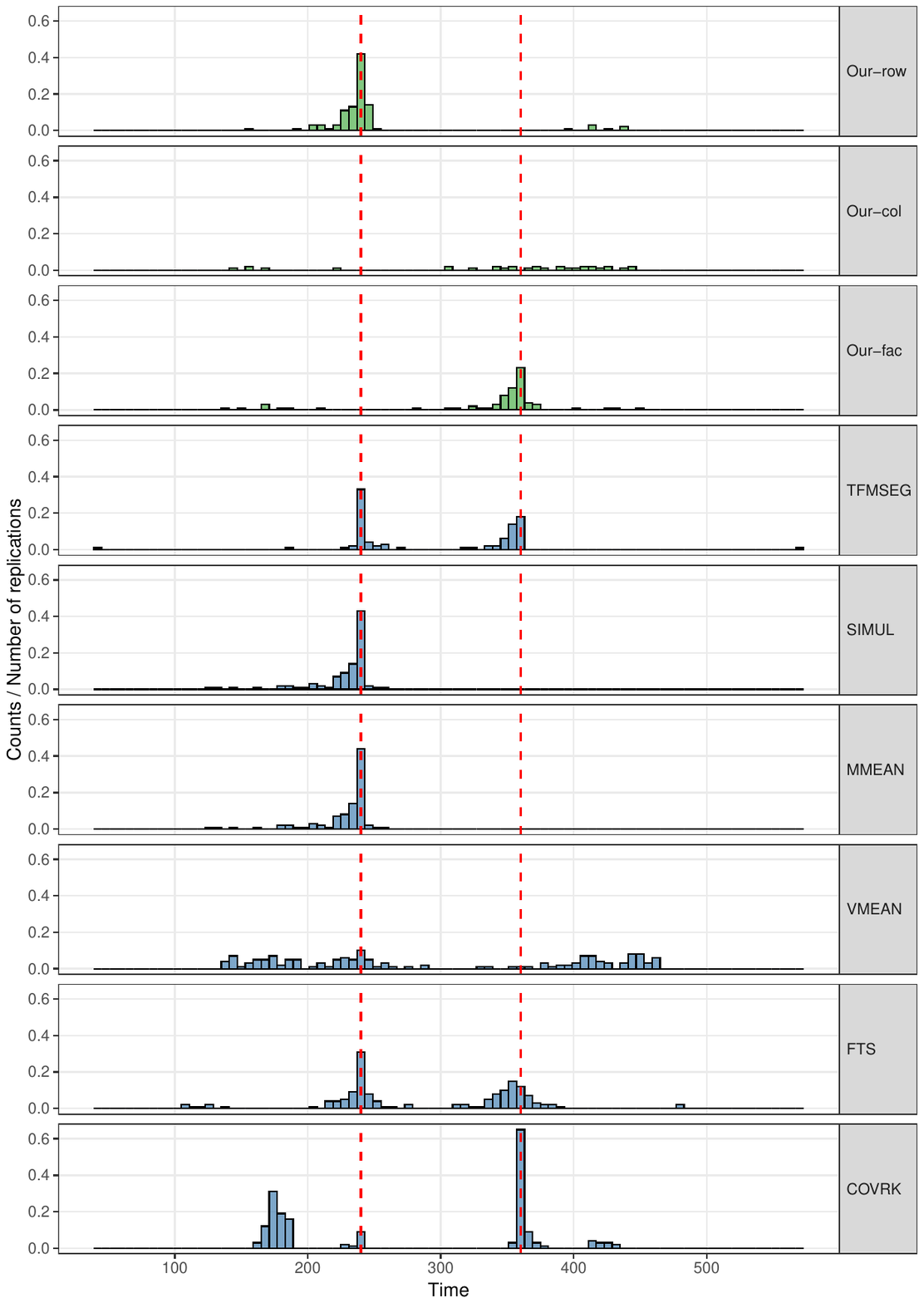}
\caption{
Histograms of detected change points for each method under Setting~(S2) but under $(T,p,q)=(600,10,10)$.
Our methods are in green;
the true change locations are denoted as dashed vertical lines.
}
\label{fig: sim_S2_600_10}
\end{figure}

\begin{figure}[!t]
\centering
\includegraphics[width=0.8\linewidth]{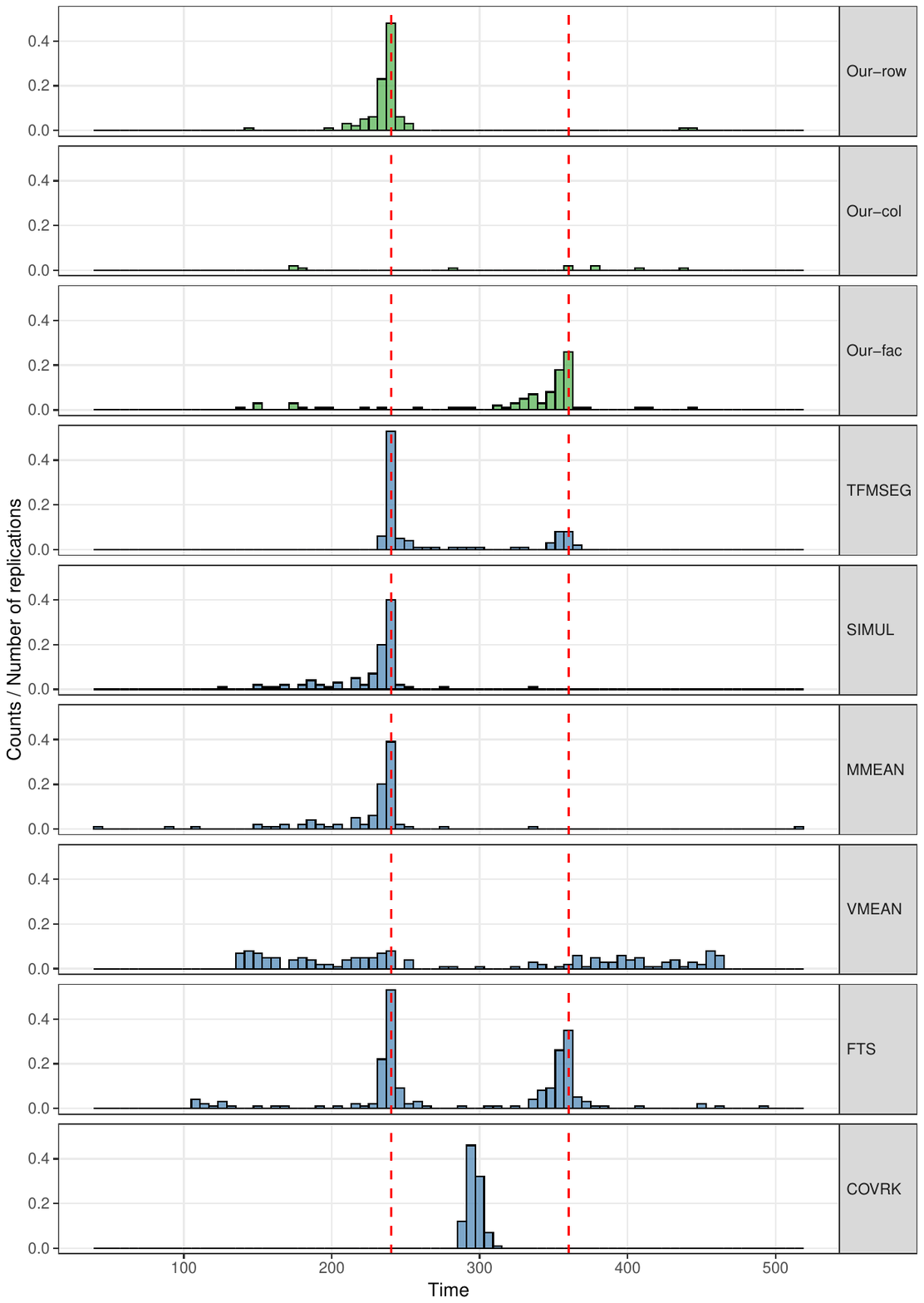}
\caption{
Histograms of detected change points for each method under Setting~(S2) but under $(T,p,q)=(600,15,15)$.
Our methods are in green;
the true change locations are denoted as dashed vertical lines.
}
\label{fig: sim_S2_600_15}
\end{figure}

\begin{figure}[!t]
\centering
\includegraphics[width=0.8\linewidth]{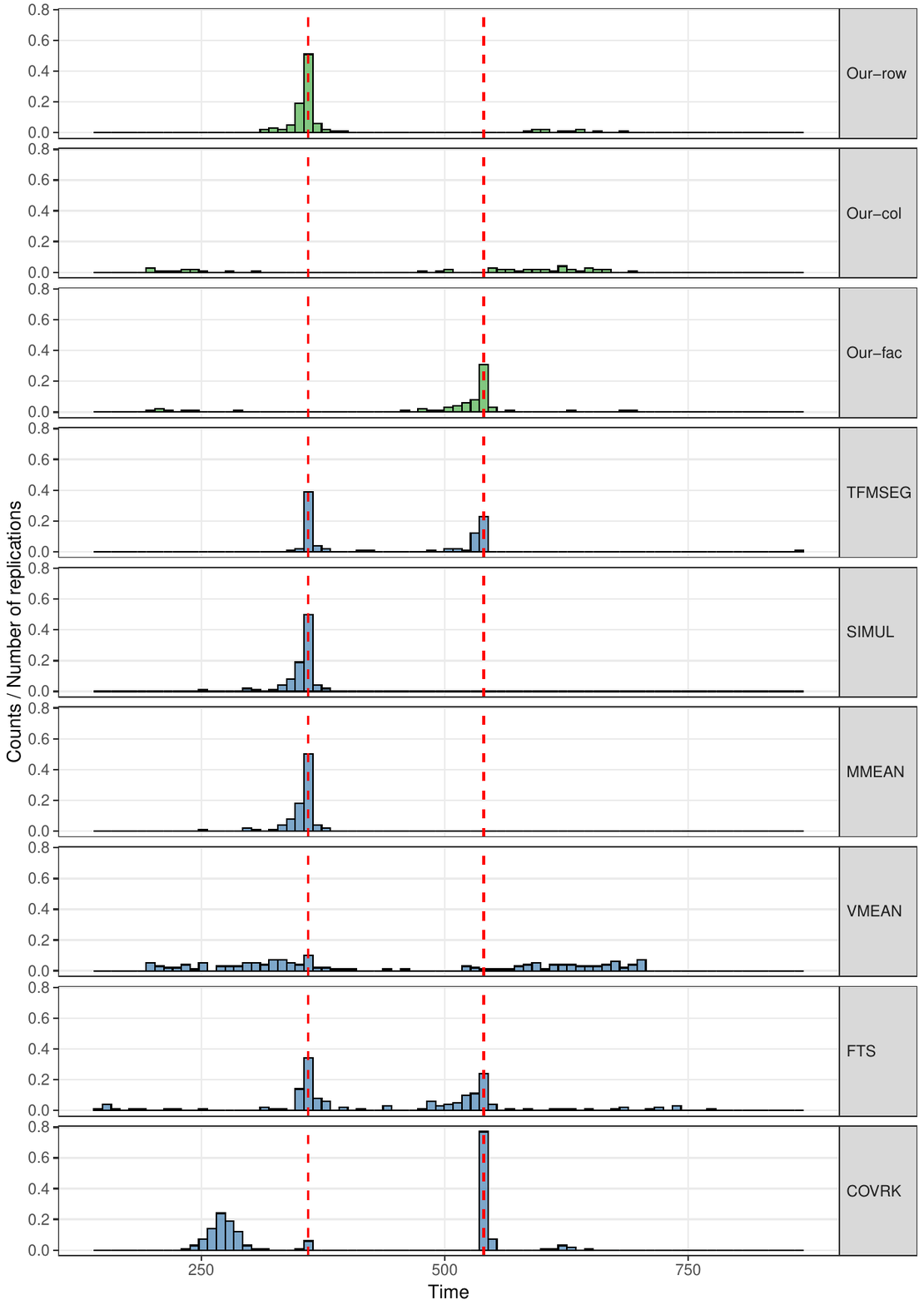}
\caption{
Histograms of detected change points for each method under Setting~(S2) but under $(T,p,q)=(900,10,10)$.
Our methods are in green;
the true change locations are denoted as dashed vertical lines.
}
\label{fig: sim_S2_900_10}
\end{figure}

\begin{figure}[!t]
\centering
\includegraphics[width=0.8\linewidth]{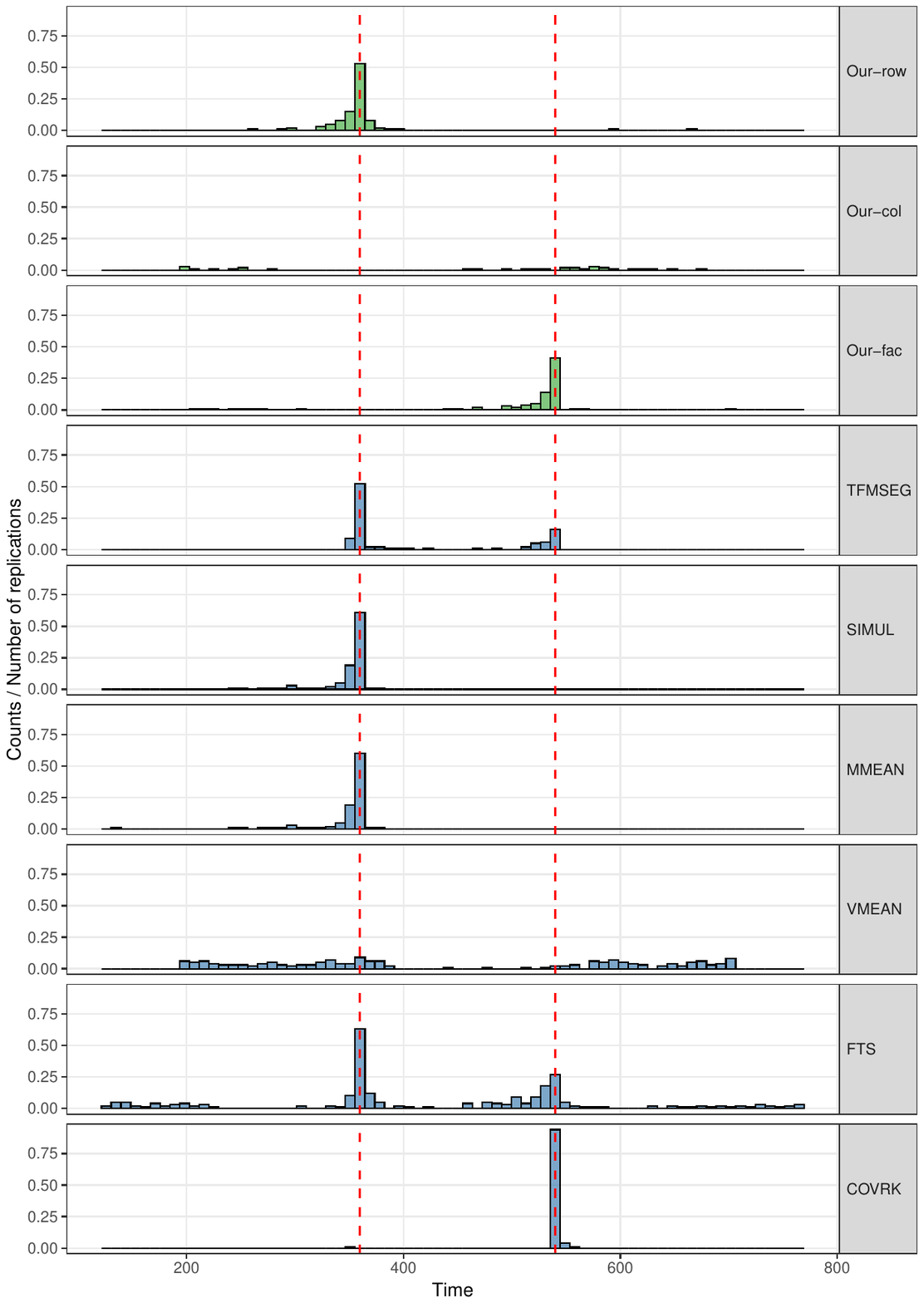}
\caption{
Histograms of detected change points for each method under Setting~(S2) but under $(T,p,q)=(900,15,15)$.
Our methods are in green;
the true change locations are denoted as dashed vertical lines.
}
\label{fig: sim_S2_900_15}
\end{figure}

\end{document}